\pdfoutput=1
\documentclass[letterpaper, USenglish, cleveref, autoref, thm-restate, numberwithinsect]{lipics-v2021}

\hideLIPIcs
\nolinenumbers

\title{The Trichotomy of Regular Property Testing}

\author{Gabriel {Bathie}}{LaBRI, Université de Bordeaux \and DIENS, Paris, France \and \url{https://perso.ens-lyon.fr/gabriel.bathie}}{gabriel.bathie@labri.fr}{}{Partially funded by the grant ANR-20-CE48-0001 from the French National Research Agency.}

\author{Nathanaël Fijalkow}{LaBRI, CNRS, Université de Bordeaux, France \and \url{https://games-automata-play.com/}}{nathanael.fijalkow@gmail.com}{}{}

\author{Corto Mascle}{LaBRI, Université de Bordeaux, France \and MPI-SWS, Kaiserslautern, Germany \and \url{https://corto-mascle.github.io/}}{corto.mascle@labri.fr}{}{}

\authorrunning{G. Bathie and N. Fijalkow and C. Mascle}

\Copyright{Gabriel Bathie and Nathanaël Fijalkow and Corto Mascle}

\ccsdesc[500]{Theory of computation~Regular languages}

\keywords{property testing, regular languages}

\usepackage[T1]{fontenc}
\usepackage[utf8]{inputenc}
\usepackage{amsmath}
\usepackage{amssymb}
\usepackage{amsthm}
\usepackage{mathtools}
\usepackage{mathrsfs}
\usepackage{centernot}
\usepackage{hyperref}
\usepackage{tikz}
\usepackage{algorithm}
\usepackage[noend]{algpseudocode}
\usepackage{xspace}
\usepackage{ifthen}
\usetikzlibrary{automata,arrows,arrows.meta,decorations,shapes.geometric,patterns,positioning}

\tikzset{
	node distance=1cm,
	every state/.style={thick},
	initial text=$ $,
}

\theoremstyle{plain}
\newtheorem{fact}[theorem]{Fact}
\newtheorem{property}[theorem]{Property}
\theoremstyle{definition}
\newtheorem{problem}[theorem]{Problem}

\DeclareMathOperator{\poly}{poly}
\newcommand{\Aa}{\mathcal{A}}
\newcommand{\Bb}{\mathcal{B}}

\newcommand{\cO}{\mathcal{O}}
\newcommand{\curly}{\mathrel{\leadsto}}
\newcommand{\dd}{.\,.}
\newcommand{\Dd}{\mathcal{D}}
\newcommand{\EE}{\mathbb{E}}
\newcommand{\emptyw}{\gamma}
\newcommand{\eps}{\varepsilon}
\newcommand{\equivportals}{\simeq}
\newcommand{\factor}{\preccurlyeq}
\newcommand{\Ff}{\mathcal{F}}
\newcommand{\geqportals}{\succeq}
\newcommand{\lang}[1]{\mathcal{L}(#1)}
\newcommand{\lcm}{\mathsf{lcm}}
\newcommand{\lefteffect}[2]{(#1 \gg #2)}
\newcommand{\leqportals}{\preceq}

\newcommand{\MBF}{\textsf{MBF}\xspace}
\newcommand{\MBS}{\textsf{MBS}\xspace}
\newcommand{\Mm}{\mathcal{M}}
\newcommand{\nequivportals}{\not\simeq}

\newcommand{\NN}{\mathbb{N}}

\newcommand{\pobs}{\trianglelefteq}
\newcommand{\portal}[4]{#1,#2 \curly #3, #4}
\newcommand{\pow}[1]{2^{#1}}
\newcommand{\PP}{\mathbb{P}}
\newcommand{\Pp}{\mathcal{P}}
\newcommand{\PSPACE}{\textsf{PSPACE}\xspace}
\newcommand{\righteffect}[2]{(#2 \ll #1)}

\newcommand{\Ss}{\mathcal{S}}
\newcommand{\SCCpath}{\pi}
\newcommand{\SCCset}{\mathscr{S}}
\newcommand{\set}[1]{\{ #1 \}}

\newcommand{\timedA}{\widehat{\Aa}}
\newcommand{\timedlanguage}[4]{\mathcal{L}(\portal{#1}{#2}{#3}{#4})}
\newcommand{\timedlang}[1]{\mathcal{PL}( #1 )}
\newcommand{\timedword}[2]{(#1:#2)}
\newcommand{\Tt}{\mathcal{T}}
\newcommand{\twu}{\timedword{0}{u}}

\newcommand{\ZZ}{\mathbb{Z}}
\newcommand{\Prob}[2][]{%
\ifthenelse{\equal{#1}{}}%
{\PP}%
{\PP_{ #1 }}%
\left( #2 \right)}
\newcommand{\epslogeps}[1][]
{%
\ifthenelse{\equal{#1}{}}%
{\log(\eps^{-1})/\eps}%
{\log^{ #1 }(\eps^{-1})/\eps}%
}

\newcommand{\doc}{article\xspace}

\begin{document}

\maketitle
\begin{abstract}
Property testing is concerned with the design of algorithms making a sublinear number of queries to distinguish whether the input satisfies a given property or is far from having this property. A seminal paper of Alon, Krivelevich, Newman, and Szegedy in 2001 introduced property testing of formal languages: the goal is to determine whether an input word belongs to a given language, or is far from any word in that language. They constructed the first property testing algorithm for the class of all regular languages. This opened a line of work with improved complexity results and applications to streaming algorithms. In this work, we show a trichotomy result: the class of regular languages can be divided into three classes, each associated with an optimal query complexity.
Our analysis, which reasons directly and combinatorially on automata, yields effective characterizations of all three classes using so-called minimal blocking sequences.

\end{abstract}

\emph{This paper is the full version of an article published at ICALP 2025~\cite{BathieFM25}.
Compared to the conference version, it contains complete proofs of all results, running examples illustrating the lower bounds, and a discussion of the scope of the results.
It also corrects several points of the conference version: the trichotomy is now stated for non-adaptive testers, a restriction that is necessary by \cref{prop:adaptive}; minimal blocking sequences are considered up to equivalence; and the characterisation is stated separately for each residue modulo~$p$ of the length of the input, since some languages behave differently depending on this residue.}

\section{Introduction}
Property testing was introduced by Goldreich, Goldwasser and Ron~\cite{goldreich1998property} in 1998: it is the study of randomized approximate decision procedures that must distinguish objects that have a given property from those that are \emph{far} from having it.
Because of this relaxation on the specification, the field focuses on very efficient decision procedures, typically with sublinear (or even constant) running time; in particular, the algorithm does not even have the time to read the whole input.

In a seminal paper, Alon et al.~\cite{alon2001regular} introduced property testing of formal languages: given a language $L$ of finite words, the goal is to construct a (randomized) algorithm to determine whether an input word $u$ belongs to the language or is $\eps$-far\footnote{Informally, $u$ is $\eps$-far from $L$ means that even by changing an $\eps$-fraction of the letters of $u$, we cannot obtain a word in $L$. See \cref{sec:preliminaries} for a formal definition.} from it, where $\eps$ is the precision parameter. The output on words that are close to the language but not in it can be arbitrary.
The model assumes random access to the input word: a \emph{query} specifies a position in the word and asks for the letter at that position, and the \emph{query complexity} of the algorithm is the worst-case number of queries it makes to the input.
A tester is \emph{non-adaptive} if it chooses all the positions that it queries before seeing the answer to any query.
Alon et al.~\cite{alon2001regular} showed a surprising result: under the Hamming distance, all regular languages are testable with $\cO(\epslogeps[3])$ queries, where the $\cO(\cdot)$ notation hides constants that depend on the language but not on the length of the input word.
In that paper, they also identified the class of \emph{trivial} regular languages, those for which the answer is always \emph{yes} or always \emph{no} for sufficiently large $n$, e.g., any finite language, or the language of words starting with an $a$, and showed that testing membership in a \emph{non-trivial} regular language requires $\Omega(1/\eps)$ queries.

The results of Alon et al.~\cite{alon2001regular} leave a multiplicative gap of $\cO(\log^3(1/\eps))$ between the best upper and lower bounds.
We set out to improve our understanding of property testing of regular languages by closing this gap.
Bathie and Starikovskaya obtained in 2021~\cite{bathie2021property} the first improvement over the result of Alon et al.~\cite{alon2001regular} in more than 20 years:
\begin{fact}[{From \cite[Theorem 5]{bathie2021property}}]
  Under the edit distance, every regular language can be tested with $\cO(\epslogeps)$ queries.
\end{fact}
Testers under the edit distance are weaker than testers under the Hamming distance, hence this result does not exactly improve the result of Alon et al.~\cite{alon2001regular}. We overcome this shortcoming later in this \doc: \cref{thm:gen-ub} extends the above result to the case of the Hamming distance.

Bathie and Starikovskaya~\cite{bathie2021property} also showed that this upper bound is tight, in the sense that there is a regular language $L_0$ for which this complexity cannot be further improved, thereby closing the query complexity gap.
\begin{fact}[{From \cite[Theorem 15]{bathie2021property}}]
  There is a regular language for which every non-adaptive tester has query complexity $\Omega(\epslogeps)$ under the edit distance\footnote{Note that, as opposed to testers, lower bounds for the edit distance are stronger than lower bounds for the Hamming distance.}, for all small enough $\eps > 0$.
\end{fact}
This being said, it is easy to find specific non-trivial regular languages for which there is an algorithm using only $\cO(1/\eps)$ queries, e.g. $L = a^*$ over the alphabet $\{a,b\}$, $L = (ab)^*$ or $L = (aa + bb)^*$.

Hence, these results combined with those of Alon et al.~\cite{alon2001regular} show that there exist trivial languages (that require 0 queries for large enough $n$), \emph{easy} languages (with query complexity $\Theta(1/\eps)$) and \emph{hard} languages (with query complexity $\Theta(\epslogeps)$).
This raises the question of whether there exist languages with a different query complexity (e.g. $\Theta(\log\log(\eps^{-1})/\eps)$), or if every regular language is either trivial, easy or hard.
It also raises the question of characterizing the languages that belong to each class, inspired by the recent success of exact characterizations of the complexity of sliding window recognition~\cite{GanardiHL18} and dynamic membership~\cite{amarilli2021dynamic} of regular languages.

In this \doc, we answer both questions: we show a trichotomy of the complexity of testing regular languages under the Hamming distance\footnote{We consider non-adaptive testers with one-sided error, also called testing with perfect completeness, see definitions below. As in~\cite{bathie2021property}, the restriction to non-adaptive testers is necessary for the lower bounds, see \cref{sec:intro-overview}.}, showing that there are only the three aforementioned complexity classes (trivial, easy and hard), we give a characterization of all three classes using a combinatorial object called \emph{blocking sequences}, and show that this characterization can be decided in polynomial space (and that it is complete for \PSPACE).
This trichotomy theorem closes a line of work on improving query complexity for property testers and identifying easier subclasses of regular languages.

\subsection{Related work}

A very active branch of property testing focuses on graph properties, for instance one can test whether a given graph appears as a subgraph~\cite{AlonDLRY94} or as an induced subgraph~\cite{alon2000efficient}, and more generally every monotone graph property can be tested with one-sided error~\cite{AlonS08a}. Other families of objects heavily studied under this algorithmic paradigm include probabilistic distributions~\cite{paninski2008coincidence,diakonikolas2016new} combined with privacy constraints~\cite{Aliakbarpour2018DifferentiallyPI},
numerical functions~\cite{blum1990self, rubinfeld1996robust}, and programs~\cite{ergun1998spot,dodis1999improved}.
We refer to the book of Goldreich~\cite{goldreich2017introduction} for an overview of the field of property testing.

\subparagraph*{Testing formal languages.} Building upon the seminal work of Alon et al.~\cite{alon2001regular}, Magniez et al.~\cite{magniez2007property} gave a tester using $\cO(\epslogeps[2])$ queries for regular languages under the edit distance with moves, and François et al.~\cite{francois_et_al:LIPIcs:2016:6355} gave a tester using $\cO(1/\eps^2)$ queries for the case of the weighted edit distance.
Alon et al.~\cite{alon2001regular} also show that some context-free languages cannot be tested with a constant number of queries, and subsequent work has considered testing specific context-free languages such as the \textsc{Dyck} languages~\cite{parnas2003testing,fischer2018improved} or regular tree languages~\cite{magniez2007property}.
Property testing of formal languages has been investigated in other settings: Ganardi et al.~\cite{ganardi2019sliding} studied the question of testing regular languages in the so-called ``sliding window model'', while others considered property testing for subclasses of context-free languages in the streaming model: Visibly Pushdown languages~\cite{francois_et_al:LIPIcs:2016:6355}, \textsc{Dyck} languages~\cite{JN14,Krebs,MagniezMN14} or $\mathsf{DLIN}$ and $\mathsf{LL}(k)$~\cite{BabuLRV13}.
A recent application of property testing of regular languages was to detect race conditions in execution traces of distributed systems~\cite{ThokairZMV23}.
All those results assume that the automaton is fixed and that only the input word and the precision parameter vary, so the constants hidden in the query-complexity bounds may depend on the automaton.
This is consistent with the standard perspective in property testing, where the goal is to determine the complexity of testing a fixed property on arbitrarily large inputs, disregarding the one-time cost of preprocessing its representation.

\subsection{Main result and overview of the paper}\label{sec:intro-overview}
In this section, we give a high-level overview of the approach and notions used for proving the results of this \doc. This section assumes familiarity with classical notions such as finite automata; the formal definitions of all concepts used in this section can be found in \cref{sec:preliminaries}.

Let us start with the notion of a property tester for a language $L$: the goal is to determine whether an input word $u$ belongs to the language $L$, or whether it is $\eps$-far from it. We say that $u$ of length $n$ is \emph{$\eps$-far from $L$} with respect to a metric $d$ over words if all words $v \in L$ satisfy $d(u, v) \ge \eps n$, written $d(u,L) \ge \eps n$.
Throughout this work and unless explicitly stated otherwise, we will consider the case where $d$ is the Hamming distance, defined for two words $u$ and $v$ as the number of positions at which they differ if they have the same length, and as $+\infty$ otherwise.
In that case, $d(u,L) \ge \eps n$ means that one cannot change a proportion $\eps$ of the letters in $u$ to obtain a word in $L$.
We assume random access to the input word: a query specifies a position in the word and asks for the letter in this position.
A \emph{$\eps$-property tester} (or for short, simply a \emph{tester}) $T$ for a language $L$ is a randomized algorithm that, given an input word $u$ of length $n$, always answers ``yes'' if $u\in L$ and answers ``no'' with probability bounded away from 0 when $u$ is $\eps$-far from $L$. In this \doc, testers are \emph{non-adaptive}, meaning that the positions they query are chosen at random before seeing the answer to any query. The complexity measure we focus on is the \emph{query complexity} of $T$, which is the maximum number of queries that $T$ makes to an input of length $n$, as a function of $n$ and $\eps$, in the worst case over all words of length $n$ and all possible random choices.

We can now formally define the classes of \emph{trivial, easy} and \emph{hard} regular languages.
\begin{definition}[Hard, easy and trivial languages]
	Let $L$ be a regular language. We say that:
	\begin{itemize}
		\item $L$ is \emph{hard} if  the optimal query complexity for a (non-adaptive) property tester for $L$ is $\Theta(\epslogeps)$.

		\item $L$ is \emph{easy} if the optimal query complexity for a (non-adaptive) property tester for $L$ is $\Theta(1/\eps)$.

		\item $L$ is \emph{trivial} if there exists $\eps_0> 0$ such that for all positive $\eps < \eps_0$, there is a property tester and some $n \in \NN$ such that the tester makes $0$ queries on words of length $\geq n$.
	\end{itemize}
\end{definition}

\begin{remark}\label{rmk:finite}
    If $L$ is finite, then it is trivial: since there is a bound~$B$ on the lengths of its words, a tester can reject words of length at least $n_0 = B+1$ without querying them.
    For that reason, we only consider \emph{infinite} languages in the rest of the article.
\end{remark}

Our characterization of those three classes uses the notion of \emph{blocking sequence} of a language $L$. Intuitively, they  are sequences of words such that finding those words as factors of a word $w$ proves that $w$ is not in $L$.
We will also define a preorder on those sequences, which gives a notion of \emph{minimal} blocking sequence.

\begin{restatable}{theorem}{ptestmainthm}\label{thm:general}
	Let $L$ be an infinite regular language recognized by an NFA $\Aa$, let $p$ be the least common multiple of the lengths of the simple cycles of $\Aa$, and let $R(L)$ be the set of residues $r \in \ZZ/p\ZZ$ such that $L$ contains words of infinitely many lengths equal to $r$ modulo $p$.
	For $r \in R(L)$, let $\MBS_r(\Aa)$ denote the set of minimal blocking sequences of $\Aa$ for the lengths equal to $r$ modulo $p$.
	The complexity of testing $L$ is characterized as follows:
	\begin{enumerate}
		\item $L$ is trivial if and only if $\MBS_r(\Aa)$ is empty for every $r \in R(L)$;
		\item $L$ is easy if and only if $\MBS_r(\Aa)$ is finite for every $r \in R(L)$, and nonempty for some $r \in R(L)$;
		\item $L$ is hard if and only if $\MBS_r(\Aa)$ is infinite for some $r \in R(L)$.
	\end{enumerate}
\end{restatable}
In the case where $L$ is recognised by a strongly connected automaton, blocking sequences can be replaced by \emph{blocking factors}, which are single words that are not factors of any word in $L$.
The residues modulo $p$ are needed because a language can behave differently depending on the length of the input.
For instance, the language $(aa)^* \cup \set{w \in \set{a,b}^* : |w| \text{ is odd}}$ is trivial on words of odd length, but it requires $\Omega(1/\eps)$ queries on words of even length, where it coincides with $(aa)^*$.

\cref{sec:preliminaries} defines the necessary terms and notations, and the rest of the paper is structured as follows.
In Sections~\ref{sec:scc} and~\ref{sec:general}, we delimit the set of hard languages, that is, the ones that require $\Theta(\epslogeps)$ queries.
More precisely, \cref{sec:scc} focuses on the subcase of languages defined by \emph{strongly connected automata}.
First, we combine the ideas of Alon et al.~\cite{alon2001regular} with those that we presented in~\cite{bathie2021property} to obtain a property tester that uses $\cO(\epslogeps)$ queries for any language with a strongly connected automaton, under the Hamming distance.
Second, we show that if the language of a strongly connected automaton has infinitely many blocking factors then it requires $\Omega(\epslogeps)$ queries.
This result generalizes the result of Bathie and Starikovskaya~\cite{bathie2021property}, which was for a single language, to all regular languages with infinitely many minimal blocking factors. The proof uses Yao's minimax principle~\cite{yao1977probabilistic}, which consists in constructing a hard distribution over inputs and showing that no deterministic property testing algorithm can distinguish between positive and negative instances drawn from this distribution.

In \cref{sec:general}, we extend those results to all automata. The way we use the previous section is different for the upper and the lower bound. For the upper bound of $\cO(\epslogeps)$ queries, we use a natural but technical extension of the proof in the strongly connected case.
Note that since testers for the Hamming distance are also testers for the edit distance, this result is an improvement over the result of Bathie and Starikovskaya~\cite{bathie2021property}, which works under the edit distance.
For the lower bound of $\Omega(\epslogeps)$ queries for languages with infinitely many minimal blocking sequences, we reduce to the strongly connected case. The main difficulty is that it is not enough to consider strongly connected components in isolation: a finite automaton may contain a strongly connected component that induces a hard language, yet the language of the whole automaton is easy\footnote{Think of a non-deterministic automaton made of two disjoint components, one universal, and the other recognizing a hard language. In this case one of the components is subsumed by the other, but we will see that the interaction between components can be more complex.}. We overcome this difficulty with an appropriate definition of minimality for blocking sequences.

\cref{sec:trivial-easy} completes the trichotomy, by characterising the easy and trivial languages.
We show that languages of automata with finitely many minimal blocking sequences can be tested with $\cO(1/\eps)$ queries.
We also prove that if an automaton has at least one blocking sequence, then it requires $\Omega(1/\eps)$ queries to be tested, by showing that our notion of trivial language coincides with the \emph{trivial} languages of Alon et al.~\cite{alon2001regular}.
By contrast, we show that automata without blocking sequences recognize trivial languages.

Once we have the trichotomy, it is natural to ask whether it is effective: given an automaton $\Aa$, can we determine if its language is trivial, easy or hard?
The answer is yes, and we show in \cref{sec:complexity} that all three decision problems are \PSPACE-complete when the input is a non-deterministic automaton, even for strongly connected automata.

\subparagraph*{Scope of the results.}
All our results are stated for the Hamming distance.
Since a tester for the Hamming distance is also a tester for the edit distance, our upper bounds also hold for the edit distance, but our lower bounds do not transfer, as a word can be far from a language for the Hamming distance and close to it for the edit distance.
We conjecture that the same classification holds for the edit distance, but proving it would require adapting both the $\Omega(1/\eps)$ lower bound of Alon et al.~\cite{alon2001regular} for non-trivial languages and our $\Omega(\epslogeps)$ lower bound for hard languages.
For the latter, the lower bound of Bathie and Starikovskaya~\cite{bathie2021property} for a specific language under the edit distance suggests that this is possible.

We only consider testers with perfect completeness (also called one-sided error), which accept every word of the language with probability one.
The known testers for regular languages have this property, which also suits our approach, since our testers reject only when they have found a blocking factor or a blocking sequence.
As our lower bounds rely on perfect completeness, they do not give a classification for testers with two-sided error.

All our testers are non-adaptive, as are those of Alon et al.~\cite{alon2001regular} and Bathie and Starikovskaya~\cite{bathie2021property}, and our lower bounds only hold for non-adaptive testers, as in~\cite{bathie2021property}.
This restriction is necessary, since we show in \cref{prop:adaptive} that the language $(a+b+cPd)^*$, where $P$ is the set of words over $\set{a,b}$ with an even number of $b$'s, can be tested with $\cO(1/\eps)$ queries by an adaptive tester although it is hard for non-adaptive testers.
The complexity of adaptive testing of regular languages remains an open question.

We represent regular languages by NFAs, since our approach does not need determinism and NFAs can be exponentially more succinct than DFAs.
Although blocking sequences are defined on the automaton, the three classes are defined by the query complexity and therefore only depend on the language, so two automata recognising the same language are always in the same class.
The choice of representation only matters for the complexity of deciding the class of a language, and while our \PSPACE-hardness proofs in \cref{sec:complexity} use NFAs, the case of DFAs is open.

\section{Preliminaries}
\label{sec:preliminaries}
\subparagraph*{Words and automata.}
We write $\Sigma^*$ (resp.\ $\Sigma^+$) for the set of finite words (resp.\ non-empty words) over the alphabet $\Sigma$.
The length of a word $u$ is denoted $|u|$, and its $i$th letter is denoted $u[i]$. The empty word is denoted $\emptyw$.
Given $u\in\Sigma^*$ and $0\le i, j \le |u|-1$, define $u[i\dd j]$ as the word $u[i]u[i+1]\ldots u[j]$ if $i\le j$ and $\emptyw$ otherwise.
Further, $u[i\dd j)$ denotes the word $u[i\dd j-1]$.
A word $w$ is a \emph{factor} (resp.\ \emph{prefix}, \emph{suffix}) of $u$ if there exist indices $i,j$ such that $w = u[i\dd j]$ (resp.\ with $i = 0, j = |u|-1$). We use $w\factor u$ to denote ``$w$ is a factor of $u$''.
Furthermore, if $w$ is a factor of $u$ and $w \neq u$, we say that $w$ is a \emph{proper factor} of $u$.

A \emph{nondeterministic finite automaton} (NFA) $\Aa$ is a transition system defined by a tuple $(Q, \Sigma, \delta, q_0, F)$, with $Q$ a finite set of states, $\Sigma$ a finite alphabet, $\delta : Q \times \Sigma \to \pow{Q}$ the transition function, $q_0 \in Q$ the initial state and $F \subseteq Q$ the set of final states. The semantics is as usual~\cite{Pin2021}.
When there is a path from a state $p$ to a state $q$ in $\Aa$, we say that $q$ is \emph{reachable} from $p$ and that $p$ is \emph{co-reachable} from $q$.
In this work, we assume w.l.o.g. that all NFAs $\Aa$ are \emph{trim}, i.e., every state is reachable from the initial state and co-reachable from some final state.

\subparagraph*{Property testing.}
\begin{definition}\label{def:eps-far}
  Let $L$ be a language, let $u$ be a word of length $n$, let $\eps > 0$ be a precision parameter and let $d : \Sigma^*\times\Sigma^*\rightarrow \NN \cup\{+\infty\}$ be a metric.
  We say that the word $u$ is \emph{$\eps$-far from $L$ w.r.t. $d$} if $d(u, L) \ge \eps n$, where
  \[d(u, L) := \inf_{v\in L} d(u,v).\]
\end{definition}
With this definition, the empty word is $\eps$-far from every language, so we only consider inputs of length $n \geq 1$.
We assume random access to the input word: a query specifies a position in the word and asks for the letter in this position.

Throughout this work and unless explicitly stated otherwise, we will consider the case where $d$ is the Hamming distance, defined for two words $u$ and $v$ as the number of positions at which they differ if they have the same length, and as $+\infty$ otherwise.
In that case, $d(u,L) \ge \eps n$ means that one cannot change an $\eps$-fraction of the letters in $u$ to obtain a word in $L$.

A \emph{$\eps$-property tester} (or simply a \emph{tester}) $T$ for a language $L$ is a randomized algorithm that, given an input word $u$, always answers ``yes'' if $u\in L$ and answers ``no'' with probability bounded away from 0 when $u$ is $\eps$-far from $L$.

\begin{definition}\label{def:tester}
  A property tester for the language $L$ with precision $\eps > 0$ is a randomized algorithm $T$ that, for any input $u$ of length $n$, given random access to $u$, chooses the set of positions that it queries as a function of $n$ and of its random choices only, and satisfies the following properties:
  \begin{align*}
    \text{ if } u\in L, & \text{ then } T(u) = 1, \\
    \text{ if $u$ is $\eps$-far from $L$}, & \text{ then } \Prob{T(u) = 0} \ge 2/3.
  \end{align*}
  The query complexity of $T$ is a function of $n$ and $\eps$ that counts the maximum number of queries that $T$ makes over all inputs of length $n$ and over all possible random choices.
\end{definition}

We measure the complexity of a tester by its \emph{query complexity}.
Throughout this \doc, we focus on so-called ``testers with perfect completeness'', or with ``one-sided error'', which answer positively (with probability $1$) on every word in the language, and our characterization applies to this class.
Because they are based on the notion of blocking factors that we will discuss below, all known testers for regular languages~\cite{alon2001regular,magniez2007property,francois_et_al:LIPIcs:2016:6355, bathie2021property} have perfect completeness.
The testers of \cref{def:tester} are also \emph{non-adaptive}, meaning that the positions that they query do not depend on the answers to previous queries. This is the case of the testers of Alon et al.~\cite{alon2001regular} and of Bathie and Starikovskaya~\cite{bathie2021property}, as well as of all the testers that we construct. As explained in \cref{sec:intro-overview} (see also \cref{prop:adaptive}), our lower bounds do not hold for adaptive testers.
In this \doc, we assume that the automaton $\Aa$ that describes the tested language $L$ is \emph{fixed}, and not part of the input, so we consider its number of states $m$ as a constant.

\subparagraph*{Graphs and periodicity.}
We now recall tools introduced by Alon et al.~\cite{alon2001regular} to deal with periodicity in finite automata.

Let $G = (V,E)$ with $E \subseteq V^2$ be a directed graph.
A \emph{strongly connected component} (or SCC) of $G$ is a maximal set of vertices that are all reachable from each other.
It is \emph{transient} if it contains a single state with no self-loop on it, and \emph{recurrent} otherwise.
We say that $G$ is \emph{strongly connected} if its entire set of vertices is an SCC.

The \emph{period} $\lambda = \lambda(G)$ of a recurrent strongly connected graph $G$ is the greatest common divisor of the length of the cycles in $G$.
Following the work of Alon et al.~\cite{alon2001regular}, we will use the following property of directed graphs.
It is a consequence of the well-known fact that given a set of natural numbers, there is a threshold above which any multiple of their gcd can be obtained as a sum of numbers from this set.
Similarly, between two states of an SCC, one can take a path from one to the other and then extend this path by any length (above some threshold, modulo the gcd of cycles) by adding a combination of cycles of the SCC.

\begin{fact}[{From \cite[Lemma 2.3]{alon2001regular}}]\label{fact:periodicity}
  Let $G = (V,E)$ be a non-empty, recurrent, strongly connected graph with finite period $\lambda = \lambda(G)$.
  Then there exists a partition $V = Q_0\sqcup \ldots \sqcup Q_{\lambda-1}$ and a reachability constant $\rho = \rho(G)$ that does not exceed $3|V|^2$ such that:
  \begin{enumerate}
    \item\label{case:path-length} For every $0 \le i,j\le \lambda-1$ and for every $s\in Q_i, t\in Q_j$, the length of any directed path from $s$ to $t$ in $G$ is equal to $(j-i)\mod \lambda$.
    \item For every $0 \le i,j\le \lambda-1$, for every $s\in Q_i, t\in Q_j$ and for every integer $r\ge \rho$, if $r = (j-i) \pmod{\lambda}$, then  there exists a directed path from $s$ to $t$ in $G$ of length $r$.
  \end{enumerate}
\end{fact}

The sets $(Q_i: i = 0,\ldots, \lambda-1)$ are the \emph{periodicity classes} of $G$. In what follows, we will slightly abuse notation and use $Q_i$ even when $i \ge \lambda$ to mean $Q_{i \pmod{\lambda}}$ .

An automaton $\Aa = (Q, \Sigma, \delta, q_0, F)$ defines an underlying graph $G = (Q, E)$ where $E= \{(p,q) \in Q^2 \mid \exists a\in\Sigma: q\in\delta(p, a)\}$.
In what follows, we extend the notions defined above to finite automata through this graph $G$.
Note that the numbering of the periodicity classes is defined up to a shift mod $\lambda$: we can thus always assume that $Q_0$ is the class that contains the initial state $q_0$.
The period of $\Aa$ is written $\lambda(\Aa)$.

\subparagraph*{Positional words and positional languages.}

Consider the language $L_3 = (ab)^*$.
The word $v = ab$ can appear as a factor of a word $u \in L_3$ if $v$ occurs at an even position (e.g. position 0, 2, etc.) in $u$. However, if $v$ occurs at an \emph{odd} position in $u$, then $u$ cannot be in $L_3$.
Therefore, $v$ can be used to witness that $u$ is not in $L_3$, but only if we find it at an odd position.
This example leads us to introducing \emph{$p$-positional words}, which additionally encode information about the index of each letter modulo an integer $p$.

More generally, we will associate to each regular language a period $\lambda$, and working with $\lambda$-positional words will allow us to define blocking factors in a position-dependent way without explicitly considering the index at which the factor occurs.

\begin{definition}[Positional words]
  Let $p$ be a positive integer.
  A $p$-\emph{positional word} is a word over the alphabet $\ZZ/p\ZZ\times \Sigma$ of the form $(n\pmod{p}, a_0) ((n+1) \pmod{p}, a_1) \cdots ((n+\ell) \pmod{p}, a_\ell)$ for some non-negative integer~$n$.
  If $u= a_0 \cdots a_\ell$, we write $\timedword{n}{u}$ to denote this word.
\end{definition}

With this definition, if $u = abcd$ and we consider the $2$-positional word $\tau = \timedword{1}{u}$, the factor $bc$ appears at position $1$ in $u$ and is mapped to the factor $\mu = (0, b)(1, c)$.
In this case, even when taking factors of $\mu$, we still retain the (congruence classes of the) indices in the original word $\tau$.

Any strongly connected finite automaton $\Aa = (Q, \Sigma, \delta, q_0, F)$ can be extended into an automaton $\timedA$ over $\lambda(\Aa)$-positional words with at most $\lambda(\Aa)|Q|$ states.
It suffices to keep track in the states of the current state of $\Aa$ and the number of letters read modulo $\lambda(\Aa)$.
Note that the resulting automaton is still trim and strongly connected: if we can read a word $(0:u)$ from $(q_0, 0)$ to $(q,i)$, then let $v$ be a word labeling a path from $q $ to $q_0$ in $\Aa$. Then $v$ must have length $-i$ modulo $\lambda(\Aa)$, by definition of $\lambda(\Aa)$. Hence we have a path reading $(i:v)$ from $(q,i)$ to $(q_0,0)$.

We call the language recognized by $\timedA$ the \emph{positional language of $\Aa$},
and denote it $\timedlang{\Aa}$. This definition is motivated by the following property:
\begin{property}
  For any word $u\in\Sigma^*$, we have $u\in\lang{\Aa}$ if and only if $\twu\in\timedlang{\Aa}$.
\end{property}

Positional words make it easier to manipulate factors with positional information, hence we phrase our property testing results in terms of positional languages. Notice that a property tester for $\timedlang{\Aa}$ immediately gives a property tester for $\lang{\Aa}$, as one can simulate queries to $\twu$ with queries to $u$ by simply pairing the index of the query modulo $\lambda(\Aa)$ with its result.

Moreover, we show that the Hamming distance between~$u$ and $\lang{\Aa}$ is the same as the (Hamming) distance between $\twu$ and $\timedlang{\Aa}$.
\begin{claim}\label{claim:far-equiv}
  For any word $u\in\Sigma^*$, we have $d(u, \lang{\Aa}) = d(\twu, \timedlang{\Aa})$.
\end{claim}
\begin{claimproof}
  By definition of $\timedA$, every positional word of $\timedlang{\Aa}$ is of the form $\timedword{0}{v}$, and by the above property, $\timedword{0}{v} \in \timedlang{\Aa}$ if and only if $v \in \lang{\Aa}$.
  Moreover, for every word $v$ of the same length as $u$, the positional words $\twu$ and $\timedword{0}{v}$ carry the same index at every position, hence $d(\twu, \timedword{0}{v}) = d(u,v)$.
  Taking the infimum over $v \in \lang{\Aa}$ gives the result.
\end{claimproof}
The above claim allows us to interchangeably use the statements ``$u$ is $\eps$-far from $\lang{\Aa}$'' and ``$\twu$ is $\eps$-far from $\timedlang{\Aa}$''.

\section{Hard languages for strongly connected NFAs}
\label{sec:scc}
Before considering the case of arbitrary NFAs, we first study the case of strongly connected NFAs, which are NFAs such that for any pair of states $p,q\in Q$ there exists a word~$w$ such that $p\xrightarrow{w}q$.
We will generalize the results of this section to all NFAs in Section~\ref{sec:general}.

We show that the query complexity of the language of such an NFA~$\Aa$ can be characterized by the cardinality of the set of \emph{minimal blocking factors} of~$\Aa$, which are factor-minimal $\lambda(\Aa)$-positional words that witness the fact that a word does not belong to $\timedlang{\Aa}$.
In this section, we consider a fixed NFA~$\Aa$ and simply use ``positional words'' to refer to $\lambda$-positional words, where $\lambda = \lambda(\Aa)$ is the period of~$\Aa$.

\begin{definition}[Blocking factors]\label{def:blocking}
  Let~$\Aa$ be a strongly connected NFA.
  A positional word $\tau$ is a \emph{blocking factor} of~$\Aa$
  if for any positional word $\mu$
  we have $\tau \factor \mu \Rightarrow \mu \notin \timedlang{\Aa}$.

  Further, we say that $\tau$ is a \emph{minimal} blocking factor of~$\Aa$
  if no proper factor of $\tau$ is a blocking factor of~$\Aa$.
  We use $\MBF(\Aa)$ to denote the set of all minimal blocking factors of~$\Aa$.
\end{definition}
Intuitively and in terms of automata, the positional word~$\timedword{i}{u}$ is blocking for~$\Aa$ if no path of~$\Aa$ labelled by $u$ starts in a state of $Q_i$. (This property is formally established later in \cref{lemma:intuition-timed-word}.)
The main result of this section is the following:
\begin{theorem}\label{thm:scc}
  Let $L$ be an infinite language recognised by a strongly connected NFA~$\Aa$. If $\MBF(\Aa)$ is infinite, then $L$ is hard, i.e., it has query complexity $\Theta(\epslogeps)$
\end{theorem}
This result gives both an upper bound of $\cO(\epslogeps)$ and a lower bound of $\Omega(\epslogeps)$ on the query complexity of a tester for~$L$:
we prove the upper bound in \cref{sec:scc-ub} and the lower bound in \cref{sec:scc-lb}.

\subsection{Positional words, blocking factors and strongly connected NFAs}\label{sec:scc-tw}

We first establish some properties of positional words that will help us ensure that we are creating well-formed positional words, that is, positional words where the index $i$ of a letter $\timedword{i}{a}$ is equal to $j+1\pmod{\lambda}$, where $j$ is the index of the previous letter. In \cref{sec:scc-ub}, we explain the connection between property testing and blocking factors in strongly connected NFAs.

We start with the following properties, which are consequences of \cref{fact:periodicity}.
\begin{corollary}\label{corollary:position-length}
  Let~$n$ be a nonnegative integer, let~$w$ be a word of length~$n$.
  If for some states $p\in Q_i, q\in Q_j$ of~$\Aa$ we have $p\xrightarrow{w}q$,
  then the indices $i,j$ satisfy the equation
  \begin{equation*}
    j-i = |w| \pmod{\lambda}
  \end{equation*}
\end{corollary}
\begin{corollary}\label{corollary:position-factor}
  Let $\tau = \timedword{i}{u}$ and $\mu = \timedword{j}{v}$ be positional words, with $u$ non-empty.
  If $\tau \factor \mu$, then there exists positional words $\eta, \eta'$ with $|\eta| = i-j \pmod{\lambda}$
  such that $\mu = \eta\tau\eta'$.
  In particular, this implies that there exists words $w,w'$ with $|w| = i-j \pmod{\lambda}$
  such that $v = wuw'$.
\end{corollary}

These properties allows us to formalize the intuition we gave earlier about blocking factors.

\begin{lemma}\label{lemma:intuition-timed-word}
  A positional word $\tau = \timedword{i}{u}$ is a blocking factor for $\Aa$
  iff for every states $p\in Q_i, q\in Q$, we have $p\centernot{\xrightarrow{u}}q$.
\end{lemma}

\begin{proof}
  If $u$ is empty, then $\tau$ is a factor of every positional word, so it is not blocking since $\timedlang{\Aa}$ is non-empty; moreover $Q_i$ is non-empty, as every periodicity class of a recurrent strongly connected graph is non-empty, and $p\xrightarrow{\emptyw}p$ for $p \in Q_i$. Hence both conditions fail, and we now assume that $u$ is non-empty.
  We first show that if there exists states $p\in Q_i, q\in Q$ such that $p\xrightarrow{u}q$,
  then $\tau$ is not blocking, i.e. there exists $\mu\in \timedlang{\Aa}$ such that $\tau\factor\mu$.
  As~$\Aa$ is strongly connected, there exist positional words $\eta,\eta'$ such that $q_0 \xrightarrow{\eta} p$
  and $q\xrightarrow{\eta'} q_f$ for some $q_f\in F$.
  By \cref{fact:periodicity}, the positional word $\mu = \eta\tau \eta'$ is well-formed.
  Furthermore, it labels a path from $q_0$ to $q_f$, hence it is in $\timedlang{\Aa}$, and $\tau$ is not blocking.

  For the converse, assume that $\tau$ is non-blocking: we show that there exists two states $p\in Q_i, q\in Q$ such that $p\xrightarrow{u}q$.
  As $\tau$ is non-blocking, there exists a positional word $\mu = \timedword{0}{w}$ such that $\tau\factor \mu$
  and there exists a final state~$r$ such that $q_0\xrightarrow{\mu}r$, and equivalently, $q_0\xrightarrow{w}r$.
  By \cref{corollary:position-factor}, since $\tau\factor \mu$, there exists words $v,v'$ such that $w = vuv'$ and the length of~$v$ is equal to~$i$ modulo $\lambda$.
  In particular, the path $q_0\xrightarrow{w}r$ can be decomposed into $q_0\xrightarrow{v}p\xrightarrow{u}q\xrightarrow{w}r$, and we have $p\xrightarrow{u}q$.
  It only remains to show that~$p$ is in $Q_i$: this follows by \cref{corollary:position-length} since $|v| = i\pmod{\lambda}$.
\end{proof}

\subsection{An efficient property tester for strongly connected NFAs.}\label{sec:scc-ub}
In this section, we show the following statement.

\begin{theorem}\label{thm:generic-tester-scc}
  Let~$\Aa$ be a strongly connected NFA.
  For any small enough $\eps > 0$, there exists an $\eps$-property tester for $\lang{\Aa}$
  that uses $\cO(\epslogeps)$ queries.
\end{theorem}

Our proof is similar to the one given in~\cite{bathie2021property} for the edit distance, with one notable technical improvement: we use a new method for sampling factors in~$u$, which greatly simplifies the correctness analysis.

\subsubsection{An efficient sampling algorithm}

We first introduce a sampling algorithm (\cref{alg:generic-sampling}) that uses few queries and has a high probability of finding at least one factor from a large set $\Ss$ of disjoint ``special'' factor occurrences.
Using this algorithm on a large set of disjoint blocking factors gives us an efficient property tester for strongly connected NFAs. We will re-use this sampling procedure later in the case of general NFAs (\cref{thm:gen-ub}).

The procedure is fairly simple, as the algorithm samples factors of various lengths in~$u$ at random, but the correctness of the tester is far from trivial.
The lengths and the number of factors of each length are chosen so that the number of queries is minimized and the probability of finding a ``special'' factor is maximized, regardless of their distribution within~$u$. (In what follows, the ``special'' factors are blocking factors.)

Note that the following algorithm does not  directly take $\Ss$ into account: it only takes two parameters $N$ and $B$ as input, which should be thought of as a lower bound on the number of elements of $\Ss$ in the word and an upper bound on their maximum length.
The algorithm works on arbitrary words; we will apply it to positional words, in which case the factors that it returns are themselves positional words.

\begin{algorithm}[htbp]
  \caption{Efficient generic sampling algorithm}\label{alg:generic-sampling}
  \begin{algorithmic}[1]
    \Function{OneSample}{$u, \ell$}
    \Statex \hspace{\algorithmicindent}\textit{Input:} random access to a word $u$, and an integer $\ell \geq 1$.
    \Statex \hspace{\algorithmicindent}\textit{Output:} a factor occurrence of $u$ of length at most $2\ell+1$.
    \State $n \gets |u|$
    \State $i\gets \Call{uniform}{0, n-1}$
    \State $i_1\gets \max(i-\ell, 0), i_2\gets\min(i+\ell, n-1)$
    \State \Return $(u[i_1\dd  i_2], [i_1,i_2])$
    \EndFunction
    \Function{Sampler}{$u, N, B$}
    \Statex \hspace{\algorithmicindent}\textit{Input:} random access to a word $u$, and two integers $N, B \geq 1$.
    \Statex \hspace{\algorithmicindent}\textit{Output:} a set of factor occurrences of $u$.
    \State $n \gets |u|$
    \State $\beta \gets n/N$
    \State $T \gets \lceil\log(B)\rceil$
    \State $\Ff \gets \emptyset$
    \For{$t = 0$ to $T$}
    \State $r_t \gets 2^t, c_t \gets \lceil 2\ln(3)\beta/r_t \rceil$
    \For{$i = 1$ to $c_t$}
    \State $\Ff \gets \Ff \cup \{\Call{OneSample}{u, r_t}\}$
    \EndFor
    \EndFor
    \State \Return $\Ff$
    \EndFunction

  \end{algorithmic}
\end{algorithm}

\begin{claim}\label{claim:generic-sampling-alg-complexity}
  A call to $\Call{Sampler}{u, N, B}$ (\cref{alg:generic-sampling}) makes $\cO(n(1+\log B)/N)$ queries to~$u$, assuming $1 \le N\leq n/B$.
\end{claim}
\begin{claimproof}
  A call to $\Call{OneSample}{u, r_t}$ makes at most $2r_t+1$ queries to $u$.
  Since $N \leq n/B$, we have $2\ln(3)\cdot n/(Nr_t) \geq 2\ln(3)B/2^T  \geq \ln(3) \geq 1$. As a consequence, the function $\Call{Sampler}{u, N, B}$ makes $c_t = \lceil 2\ln(3)\cdot \beta/r_t\rceil \leq  4\ln(3)\cdot \beta/r_t = 4\ln(3)\cdot n/(Nr_t)$ calls to $\Call{OneSample}{u, r_t}$ for each $t = 0,\ldots, T$, where $T = \lceil\log(B)\rceil$.

  This adds up to \[\sum_{t=0}^T  c_t \cdot(2r_t+1) \leq  \sum_{t=0}^T  c_t \cdot(3r_t) \leq 3 \cdot(\lceil\log(B)\rceil +1) \cdot 4\ln(3)\cdot n/N = \cO(n(1+\log B)/N)\] queries to $u$.
\end{claimproof}

An \emph{occurrence} of a word $w$ in a word $u$ is an interval $[i,j]$ such that $w=u[i..j]$.
A \emph{factor occurrence} in $u$ is a word $w$ along with an occurrence $[i,j]$ of $w$ in $u$.
Given a set of factor occurrences in $u$, we say they are disjoint if the intervals are pairwise disjoint.
A factor occurrence $(w, [i,j])$ covers another $(w',[i',j'])$ if $[i',j']$ is included in $[i,j]$.

\begin{lemma}\label{lemma:generic-sampling-alg}
  Let $u$ be a word of length~$n$, and consider a set~$\Ss$ containing at least~$N$ pairwise disjoint factor occurrences of~$u$, each of length at most~$B$.
  A call to the function $\Call{Sampler}{u, N, B}$ (\cref{alg:generic-sampling}) returns a set $\Ff$ of factor occurrences of $u$ such that, with probability at least $2/3$, one of them covers a factor occurrence of $\Ss$.
\end{lemma}

\begin{proof}
  We conceptually divide the factor occurrences in $\Ss$ into different categories depending on their length:
  let $T = \lceil\log(B)\rceil$, and for $t= 0,\ldots, T,$ let $S_t$ denote the subset of $\Ss$ which contains factor occurrences of length at most $r_t = 2^t$.
  We then analyze the probability that randomly sampled factor occurrences of length $2r_t$ cover an element
  from $S_{t}$, and show that over all~$t$, at least one sampled factor occurrence covers one of $\Ss$, with probability at least $2/3$.

  \begin{claim}
    If in a call to \textsc{OneSample}, the values $i_1, i_2$ are such that $[i_1, i_2]$ contains the interval of a factor occurrence in $\Ss$, then the set $\Ff$ returned by the algorithm contains a factor occurrence  covering an element of $\Ss$.
  \end{claim}

  As the factor occurrences given in $\Ss$ are disjoint, so are the ones in $S_t$.
  As they have length at most $r_t$,
  the probability $p_t$ that the factor occurrence returned by \textsc{OneSample} covers one from $S_t$ is lower bounded by the probability to sample a position inside one of them,
  \(p_t \ge \frac{1}{n}\sum_{(v, [i,j])\in S_t} |v|.\)
  The \textsc{OneSample} function is called $c_t = \lceil 2\ln(3)\beta/r_t \rceil$ times independently for each~$t$,
  hence the probability~$p$ that the algorithm samples a factor occurrence covering an element of $\Ss$ satisfies the following:
  \begin{align*}
    (1-p)
    &\le \prod_{t=0}^T (1-p_t)^{c_t} \le \exp\left(\sum_{t=0}^T c_t \ln(1-p_t)\right) \le \exp\left(-\sum_{t=0}^T p_t c_t\right)\\
    &\le \exp\left(-\frac{2\ln(3)\beta}{n} \sum_{t=0}^T \frac{1}{r_t} \sum_{(v, [i,j])\in S_t} |v|\right)\\
    &= \exp\left(-\frac{2\ln(3)\beta}{n} \sum_{(v, [i,j])\in \Ss}|v| \sum_{t=\lceil\log|v|\rceil}^T 2^{-t}\right).
  \end{align*}
  Now lower bounding the sum of powers of $2$ by its first term,
  we obtain:
  \begin{align*}
    (1-p) &\le \exp\left(-\frac{2\ln(3)\beta}{n} \sum_{(v, [i,j])\in \Ss}|v| \cdot 2^{-\lceil\log|v|\rceil}\right)\\
    &\le \exp\left(-\frac{2\ln(3)\beta}{n} \sum_{(v, [i,j])\in \Ss}|v| \frac{1}{2|v|}\right)\\
    &= \exp\left(-\frac{2\ln(3)\beta}{n} \cdot\frac{|\Ss|}{2}\right)
    \le \exp\left(-\frac{\ln(3)\beta N}{n}\right)\\
    &= \exp\left(-\ln(3) \right) = 1/3
  \end{align*}
  It follows that $p \ge 2/3$, which concludes the proof.
\end{proof}

\subsubsection{The tester}

The algorithm for \cref{thm:generic-tester-scc} is given in \cref{alg:generic-tester-scc}.

\begin{algorithm}[htbp]
  \caption{Generic $\eps$-property tester that uses $\cO(\epslogeps)$ queries}\label{alg:generic-tester-scc}
  \begin{algorithmic}[1]
    \Function{Tester}{$u, \eps$}
    \State $n \gets |u|, m\gets |Q|$
    \State $B \gets 12m^2/\eps$
    \If{$\lang{\Aa} \cap \Sigma^n = \emptyset$}
    \State Reject
    \ElsIf{$n < B$}
    \State Query all of $u$ and run~$\Aa$ on it
    \State Accept if and only if~$\Aa$ accepts
    \Else \label{line:intersting-case}
    \State \label{line:compute-f}$\Ff \gets\Call{Sampler}{\twu, n/B, B}$
    \State Reject if and only if $\Ff$ contains a blocking factor for~$\Aa$.
    \EndIf
    \EndFunction
  \end{algorithmic}
\end{algorithm}

We now show that \cref{alg:generic-tester-scc} is a property tester for $\lang{\Aa}$ that uses $\cO(\epslogeps)$ queries. In what follows, we use $n$ to denote the length of the input word $u$ and $m$ to denote the number of states of $\Aa$.
\begin{claim}
  The tester given in \cref{alg:generic-tester-scc} makes $\cO(\epslogeps)$ queries to~$u$.
\end{claim}
\begin{proof}
  If $n < B$, then the tester makes $n < B =\cO(1/\eps)$ queries, and the claim holds.
  Otherwise, the number of queries is given by the call to $\Call{Sampler}{u, n/B, B}$:
  by \cref{claim:generic-sampling-alg-complexity}, this uses $\cO(\frac{n (1+\log B)}{n/B}) = \cO(B (1+\log B)) = \cO(\epslogeps)$ queries for $\eps \le 1$.
\end{proof}

We have presented an algorithm to search for blocking factors within the desired number of queries, it remains to argue that this is a valid property testing algorithm for the language.
Alon et al.~\cite[Lemma 2.6]{alon2001regular} first noticed that if a word~$u$ is $\eps$-far from $\lang{\Aa}$, then it contains $\Omega(\eps n)$ short factors that witness the fact that~$u$ is not in $\lang{\Aa}$.
We start by translating the lemma of Alon et al. on ``short witnesses'' to the framework of blocking factors.
More precisely, we show that if~$u$ is $\eps$-far from $\lang{\Aa}$, then $\twu$ contains many disjoint (i.e. non-overlapping) blocking factors.

\begin{lemma}\label{lemma:many-blocking}
  Let $\eps> 0$, let~$u$ be a word of length $n \ge 12m^2/\eps$ and assume that $\lang{\Aa}$ contains at least one word of length~$n$.
  If $\tau = \twu$ is $\eps$-far from $\timedlang{\Aa}$, then $\tau$ contains at least $\eps n/(6m^2)$ disjoint blocking factor occurrences.
\end{lemma}
\begin{proof}
  We build a set $\Pp$ of disjoint blocking factors of $\tau$ as follows: we process~$u$ from left to right, starting at index $i_1 = \rho$, where~$\rho$ is the reachability constant of $\Aa$ (see \cref{fact:periodicity}).
  Next, at iteration~$t$, set~$j_t$ to be the smallest integer greater than or equal to~$i_t$ and smaller than $n-\rho$ such that $\tau[i_t\dd  j_t]$ is a blocking factor occurrence. If there is no such integer, we stop the process.
  Otherwise, we add $\tau[i_t\dd  j_t+\rho-1]$ to the set~$\Pp$, and iterate starting from the index $i_{t+1} = j_t+\rho$.

  Let~$k$ denote the size of~$\Pp$.
  We will show that we can substitute at most $3(k+2)m^2$ positions in $\tau$ to obtain a word in $\timedlang{\Aa}$. (See \cref{fig:many-blocking} for an illustration of this construction.)
  Using the assumption that $\tau$ is $\eps$-far from $\timedlang{\Aa}$ (which follows from \cref{claim:far-equiv}) will give us the desired bound on~$k$.

  \begin{figure}[htbp]
    \begin{center}
      \begin{tikzpicture}[scale=0.9]
        \node (label) at (-.5, .25) {\textbf{(a)}};
        \draw (0, 0) rectangle ++(14, .5);
        \draw[pattern=crosshatch,pattern color=red!40] (0, 0) rectangle ++(1, .5);
        \draw[pattern=north west lines,pattern color=gray!40] (1, 0) rectangle ++(2.8, .5) node[midway] {$\tau[i_1\dd  j_1]$};
        \draw[pattern=crosshatch,pattern color=red!40] (3.8, 0) rectangle ++(1, .5);
        \draw[pattern=north west lines,pattern color=gray!40] (4.8, 0) rectangle ++(2.8, .5) node[midway] {$\tau[i_2\dd  j_2]$};
        \draw[pattern=crosshatch,pattern color=red!40] (7.6, 0) rectangle ++(1, .5);
        \node (dots) at (9.3, .25) {$\ldots$};
        \draw[pattern=north west lines,pattern color=gray!40] (10, 0) rectangle ++(3, .5) node[midway] {$\tau[i_k\dd  j_k]$};
        \draw[pattern=crosshatch,pattern color=red!40] (13, 0) rectangle ++(1, .5);
        \node[white] (p1) at (1, -.3) {$p_1$};
        \node[white] (qf) at (14, -.3) {$q_f$};
      \end{tikzpicture}
      \begin{tikzpicture}[scale=0.9]
        \node (label) at (-.5, .25) {\textbf{(b)}};
        \draw (0, 0) rectangle ++(14, .5);
        \draw[pattern=north west lines,pattern color=gray!40] (1, 0) rectangle ++(2.5, .5) node[midway] {$\tau[i_1\dd  j_1-1]$};
        \node (p1) at (1, -.3) {$p_1$};
        \node (q1) at (3.5, -.3) {$q_1$};
        \draw[->] (p1) -> (q1);
        \draw[pattern=north west lines,pattern color=gray!40] (4.8, 0) rectangle ++(2.5, .5) node[midway] {$\tau[i_2\dd  j_2-1]$};
        \node (p2) at (4.8, -.3) {$p_2$};
        \node (q2) at (7.3, -.3) {$q_2$};
        \draw[->] (p2) -> (q2);

        \node (p3) at (8.8, -.3) {$p_3\ldots$};
        \node (dots) at (9.3, .25) {$\ldots$};
        \draw[pattern=north west lines,pattern color=gray!40] (10, 0) rectangle ++(2.6, .5) node[midway] {$\tau[i_k\dd  j_k-1]$};
        \node (pk) at (10, -.3) {$p_k$};
        \node (qk) at (12.6, -.3) {$q_k$};
        \draw[->] (pk) -> (qk);
        \node[white] (qf) at (14, -.3) {$q_f$};
      \end{tikzpicture}
      \begin{tikzpicture}[scale=0.9]
        \node (label) at (-.5, .25) {\textbf{(c)}};
        \draw (0, 0) rectangle ++(14, .5);
        \draw[pattern=north west lines,pattern color=gray!40] (1, 0) rectangle ++(2.5, .5) node[midway] {$\tau[i_1\dd  j_1-1]$};
        \node (p1) at (1, -.3) {$p_1$};
        \node (q1) at (3.5, -.3) {$q_1$};
        \draw[pattern=north west lines,pattern color=gray!40] (4.8, 0) rectangle ++(2.5, .5) node[midway] {$\tau[i_2\dd  j_2-1]$};
        \node (p2) at (4.8, -.3) {$p_2$};
        \node (q2) at (7.3, -.3) {$q_2$};

        \node (p3) at (8.8, -.3) {$p_3\ldots$};
        \node (dots) at (9.3, .25) {$\ldots$};
        \draw[pattern=north west lines,pattern color=gray!40] (10, 0) rectangle ++(2.6, .5) node[midway] {$\tau[i_k\dd  j_k-1]$};
        \node (pk) at (10, -.3) {$p_k$};
        \node (qk) at (12.6, -.3) {$q_k$};

        \node (q0) at (0, -.3) {$q_0$};
        \node (qf) at (14, -.3) {$q_f$};
        \draw[->] (q0) -> (p1);
        \draw[->] (p1) -> (q1);
        \draw[->] (q1) -> (p2);
        \draw[->] (p2) -> (q2);
        \draw[->] (q2) -> (p3);
        \draw[->] (p3) -> (pk);
        \draw[->] (pk) -> (qk);
        \draw[->] (qk) -> (qf);
      \end{tikzpicture}

      \caption[Decomposition process for finding many blocking factors in $\eps$-far words]{\textbf{(a)} The decomposition process returns~$k$ factors $\tau[i_1\dd j_1], \ldots, \tau[i_k\dd j_k]$ (represented as diagonally hatched gray regions), separated together and from the start and end of the word by padding regions of $\rho-1$ letters (red crosshatched regions).
        \textbf{(b)} If we exclude the last letter of each blocking factor, we obtain factors that label transitions between some pair of states $p_t, q_t$ for each $t = 1,\ldots, k$.
      \textbf{(c)} We use the padding regions to bridge between consecutive factors as well as the start and end of the word.}
      \label{fig:many-blocking}
    \end{center}
  \end{figure}

  For every~$t$, we chose $j_t$ to be minimal so that $\tau[i_t\dd j_t]$ is blocking, hence $\tau[i_t\dd j_t-1]$ is not blocking, and therefore $\tau[i_t\dd j_t-1]$ labels a run from some state $p_t \in Q_{i_t}$ to some state $q_t \in Q_{j_t}$ (see Lemma~\ref{lemma:intuition-timed-word}).
  Therefore, using the strong connectivity of~$\Aa$ and \cref{fact:periodicity}, we can substitute  the letters in $\tau[j_t\dd  j_t+\rho-1]$ to obtain a factor that labels a path from~$q_t$ to~$p_{t+1}$. After this transformation, the word $\tau[i_t\dd  j_t+\rho-1]$ labels a path from~$p_t$ to~$p_{t+1}$ ($\rho$ is taken from Fact~\ref{fact:periodicity} and is at most $3m^2$).
  Let $i_{k+1} = j_k + \rho$ (or $i_1 = \rho$ if $k = 0$). When the process stops, no factor $\tau[i_{k+1}\dd j]$ with $j < n - \rho$ is blocking, hence $\tau[i_{k+1}\dd n-\rho-1]$ labels a run from some state $p_{k+1} \in Q_{i_{k+1}}$ to some state $q_{k+1} \in Q_{n-\rho}$ (if $i_{k+1} \geq n-\rho$, this factor is empty and we let $p_{k+1} = q_{k+1}$ be any state of $Q_{n-\rho}$).
  In particular, the letters $\tau[j_k\dd j_k+\rho-1]$ that follow the last factor of $\Pp$ are used to bridge from $q_k$ to $p_{k+1}$, as above.
  Using the $\rho$ letters at the start and the end of the word, we add transitions from an initial state to $p_1$ and from $q_{k+1}$ to a final state; this is always possible since $\lang{\Aa}$ contains a word of length~$n$, so that $Q_n$ contains a final state.
  The resulting word is in $\timedlang{\Aa}$ and was obtained from $\tau$ using at most $(k+2)\rho \le 3(k+2)m^2$ substitutions: $\rho$ at the start, $\rho$ after each of the $k$ factors of $\Pp$, and $\rho$ at the end.
  As $\tau$ is $\eps$-far from $\timedlang{\Aa}$, we obtain the following bound on $k$:
  \begin{align*}
    3(k+2)m^2 \ge \eps n
    &\Longrightarrow k \ge \frac{\eps n}{3m^2}-2\\
    &\Longrightarrow k \ge \frac{\eps n}{6m^2}
  \end{align*}
  The last implication uses the assumption that $n \ge 12m^2 / \eps$.
\end{proof}

Next, we show that if~$u$ is $\eps$-far from $\lang{\Aa}$, then $\twu$ contains $\Omega(\eps n)$ blocking factors occurrences, each of length $\cO(1/\eps)$.
\begin{lemma}\label{lemma:many-short-blocking}
  Let $\eps> 0$, let~$u$ be a word of length $n \ge 12m^2/\eps$ and assume that $\lang{\Aa}$ contains at least one word of length~$n$.
  If~$u$ is $\eps$-far from $\lang{\Aa}$, then the positional word $\twu$ contains at least $\eps n/(12m^2)$ disjoint blocking factor occurrences of length at most $12m^2/\eps$.
\end{lemma}
\begin{proof}
  Let $u, \Aa$ be a word and an automaton satisfying the above hypotheses.
  By \cref{lemma:many-blocking}, $\twu$ contains at least $\eps n/(6m^2)$ \emph{disjoint} blocking-factor occurrences.
  As these factor occurrences are disjoint, at most half of them (that is, $\eps n/(12m^2)$ of them) can have length greater than $12m^2/\eps$, as the sum of their lengths cannot exceed~$n$.
  The rest of the  factor occurrences therefore forms a set of at least $\eps n/(12m^2)$ disjoint blocking factor occurrences of length at most $12m^2/\eps$.
\end{proof}

\begin{proof}[Proof of \cref{thm:generic-tester-scc}]
  There are three cases: $u$ belongs to $\lang{\Aa}$, $u$ is $\eps$-far from $\lang{\Aa}$, or neither. In the last case, any answer is correct.
  If $u\in\lang{\Aa}$, then $\twu$ cannot contain a blocking factor for $\Aa$, hence \cref{alg:generic-tester-scc} always accepts $u$.
  Now assume $u$ is $\eps$-far from $\lang{\Aa}$.
  If $\lang{\Aa} \cap \Sigma^n$ is empty or if $|u| \le B = 12m^2/\eps$, the tester has the same output as $\Aa$, hence it is correct.
  Now assume $|u| > B$ and $\lang{\Aa}$ contains a word of length $n$. Then \cref{lemma:many-short-blocking} gives us a large set of short blocking factor occurrences in $\twu$, which is what the \textsc{Sampler} function needs to find at least one factor occurrence covering a blocking factor with probability at least $2/3$.
  More precisely, by \cref{lemma:many-short-blocking}, $\twu$ contains at least $\eps n/(12m^2)= n/B$ blocking factor occurrences of length at most $B = 12m^2/\eps$, hence the conditions of \cref{lemma:generic-sampling-alg} are satisfied.

  As a factor containing a blocking factor is also a blocking factor, the set $\Ff$ computed on line~\ref{line:compute-f} of \cref{alg:generic-tester-scc} contains at least one blocking factor with probability at least~$2/3$, and \cref{alg:generic-tester-scc} satisfies \cref{def:tester}.
\end{proof}

\subsection{Lower bound from infinitely many minimal blocking factors}\label{sec:scc-lb}
We now show that languages with infinitely many \emph{minimal} blocking factors are hard, i.e. any (non-adaptive) tester for such a language requires $\Omega(\epslogeps)$ queries.

Let us first give an example that will motivate our construction.
Consider the parity language~$P$ consisting of words that contain an even number of $b$'s, over the alphabet $\set{a,b}$.
Distinguishing $u \in P$ from $u \notin P$ requires $\Omega(n)$ queries, as changing the letter at a single position can change membership in~$P$. However, $P$ is trivial to test, as any word is at distance at most~1 from~$P$, for the same reason.
Now, consider the language $L_0 = (a + b + cPd)^*$ over $\set{a,b,c,d}$, consisting of the words in which the letters $c$ and $d$ alternate, starting with a $c$ and ending with a $d$, and such that between a~$c$ and the next~$d$, there is a word in~$P$.
Intuitively, this language encodes multiple instances of~$P$, hence we can construct words $\eps$-far from $L_0$, and each instance is hard to recognize for property testers, hence the whole language is.
In~\cite[Theorem 15]{bathie2021property}, Bathie and Starikovskaya proved a lower bound of $\Omega(\epslogeps)$ on the query complexity of non-adaptive property testers for a language of this kind, matching the upper bound in the same paper.

The minimal blocking factors of $L_0$ include all words of the form $cvd$ where $v\in\set{a,b}^*$ and $v\notin P$: there are infinitely many such words.
This is not a coincidence, as we show in this paper that this lower bound can be lifted to any language with infinitely many minimal blocking factors, under the Hamming distance.

\begin{theorem}\label{thm:scc-lb}
  Let~$\Aa$ be a strongly connected NFA.
  If $\MBF(\Aa)$ is infinite, then for all small enough $\eps >0$, any non-adaptive $\eps$-property tester for $L = \lang{\Aa}$ uses $\Omega(\epslogeps)$ queries on words of length~$n$, for all large enough $n$ such that $L$ contains a word of length~$n$.
\end{theorem}
The restriction to lengths $n$ such that $L$ contains a word of length $n$ is necessary, as a tester can reject all words of other lengths without making any query.
The proof of this result is a full generalization of the lower bound against the ``repeated parity'' example given above.

Our proof is based on (a consequence of) Yao's Minimax Principle~\cite{yao1977probabilistic}: if there is a distribution $\Dd$ over inputs such that any \emph{deterministic} non-adaptive algorithm that makes at most~$\theta$ queries errs on $u\sim{} \Dd$ with probability at least~$p$, then any \emph{randomized} non-adaptive algorithm with~$\theta$ queries errs with probability at least~$p$ on some input~$u$.

To prove \cref{thm:scc-lb}, we first exhibit such a distribution $\Dd$ for $\theta = \Theta(\epslogeps)$.
We take the following steps:
\begin{enumerate}
  \item we show that with high probability, an input~$u$ sampled w.r.t.\ $\Dd$ is either in or $\eps$-far from~$L$ (\cref{lemma:far-whp}),
  \item we show that with high probability, any deterministic non-adaptive tester that makes fewer than $c \cdot \epslogeps$ queries (for a suitable constant $c$) cannot distinguish whether the instance~$u$ is positive or $\eps$-far, hence it errs with high probability (\cref{lemma:must-accept,lemma:proba-fail}),
  \item we combine the above two results to prove \cref{thm:scc-lb} via Yao's Minimax principle.
\end{enumerate}

\subsubsection{The structure of \texorpdfstring{$\MBF(\Aa)$}{MBF(A)}}
Before proving \cref{thm:scc-lb}, we show that if $\MBF(\Aa)$ is infinite, then we can find minimal blocking factors with a ``regular'' structure, which is the main ingredient of our proof.
First, we prove that the set of minimal blocking factors of an automaton is a regular language, recognized by an automaton that is possibly exponentially larger than~$\Aa$.
We start with blocking factors of the form $\timedword{i}{u}$ for a fixed $i\in\ZZ/\lambda\ZZ$.
\begin{lemma}\label{lemma:blocking-regular}
  Let $\Aa = (Q, \Sigma, \delta, q_0, F)$ be a strongly connected NFA with~$m$ states and let $\lambda = \lambda(\Aa)$.
  For every $i \in\ZZ/\lambda\ZZ$, the set of minimal blocking factors of~$\Aa$ of the form $\timedword{i}{u}$ is a regular language recognized by an NFA of size $2^{\cO(m)}$.
\end{lemma}
\begin{proof}
  Since every positional word containing a blocking factor is blocking, a positional word $\timedword{i}{u}$ is a minimal blocking factor if and only if it is blocking and the positional words obtained by deleting respectively its first and its last letter are not blocking.
  If $u$ is non-empty, these two positional words are $\timedword{i+1}{u'}$ and $\timedword{i}{u''}$, where $u'$ (resp.\ $u''$) is obtained from $u$ by deleting its first (resp.\ last) letter.
  By \cref{lemma:intuition-timed-word}, $\timedword{i}{u}$ is a minimal blocking factor if and only if no run reads $u$ from a state of $Q_i$, some run reads $u'$ from a state of $Q_{i+1}$, and some run reads $u''$ from a state of $Q_i$.

  We build a deterministic automaton that records three things while reading $u$:
  \begin{itemize}
    \item the set of states reachable by a run reading $u$ from a state of $Q_i$,
    \item the set of states reachable by a run reading $u'$ from a state of $Q_{i+1}$,
    \item the set of states reachable by a run reading $u''$ from a state of $Q_i$.
  \end{itemize}
  It also checks that the input is a well-formed positional word starting with index $i$, which only requires remembering the index of the last letter read.
  The automaton has an additional initial state $q_{init}$, since the last two components are undefined for the empty word.
  These sets are straightforward to initialise and update as we read the positional word.
  A non-empty positional word is a minimal blocking factor if and only if after reading it the first set is empty while the two others are not.
  Our set of final states is thus $\set{\emptyset} \times (2^Q \setminus \set{\emptyset}) \times (2^Q \setminus \set{\emptyset})$.
\end{proof}

It follows that the set of minimal blocking factors of~$\Aa$ is also a regular language.
\begin{corollary}
  Let~$\Aa$ be a strongly connected NFA with~$m$ states.
  The set of minimal blocking factors of~$\Aa$ is a regular language recognized by an NFA of size $2^{\cO(m)}$.
\end{corollary}

Therefore, if $\MBF(\Aa)$ is infinite, we can use a pumping argument to find an infinite family of blocking factors with a shared structure $\{x y^r z, r\in\NN\}$, for some non-empty positional words $x, y$ and~$z$. We will use this property later, when proving a lower bound against the language of automata with infinitely many blocking factors.

\begin{restatable}{lemma}{lemmagoodbf}\label{lemma:good-bf}
  If $\MBF(\Aa)$ is infinite, then there exist a constant $ S = 2^{\poly(m)}$ and positional words $x,y_+,y_-,z$ such that:
  \begin{enumerate}
    \item the words $y_+$ and $y_-$ have the same length,
    \item $|x|, |y_+|, |y_-|, |z|\le  S$,
    \item there exist an index $i_*\in\ZZ/ \lambda\ZZ$ and a state $q_*\in Q_{i_*}$ such that for every integer $r \ge 1$,
      the positional word $\tau_{-,r} = x(y_-)^r z$ starts with index $i_*$ and is blocking for~$\Aa$, and for every $s < r$, the positional word $\tau_{+,r,s} = x(y_-)^sy_+(y_-)^{r-1-s}z$ is well-formed, starts with index $i_*$, and satisfies
      \[q_* \xrightarrow{\tau_{+,r,s}} q_*.\]
      In particular, $\tau_{+,r,s}$ is not blocking for~$\Aa$.
  \end{enumerate}
\end{restatable}

Note that here, the state $q_*$ is the same for \emph{every} integer $r, s$.
Since $\tau_{+,r,s}$ labels a cycle, its length is a multiple of $\lambda$, and so is the length of $\tau_{-,r}$.
As increasing $S$ does not affect the statement, we can also assume that $S$ is a multiple of $\lambda$.

\subparagraph*{Running example: repeated parity.}
For the language $L_0$, we can take $x=c$, $y_-=a$, $y_+=b$ and $z=bd$ in \cref{lemma:good-bf}.
The word $\tau_{-,r}=ca^rbd$ is blocking, since the factor between $c$ and $d$ contains a single $b$.
If we replace the $(s+1)$-th $a$ by a $b$ for some $s < r$, we obtain $\tau_{+,r,s}=ca^sba^{r-1-s}bd$, which contains two $b$'s between $c$ and $d$ and is therefore not blocking.
The state $q_*$ is the state of the automaton outside of the factors $c\cdots d$, and every word $\tau_{+,r,s}$ labels a loop on $q_*$.

This example also explains the conditions of \cref{lemma:good-bf}.
As $y_-$ and $y_+$ have the same length, replacing $y_-$ by $y_+$ does not shift the positions of the rest of the word, and as all the words $\tau_{+,r,s}$ label a loop on the same state $q_*$, we can concatenate them and still obtain a word of the language.

\subparagraph*{Idea of the proof.}
Since the set of minimal blocking factors is an infinite regular language, it contains a family of words $\tau\mu^k\eta$ for $k \geq 0$, with $\tau$ and $\eta$ non-empty.
As $\tau\mu^k\eta$ is a minimal blocking factor, its proper factors $\tau\mu^k$ and $\mu^k\eta$ are not blocking, so by \cref{lemma:intuition-timed-word} each of them labels a path in the automaton, and these two paths may be different.
By taking a large enough power of $\mu$, we obtain a word $y_-$ which labels a loop on the first path and a loop on the second path.
We then use strong connectivity and the periodicity of the automaton to find a word $y_+$ of the same length as $y_-$ that goes from the first loop to the second one.
Replacing any of the $r$ copies of $y_-$ by $y_+$ thus gives a path that follows the first path, then the second one.
Finally, we add words $x$ and $z$ so that all these paths start and end in the same state $q_*$.

\begin{proof}[Proof of \cref{lemma:good-bf}]
  Since $\MBF(\Aa)$ is infinite, there exists an index
  $i_* \in \mathbb{Z}/\lambda\mathbb{Z}$ such that
  the language of minimal blocking factors of the form $\timedword{i_*}{u}$ is infinite.
  By Lemma~\ref{lemma:blocking-regular}, this language is recognized by an NFA
  with $T = 2^{\cO(m)}$ states.

  Since this language is infinite, there exist positional
  words $\tau,\mu,\eta$, with $\mu$ nonempty, such that
  $\tau\mu^k\eta \in \MBF(\Aa)$ for every $k\geq 0$,
  and $|\tau|,|\mu|,|\eta|\leq T$.
  Indeed, take a shortest accepting run containing a cycle
  in the NFA recognizing these minimal blocking factors,
  and decompose it into a prefix (labeled $\tau$), a nonempty simple cycle (labeled $\mu$),
  and a suffix (labeled $\eta$). By minimality, each of those parts visits each state at most once, except for the endpoints of the cycle, thus all three parts have length at most $T$.

  Since $\tau\eta$ and $\tau\mu\eta$ are both positional words starting with index $i_*$,
  the length of $\mu$ is divisible by $\lambda$.
  Moreover, we may assume that $\tau$ and $\eta$ are
  nonempty: otherwise, we replace them by $\tau\mu$
  and $\mu\eta$, respectively, without changing the
  property above. Their lengths are then at most $2T$.

  Consider the minimal blocking factor
  $\tau\mu^m\eta$.
  By minimality, both $\tau\mu^m$ and $\mu^m\eta$
  are non-blocking, since they are proper factors
  of $\tau\mu^m\eta$.

  By Lemma~\ref{lemma:intuition-timed-word},  $\tau\mu^m$
  labels a path in $\Aa$ starting from a state in
  $Q_{i_*}$. Among the $m+1$ states visited
  immediately after reading
  $\tau,\ \tau\mu,\ \ldots,\ \tau\mu^m$
  two must coincide. Hence there exist integers
  $0 \leq k_0 \leq m$ and $1 \leq k_1 \leq m$, and states $p \in Q_{i_*}$ and $p_1$ such that
  $p \xrightarrow{\tau\mu^{k_0}} p_1
  \xrightarrow{\mu^{k_1}} p_1$.
  Similarly, since $\mu^m\eta$ is non-blocking, it labels a path in $\Aa$, and among the $m+1$ states visited before reading $\mu^m\eta, \mu^{m-1}\eta, \ldots, \eta$ along this path, two must coincide.
  Hence there exist integers $1 \leq k_2 \leq m$ and $0 \leq k_3 \leq m$ and states $p_2,q$ such that
  $p_2 \xrightarrow{\mu^{k_2}} p_2
  \xrightarrow{\mu^{k_3}\eta} q$.
  Both $p_1$ and $p_2$ belong to the periodicity class $Q_{i_* + |\tau|}$, since $|\mu|$ is divisible
  by $\lambda$ (\cref{corollary:position-length}).

  Let $\rho$ be the reachability constant of $\Aa$
  from Fact~\ref{fact:periodicity}, and set
  $
  K=\rho k_1k_2
  $.
  Since $K$ is divisible by both $k_1$ and $k_2$,
  we get
  $p_1\xrightarrow{\mu^K}p_1$
  and $p_2\xrightarrow{\mu^K}p_2$.
  Moreover, $K|\mu|\geq\rho$ and
  $K|\mu|\equiv 0\pmod\lambda$.
  As $p_1$ and $p_2$ belong to the same periodicity
  class, Fact~\ref{fact:periodicity} yields a word $y_+$ of length
  $K|\mu|$ such that $p_1\xrightarrow{y_+}p_2$.

  Since the automaton $\Aa$ is strongly connected, there is a path from $q$ to $p$ of length at most $m$, labeled by a word $\eta'$.
  Define
  $ x=\tau\mu^{k_0}$,
  $y_-=\mu^K$,
  $z=\mu^{k_3}\eta \eta'$ and $q_* =p$.
  By construction, $|y_+|=|y_-|$, and $x$ starts with index $i_*$.
  As $|y_-| = K|\mu|$ is a multiple of $\lambda$, every copy of $y_-$ in $x(y_-)^r z$ starts with the same index $i_* + |x|$, and we see $y_+$ as the positional word that starts with this index. Hence $x(y_-)^sy_+(y_-)^{r-1-s}z$ is a well-formed positional word that starts with index $i_*$.
  For every $r\geq 1$, the word
  $x(y_-)^rz
  =\tau\mu^{k_0+Kr+k_3}\eta\eta'$
  contains the minimal blocking factor $\tau\mu^{k_0+Kr+k_3}\eta$, hence it is blocking.

  On the other hand, for every $r\geq 1$ and
  $0\leq s<r$, the paths constructed above yield
  \[
    q_* = p\xrightarrow{x}p_1
    \xrightarrow{(y_-)^s}p_1
    \xrightarrow{y_+}p_2
    \xrightarrow{(y_-)^{r-1-s}}p_2
    \xrightarrow{z} p = q_*.
  \]
  Thus $x(y_-)^sy_+(y_-)^{r-1-s}z$
  is non-blocking by Lemma~\ref{lemma:intuition-timed-word}.

  Finally, all four words have length at most
  $2^{\poly(m)}$, since $T=2^{\cO(m)}$,
  $k_0,k_1,k_2,k_3\leq m$, and $\rho\leq 3m^2$.
\end{proof}

\subsubsection{Constructing a hard distribution \texorpdfstring{$\Dd$}{D}}

\subparagraph*{Intuition.}
Before giving the formal construction, we explain the choice of parameters, ignoring rounding and constant factors that depend on the automaton.
The construction relies on the words $\tau_{-,r}$ and $\tau_{+,r,s}$ given by \cref{lemma:good-bf}, which have the same length and differ in at most $S$ letters, while the first one is blocking and the second one is not.

We split the input into $\Theta(\eps n)$ intervals of length $\Theta(1/\eps)$.
Each interval receives a \emph{scale} $t$, which is either $0$ or one of $K = \Theta(\log(1/\eps))$ positive values.
An interval of positive scale $t$ is filled with $N_t = \Theta(1/(S\eps 2^t))$ copies of $\tau_{-,r}$ in a negative instance, and of $\tau_{+,r,s}$ in a positive instance, where $r = 2^t$ and $s < r$ is chosen uniformly at random and is the same for all copies in the interval.
These copies fit in the interval, as each of them has length $\cO(Sr)$.
An interval of scale $0$ is filled with a word labelling a loop on $q_*$, and is the same in positive and negative instances.
For the language $L_0$, an interval of positive scale $t$ thus contains the word $(ca^{2^t}bd)^{N_t}$ in a negative instance, and the same word in which the $(s+1)$-th $a$ of each copy is replaced by $b$ in a positive instance.

The two parameters $N_t$ and $r$ play different roles.
The number $N_t$ of copies controls the distance to the language, since in a negative instance each copy of $\tau_{-,r}$ requires its own substitution.
The number $r$ controls the number of queries needed, since the tester has to query the $(s+1)$-th copy of $y_-$ or $y_+$ to distinguish $\tau_{-,r}$ from $\tau_{+,r,s}$.
As $s$ is the same in all copies of the interval, a tester making fewer than $r$ queries in this interval leaves at least one value of $s$ unqueried in all copies, so what it sees is compatible with a positive instance.

We choose the probability $p_t$ of each positive scale so that all scales contribute equally to the expected number of copies of $\tau_{-,r}$, that is, $p_tN_t = \Theta(1/K)$, which gives $p_t = \Theta(S\eps 2^t/K)$.
Small scales are therefore rare, and the total probability of a positive scale is $\Theta(1/K)$, so that most intervals have scale $0$.
Summing over all scales, an interval contains a constant number of copies of $\tau_{-,r}$ in expectation.
With suitable constants, a negative instance thus contains more than $\eps n$ disjoint blocking factors with high probability, hence it is $\eps$-far from the language.

A non-adaptive tester chooses the number of queries that it makes in each interval without knowing the scale of the interval, and this is where the logarithmic factor comes from.
If a tester makes $z$ queries in an interval, then it can only gain information on this interval if its scale $t$ is positive and $2^t \le z$.
This happens with probability
\[
  \sum_{1 \le t \le K,\ 2^t \le z} p_t = \cO\left(\frac{S\eps z}{K}\right).
\]
Summing over all intervals, a tester making $Q$ queries in total distinguishes positive and negative instances with probability $\cO(S\eps Q/K)$, which is small unless $Q = \Omega(K/(S\eps)) = \Omega(\epslogeps)$.

\subparagraph*{Formal construction.}
In what follows, $m$ denotes the number of states of~$\Aa$, $\lambda$ its period and $\rho$ its reachability constant, and we fix $S$, $x$, $y_+$, $y_-$, $z$, $i_*$ and $q_*$ as given by \cref{lemma:good-bf}, with $S$ a multiple of $\lambda$.

We first explain why we can assume that $\eps$ is of the form $2^{-K}/S$ for some integer $K$.
Let $\eps > 0$, let $K = \lfloor \log((S\eps)^{-1}) \rfloor$ and $\eps' = 2^{-K}/S$, so that $\eps \le \eps' < 2\eps$.
Every word that is $\eps'$-far from a language is also $\eps$-far from it, hence every $\eps$-property tester is also an $\eps'$-property tester.
As $\log(1/\eps') /\eps' = \Omega(\epslogeps)$, a lower bound of $\Omega(\log(1/\eps')/\eps')$ on the query complexity of $\eps'$-property testers implies \cref{thm:scc-lb} for $\eps$.
From now on, we assume that $\eps = 2^{-K}/S$ for some integer $K \geq 20$, and that $\eps \le 1/\rho$; this holds for all small enough~$\eps$.

Let $n$ be a large enough integer such that $\lang{\Aa}$ contains a word of length~$n$.
We now describe how to sample an input $\mu = \twu$ of length~$n$ w.r.t.\ $\Dd$.

Let $\ell = 4/\eps = 4S2^K$; it is an integer and a multiple of $\lambda$.
We decompose $[0\dd n-1]$ into an initial segment of length $a$, then $k$ intervals of length $\ell$, and a final segment of length $R$, where:
\begin{itemize}
  \item $a$ is the unique integer such that $\rho \le a < \rho+\lambda$ and $a = i_* \pmod{\lambda}$,
  \item $k = \lfloor (n - a - \rho - \lambda)/\ell\rfloor$,
  \item $R = n - a - k\ell$, so that $\rho + \lambda \le R < \rho+\lambda+\ell$.
\end{itemize}
Note that $\eps n/5 \le k \le \eps n/4$ for $n$ large enough.
For $j=1,\ldots, k$, let $a_j$ and $b_j$ denote the endpoints of the $j$-th interval.

Let $\pi$ be a uniformly random bit. If $\pi = 1$, we will construct a positive instance $\mu \in \timedlang{\Aa}$, and otherwise the instance will be $\eps$-far from $\timedlang{\Aa}$ with high probability.
For each interval, we sample independently at random a variable $\kappa_j$ with the following distribution:
\begin{equation}\label{eq:dist-kappa}
  \kappa_j =
  \begin{cases}
    t, &\text{ with prob. } p_t = 10\cdot 2^t S\eps/K = 10 \cdot 2^{t-K}/K \text{ for } t = 1,2,\ldots, K,\\
    0,&\text{ with prob. } p_0 = 1 - \sum_{t=1}^{K} p_t.
  \end{cases}
\end{equation}

The event $\kappa_j = 0$ means that the $j$-th interval contains no specific information:
we fix a positional word $\eta_*$ of length $\ell$ such that $q_* \xrightarrow{\eta_*} q_*$, which exists by \cref{fact:periodicity} as $\ell \geq \rho$ and $\ell = 0 \pmod\lambda$, and if $\kappa_j = 0$, we set $\mu[a_j\dd b_j] = \eta_*$, regardless of the value of $\pi$.

The event $\kappa_j = t > 0$ means that the $j$-th interval contains $N_t = 2^{K-t}$ copies of the same word, which is $\tau_{-,r}$ if $\pi = 0$, and $\tau_{+,r,s}$ for a uniformly random $s$ if $\pi = 1$, where $r = 2^t$.
Formally, as $|\tau_{-,r}| \le (r+2)S \le 3rS$, the length of $(\tau_{-,r})^{N_t}$ is at most $3S 2^K = 3/\eps$.
We fix a positional word $\eta_t$ of length $\ell - N_t|\tau_{-,r}|$ such that $q_* \xrightarrow{\eta_t} q_*$: such a word exists by \cref{fact:periodicity}, as its length is at least $1/\eps \ge \rho$ and is a multiple of $\lambda$.
We construct the $j$-th interval as follows:
\begin{itemize}
  \item if $\pi = 0$, we set $\mu[a_j\dd b_j] = (\tau_{-,r})^{N_t}\eta_t$,
  \item if $\pi = 1$, we select $s\in [0\dd  r-1]$ uniformly at random, and set $\mu[a_j\dd b_j] = (\tau_{+,r,s})^{N_t}\eta_t$.
\end{itemize}
Finally, the initial segment $\mu_i$ is a positional word of length $a$ labelling a path from $q_0$ to $q_*$, and the final segment $\mu_f$ is a positional word of length $R$ labelling a path from $q_*$ to a final state.
These words exist by \cref{fact:periodicity}: as $\lang{\Aa}$ contains a word of length $n$, some final state $q_f$ belongs to $Q_n$; we have $a, R \geq \rho$, $a = i_* \pmod \lambda$ and $R = n - i_* \pmod\lambda$.
The initial and final segments do not depend on $\pi$ and on the $\kappa_j$.

As all intervals have a length multiple of $\lambda$, every copy of $\tau_{-,r}$ or $\tau_{+,r,s}$ in $\mu$ starts at a position equal to $i_*$ modulo $\lambda$.

\subsubsection{Properties of the distribution \texorpdfstring{$\Dd$}{D}}
Next, we establish in three steps that the distribution $\Dd$ has the desired properties.
We first show that negative instances are $\eps$-far from the language with high probability, by counting disjoint blocking factors (\cref{lemma:far-whp}).
We then show that if a tester makes too few queries in every interval of positive scale, then what it has seen can be completed into a positive instance, so by perfect completeness it must accept (\cref{lemma:must-accept}).
Finally, we bound the probability that the tester makes enough queries in some interval of positive scale (\cref{lemma:proba-fail}).

\begin{observation}\label{obs:well-defined}
  The distribution $\Dd$ is well-defined, i.e. for every~$t$ between~$0$ and $K$, we have $0 \le p_t \le 1$.
\end{observation}
\begin{proof}
  For $t \geq 1$, we have $0 \le p_t \le 10/K \le 1$, and
  \[\sum_{t=1}^K p_t = \frac{10}{K}\sum_{t=1}^K 2^{t-K} < \frac{20}{K} \le 1,\]
  as $K \ge 20$. Hence $0 \le p_0 \le 1$.
\end{proof}

\begin{observation}
  If $\pi = 1$, then $\mu\in \timedlang{\Aa}$.
\end{observation}
\begin{proof}
  The initial segment labels a path from $q_0$ to $q_*$, every interval labels a loop on $q_*$, and the final segment labels a path from $q_*$ to a final state.
\end{proof}

Define the event $\Ff = \{ \mu$ is $\eps$-far from $\timedlang{\Aa} \}$.

\begin{lemma}\label{lemma:far-whp}
  Conditioned on $\pi = 0$, the probability of the event $\Ff$ goes to~$1$ as~$n$ goes to infinity.
\end{lemma}
\begin{proof}
  When $\pi = 0$, the copies of $\tau_{-,r}$ in intervals where $\kappa_j >0$ form a set of disjoint occurrences of blocking factors in $\mu$.
  Any word of $\timedlang{\Aa}$ of length $n$ differs from $\mu$ in at least one position in each of these occurrences.
  Consequently, the distance between $\mu$ and $\timedlang{\Aa}$ is at least the number of these copies. To prove the lemma, we show that this number is at least $\eps n$ with high probability, by showing that it is larger than $\eps n$ by a constant factor in expectation and using a concentration argument.

  Let $B_j$ denote the number of copies of $\tau_{-,r}$ in the $j$-th interval: it is equal to $N_{\kappa_j} = 2^{K - \kappa_j}$ when $\kappa_j>0$ and to~$0$ otherwise.

  \begin{claim}\label{claim:expe}
    Let $B = \sum_{j=1}^{k} B_j$, and let $E = \EE\left[B\right]$.
    We have $E \ge 2\eps n$.
  \end{claim}
  \begin{proof}[Claim proof]
    By direct calculation:
    \[
      E = \sum_{j=1}^{k}\EE\left[ B_j\right]= \sum_{j=1}^{k}\sum_{t= 1}^{K} 2^{K-t} \cdot \frac{10\cdot 2^{t-K}}{K}
      = \sum_{j=1}^{k}\sum_{t= 1}^{K} \frac{10}{K}
      = 10k \ge 2\eps n,
    \]
    as $k \geq \eps n/5$.
  \end{proof}

  We will now show that $\PP(B <\eps n)$ goes to~$0$ as~$n$ goes to infinity. We use Hoeffding's inequality, which we recall here:
  \begin{fact}[{\cite[Theorem 2]{hoeffding1994probability}}]\label{fact:hoeffding}
    Let $X_1,\ldots, X_k$ be independent random variables such that for every $i = 1,\ldots, k$, we have $a_i \le X_i \le b_i$, and let $X= \sum_{i=1}^k X_i$. Then, for any $t> 0$, we have \[\Prob{\EE[X] -X \ge t} \le \exp\left(-\frac{2t^2}{\sum_{i=1}^k (b_i-a_i)^2}\right).\]
  \end{fact}
  By \cref{claim:expe}, we have $B < \eps n \Rightarrow E-B\ge \eps n$, and therefore $\PP(B <\eps n) \le \PP(E-B\ge \eps n)$.
  The random variable~$B$ is the sum of~$k$ independent random variables, each taking values between~$0$ and $2^K = 1/(S\eps)$.
  Therefore, by Hoeffding's inequality (\cref{fact:hoeffding}), we have
  \[
    \PP(E - B \ge \eps n)
    \le \exp\left(-\frac{2\eps^2n^2}{k/( S\eps)^2}\right)
    \le \exp\left(-\frac{2 S^2\eps^4n^2}{\eps n/4}\right)
    = \exp\left(-8 S^2\eps^3n\right),
  \]
  using $k \le \eps n/4$.
  This probability goes to~$0$ as~$n$ goes to infinity, which concludes the proof.
\end{proof}

As $\pi = 0$ with probability $1/2$, we obtain the following.
\begin{corollary}\label{coro:large-fail-proba}
  For large enough~$n$, we have $\Prob{\Ff} \ge 11/24$.
\end{corollary}

Intuitively, our distribution is hard to test because positive and negative instances are very similar, so a tester with few queries will likely not be able to tell them apart, and the perfect completeness constraint then forces it to accept.
We now establish this formally.
Recall that a deterministic non-adaptive tester queries a fixed set of positions, which only depends on~$n$.

\begin{lemma}\label{lemma:must-accept}
  Let~$T$ be a deterministic non-adaptive tester with perfect completeness (i.e. it always accepts words of $\timedlang{\Aa}$) and let $z_j$ denote the number of queries that it makes in the $j$-th interval.
  Conditioned on the event $\Mm = \{\forall j \text{ s.t. }\kappa_j > 0, z_j < 2^{\kappa_j}\}$, the probability that~$T$ accepts~$\mu \sim{}\Dd$ is~$1$.
\end{lemma}
\begin{proof}
  Let $\mu$ be a sampled word for which the event $\Mm$ holds.
  If $\mu\in \timedlang{\Aa}$, then $T$ accepts $\mu$ by perfect completeness, so we assume that $\mu \notin \timedlang{\Aa}$, which implies that $\mu$ was generated with $\pi = 0$.

  Consider an interval $j$ such that $\kappa_j > 0$, and let $r=2^{\kappa_j}$. This
  interval contains $N_{\kappa_j}$ copies of $\tau_{-,r}=x(y_-)^rz$, followed by
  padding. We number the copies of $y_-$ inside each copy of $\tau_{-,r}$ by $0,\ldots,r-1$.
  Let $H_j$ be the set of numbers $s$ such that the
  tester queries at least one position in the copy of $y_-$ numbered $s$ in at least one of the copies of $\tau_{-,r}$.
  Every query contributes at most one number to $H_j$, and therefore
  $|H_j|\le z_j<r$.
  We choose $s_j\in\{0,\ldots,r-1\}\setminus H_j$ and, in every copy of $\tau_{-,r}$ in this
  interval, we replace the copy of $y_-$ numbered $s_j$ by $y_+$.
  No queried letter changes, and the resulting interval is $(\tau_{+,r,s_j})^{N_{\kappa_j}}\eta_{\kappa_j}$.

  We apply this operation in every interval of positive scale, leaving the other
  intervals and the initial and final segments unchanged.
  The resulting word $\mu'$ is a possible positive instance, hence it belongs to $\timedlang{\Aa}$, and every position queried by $T$ has the same letter in $\mu$ and in $\mu'$.
  Hence $T$ has the same output on $\mu$ and on $\mu'$, and since perfect completeness forces $T$ to accept $\mu'$, it also accepts $\mu$.
\end{proof}

It remains to bound the probability that some interval of positive scale receives
at least as many queries as its number of possible values of $s$ (the event $\Mm$).
\begin{lemma}\label{lemma:proba-fail}
  Let~$T$ be a deterministic non-adaptive tester, let $z_j$ denote the number of queries that it makes in the $j$-th interval, and assume that~$T$ makes at most $\frac{1}{240}\cdot K/(S\eps)$ queries.
  The probability of the event $\Mm = \{\forall j \text{ s.t. }\kappa_j > 0, z_j < 2^{\kappa_j}\}$ is at least $11/12$.
\end{lemma}

\begin{proof}
  Since $T$ is non-adaptive, the numbers $z_j$ do not depend on $\mu$, and in particular they are independent of the variables $\kappa_j$.
  We show that the probability of $\overline{\Mm}$, the complement of $\Mm$, is at most $1/12$.
  By the union bound, we have:
  \[
    \Prob{\overline{\Mm}} \le \sum_{j} \Prob{\kappa_j > 0 \land 2^{\kappa_j} \le z_j}
    = \sum_{j} \sum_{1 \le t \le K,\ 2^t \le z_j} \frac{10\cdot 2^t S \eps}{K}
    \le \frac{10 S \eps}{K} \sum_{j} 2z_j,
  \]
  where the last inequality upper bounds the sum of the powers of $2$ that are at most $z_j$ by $2z_j$.
  As $\sum_j z_j \le \frac{1}{240}\cdot K/(S\eps)$, we obtain $\Prob{\overline{\Mm}} \le 20/240 = 1/12$.
\end{proof}

We are now ready to prove \cref{thm:scc-lb}.
\begin{proof}[Proof of \cref{thm:scc-lb}]
  As explained in the construction of $\Dd$, we can assume that $\eps = 2^{-K}/S$ with $K \geq 20$ and $\eps \le 1/\rho$.
  Let $n$ be a large enough integer such that $\lang{\Aa}$ contains a word of length~$n$.
  We show that any non-adaptive tester with perfect completeness for $\lang{\Aa}$ makes at least $\frac{1}{240}\cdot K/(S\eps)$ queries on words of length $n$, which is $\Omega(\epslogeps)$ since $K = \log((S\eps)^{-1})$ and $S$ is a constant.
  By Yao's Minimax principle, it suffices to show that any deterministic non-adaptive algorithm~$T$ with perfect completeness that makes fewer queries errs with probability more than $1/3$ when $\twu\sim{}\Dd$.
  (A randomized non-adaptive tester with perfect completeness is a probability distribution over deterministic non-adaptive testers, almost all of which accept every word of $\lang{\Aa}$ of length $n$, as there are finitely many such words.)

  Consider such an algorithm~$T$.
  The probability that~$T$ makes an error on~$u$ is at least the probability that~$\twu$ is $\eps$-far from $\timedlang{\Aa}$ and~$T$ accepts, which in turn is at least the probability of $\Mm\cap \Ff$ by Lemma~\ref{lemma:must-accept}.
  By \cref{coro:large-fail-proba}, we have $\Prob{\Ff}\ge 11/24$, and by \cref{lemma:proba-fail}, we have $\Prob{\overline{\Mm}} \le 1/12$.
  Therefore, by the union bound applied to the complements of $\Ff$ and $\Mm$, we have
  \[ \Prob{T \text{ errs}} \ge \Prob{\Mm\cap \Ff} \ge 1 - \Prob{\overline{\Ff}} - \Prob{\overline{\Mm}} \ge 1 - 13/24 - 2/24 = 9/24 > 1/3.\]
  This concludes the proof of \cref{thm:scc-lb}, and consequently of \cref{thm:scc}.
\end{proof}

\subparagraph*{Adaptive testers.}
The non-adaptivity assumption is only used in the proof of \cref{lemma:proba-fail}, which requires that the number of queries that the tester makes in an interval does not depend on the scale of this interval.
It cannot be removed, as the following proposition shows that the language $L_0$, which has infinitely many minimal blocking factors, can be tested with $\cO(1/\eps)$ queries by an adaptive tester, so the trichotomy that we establish in this \doc does not hold for adaptive testers.

\begin{proposition}\label{prop:adaptive}
  There is an adaptive $\eps$-property tester for $L_0 = (a+b+cPd)^*$ with perfect completeness that makes $\cO(1/\eps)$ queries, that is, a tester in the sense of \cref{def:tester}, except that the positions it queries may depend on the answers to previous queries.
\end{proposition}
\begin{proof}
  The minimal automaton of $L_0$ has three states: a state $q_*$ outside of the factors of the form $cvd$, with loops labelled by $a$ and $b$, and two states for the parity of the number of $b$'s read since the last $c$.
  Its period is $1$, so positional words are simply words.
  We call the letters $c$ and $d$ \emph{delimiters}; a direct inspection of this automaton shows that its minimal blocking factors are the words $cvc$, $dvd$ and $cv'd$, where $v, v' \in \set{a,b}^*$ and $v' \notin P$.
  In particular, every minimal blocking factor starts and ends with a delimiter, and contains no other delimiter.

  Let $u$ be a word of length $n$.
  If $n < 108/\eps$, the tester queries all of $u$, and otherwise it repeats the following round, until it has made $324/\eps$ queries in total: it picks a uniformly random position $i$ and queries $u[i]$; if $u[i]$ is a delimiter, it queries $u[i+1], u[i+2],\ldots$ until it finds the next delimiter $u[j]$ (or reaches the end of $u$), and rejects if $u[i\dd j]$ is blocking.
  If it has not rejected after $324/\eps$ queries, it accepts.
  This tester only rejects when it has found a blocking factor, hence it has perfect completeness.

  Assume that $n \ge 108/\eps$ and that $u$ is $\eps$-far from $L_0$.
  As $L_0$ contains words of every length, by \cref{lemma:many-blocking} (with $m = 3$), $u$ contains at least $\eps n/54$ disjoint occurrences of blocking factors.
  Each of them contains a minimal blocking factor, which is a factor $u[i\dd j]$ where $i < j$ are two consecutive positions of delimiters in $u$.
  Since the occurrences are disjoint, these positions $i$ are pairwise distinct, so each round rejects with probability at least $\eps/54$.
  On the other hand, for every delimiter at position $i$, let $g(i)$ be the distance from $i$ to the next delimiter (or to the end of the word).
  As the sum of the $g(i)$ is at most $n$, the expected number of queries made by a round is at most $1 + \frac{1}{n}\sum_i g(i) \le 2$.
  By Wald's identity, the expected number of queries made before the first rejecting round ends is at most $2 \cdot 54/\eps = 108/\eps$.
  By Markov's inequality, the probability that this number exceeds $324/\eps$ is at most $1/3$, hence the tester rejects $u$ with probability at least $2/3$.
\end{proof}

\section{Characterisation of hard languages for all NFAs}
\label{sec:general}
In this section we extend the results of the previous section to all finite automata.
This extension is based on a generalization of blocking factors: we introduce \emph{blocking sequences}, which are sequences of factors that witness the fact that we cannot take any path through the strongly connected components of the automaton.
For the lower bound, we define a suitable preorder on blocking sequences, which extends the factor relation on words to those sequences, and allows us to define \emph{minimal} blocking sequences.

\subsection{Blocking sequences}\label{sec:blocking-seq}

\subsubsection{Examples motivating blocking sequences}

Before presenting the technical part of the proof, let us go through two examples, which motivate the notions that we introduce and illustrate some of the main difficulties.
For simplicity, we will use words instead of positional words since they suffice to describe the high-level ideas. Also, note that in the first example all SCCs have period one anyway.

\begin{example}
  Consider the automaton $\Aa_1$ depicted in \cref{fig:aut-ab-bc}: it recognizes the language $L_1$ of words in which all $c$'s appear before the first $b$, over the alphabet $\set{a,b,c}$.

  \begin{figure}[htbp]
    \begin{center}
      \begin{tikzpicture}[shorten >=1pt, node distance=3cm, on grid, auto,
        every initial by arrow/.style={thick}, every accepting by arrow/.style={thick}, >={stealth'}]
        \tikzset{
          state/.style={
            circle,
            thick,
            draw,
            minimum size=8mm,
            inner sep=2pt
          },
          accepting/.style={
            state,
            accepting by arrow,
            accepting below
          }
        }

        \node[state, initial , accepting] (q0) {$q_0$};
        \node[state, accepting] (q1) [right=of q0] {$q_1$};
        \node[state] (q2) [right=of q1] {$q_2$};

        \path[->, thick, >=stealth']
        (q0) edge[above] node {\textcolor{red}{$b$}} (q1)
        (q1) edge[above] node {\textcolor{red}{$c$}} (q2)
        (q0) edge[loop above] node {$a, c$} ()
        (q1) edge[loop above] node {$a, b$} ()
        (q2) edge[loop above] node {$a, b, c$} ();

      \end{tikzpicture}
      \caption{An automaton $\Aa_1$ that recognizes the language $L_1 = (a+c)^* (a+b)^*$.}
      \label{fig:aut-ab-bc}
    \end{center}
  \end{figure}

  The set of minimal blocking factors of $\Aa_1$ is infinite: it is the language $ba^*c$.
  Yet, $L_1$ is easy to test: we sample $\cO(1 / \eps)$ letters at random, answer ``no'' if the sample contains a~$c$ occurring after a~$b$, and ``yes'' otherwise.
  To prove that this yields a property tester, we rely on the following property:
  \begin{property}
    If $u$ is $\eps$-far from $L_1$, then it can be decomposed into $u = u_1 u_2$ where $u_1$ contains $\Omega(\eps n)$ letters $b$ and $u_2$ contains $\Omega(\eps n)$ letters $c$.
  \end{property}

  The pair of factors $(b, c)$ is an example of a sequence of factors that cannot occur in a word of $L_1$: a word that contains an occurrence of the first followed by an occurrence of the second cannot be in $L_1$.
  We can also show that a word $\eps$-far from $L_1$ must contain many $b$ and then many $c$, and this property (\cref{lem:far-from-path-implies-many-blocking}) is the basis of the algorithm for general regular languages.

  This example shows that blocking factors are not enough, and that we need to consider sequences of factors, which leads to the notion of \emph{blocking sequences}.
  Intuitively, a blocking sequence for a language $L$ is a sequence $\sigma = (v_1,\ldots, v_k)$ of (positional) words such that if each word of the sequence appears in $u$, with the occurrences of the $v_i$'s ordered as in $\sigma$, then $u$ is not in the language\footnote{This is not quite the definition, but it conveys the right intuition. The difference is that in the actual definition it may be that to witness non-membership we need each of those words to appear a few times.}.
  While $L_1$ has infinitely many minimal blocking factors, it has a single minimal blocking sequence $\sigma = (b, c)$.
\end{example}

Notice that the (unique) blocking sequence $(b,c)$ can be visualized on \cref{fig:aut-ab-bc}:
it is composed of the red letters that label transitions between the different SCCs. This is not a coincidence, since in many simple cases blocking sequences are exactly sequences that contain a blocking factor for each SCC.
This fact could lead one to believe that the set of minimal blocking sequences is exactly the set of sequences of minimal blocking factors, one for each SCC. In particular, this would imply that as soon as one SCC has infinitely many minimal blocking factors, the language of the whole automaton is hard to test.
We show in the next example that this is not always the case, because SCCs might share minimal blocking factors.
\begin{example}\label{example:hard-to-easy}
  Consider the automaton in \cref{fig:series}: it has two SCCs and a sink state.
  The sink state $q_3$ is not co-reachable, so we ignore it: the two SCCs of interest are $\set{q_0,q_1}$ and $\set{q_2}$.
  The minimal blocking factors of the first SCC are given by $B_1 = be^*c + a$, and $B_2 = \set{a}$ for the second SCC.
  This automaton is easy to test.
  Indeed, the run of the automaton on a word without $a$ never blocks, and it ends in $q_0$, $q_1$ or $q_2$; if it ends in $q_0$, replacing the last letter by a $b$ leads to an accepting state.
  Hence every word $u$ is at distance at most $k+1$ from the language, where $k$ is the number of $a$'s in $u$: we replace all $a$'s by $b$'s, and then possibly the last letter.
  As a consequence, a word that is $\varepsilon$-far from this language has to contain many~$a$'s.
  Conversely, an accepting run reads at most one $a$, since $a$ is a blocking factor for both SCCs, so a word with $k \geq 1$ letters $a$ is at distance at least $k-1$ from the language.
  Hence property testing for this language comes down to looking for two letters $a$, although the set of blocking factors of the first component is infinite.
  \begin{figure}[htbp]
    \begin{center}
      \begin{tikzpicture}[shorten >=1pt, node distance=3cm, on grid, auto,
        every initial by arrow/.style={thick}, every accepting by arrow/.style={thick}, >={stealth'}]
        \tikzset{
          state/.style={
            circle,
            thick,
            draw,
            minimum size=8mm,
            inner sep=2pt
          },
          accepting/.style={
            state,
            accepting by arrow
          }
        }

        \node[state , accepting, accepting below] (q1) {$q_1$};
        \node[state , initial, above=of q1] (q0) {$q_0$};
        \node[state, accepting, accepting below] (q2) [right=of q1] {$q_2$};
        \node[state] (q3) [right=of q2] {$q_3$};

        \path[->, thick, >=stealth']
        (q1) edge[above] node {\textcolor{red}{$a,c$}} (q2)
        (q2) edge[above] node {\textcolor{red}{$a$}} (q3)
        (q0) edge[above right] node {\textcolor{red}{$a$}} (q2)
        (q1) edge[loop above] node {$b,e$} ()
        (q0) edge[loop above] node {$c, d, e$} ()
        (q2) edge[loop above] node {$b, c, d, e$} ()
        (q3) edge[loop above] node {$a, b, c,d,e$} ();

        \path[->, thick, >=stealth',bend right]
        (q1) edge[right] node {{$d$}} (q0)
        (q0) edge[left] node {{$b$}} (q1);

      \end{tikzpicture}
      \caption{An automaton $\Aa_2$ that recognizes the language $L_2 = Y b (b+e)^* + Y\big(a + b(b+e)^*(a+c)\big)(b+c+d+e)^*$, where $Y = (c+d+e+b(b+e)^*d)^*$ is the set of words labelling a loop on $q_0$.}
      \label{fig:series}
    \end{center}
  \end{figure}

  What matters here is that the set $B_2$ of minimal blocking factors of the second SCC is finite and is a subset of $B_1$, which makes the infinite size of $B_1$ irrelevant, because any word far from the language contains many~$a$'s.
  We will see that $\Aa_2$ has a single minimal blocking sequence, namely $(a)$, in which the factor $a$ is used to block both SCCs.
\end{example}

\subsubsection{Portals and SCC-paths}

Intuitively, blocking sequences are sequences of blocking factors of successive strongly connected components. To formalize this intuition, we use \emph{portals}, which describe how a run in the automaton interacts with a strongly connected component, and \emph{SCC-paths}, that describe a succession of portals.

In what follows, we fix an NFA $\Aa = (Q, \Sigma, \delta, q_{0}, F)$.
Let $\SCCset$ be the set of SCCs of $\Aa$.
We define the partial order relation $\leq_\Aa$ on $\SCCset$ by $S \leq_\Aa T$ if and only if $T$ is reachable from $S$.  We write $<_\Aa$ for its strict part $\leq_\Aa \setminus \geq_\Aa$. These relations extend to states through their SCC: if $s\in S$ and $t\in T$, then $s\leq_\Aa t$ if and only if $S\leq_\Aa T$.

We define $p$ as the least common multiple of the lengths of all simple cycles of $\Aa$.
Given a number $k \in \ZZ/p\ZZ$, we say that a state $t$ is \emph{$k$-reachable} from a state $s$ if there is a path from $s$ to $t$ of length $k$ modulo $p$.
In what follows, we use ``positional words'' for $p$-positional words with this value of $p$.

\begin{remark}
  In the rest of this section we will not try to optimize the constants in the formulas. They will, in fact, become quite large in some of the proofs.
  We make this choice to make the proofs more readable, although some of them are already technical.

  For instance, the choice of $p$ as the $\lcm$ of the lengths of simple cycles is not optimal: we could use, for instance, the $\lcm$ of the periods of the SCCs.
\end{remark}

\begin{definition}[Portal]
  A \emph{portal} is a 4-tuple $P = \portal{s}{x}{t}{y} \in (Q \times \ZZ/p\ZZ)^2$, such that $s$ and $t$ are in the same SCC.
  We will use portals to describe the first and last states visited by a path in an SCC, and the positions $x,y$ (modulo $p$) at which it first and last visits that SCC.
\end{definition}

The \emph{positional language} of a portal is the set \[\timedlanguage{s}{x}{t}{y} = \set{(x:w) \mid s\xrightarrow{w}t \land x+|w| = y \pmod{p}}.\]
A run from $s$ to $t$ stays in the SCC of $s$ and $t$, since a run that leaves an SCC cannot come back to it.
Note also that the language of a portal does not require the run to leave the SCC after reaching $t$: this constraint is imposed by SCC-paths, defined below.
Portals were already defined by Alon et al.~\cite{alon2001regular}, in a slightly different way.
Our definition will allow us to express blocking sequences more naturally.

\begin{definition}
  A positional word $\timedword{n}{u}$ is \emph{blocking} for a portal $P$ if it is not a factor of any word of $\lang{P}$.
  In other words, there is no path that starts in $s$ and ends in $t$, of length $y-x$ modulo $p$, which reads $u$ after $n-x$ steps modulo $p$.
\end{definition}

If the SCC of $P$ is transient, then $\lang{P}$ is either empty or reduced to the empty word, so every non-empty positional word is blocking for $P$.
The above definition matches the definition of blocking factors for strongly connected automata.

Portals describe the behavior of a run inside a single strongly connected component of the automaton.
Next, we introduce SCC-paths, which describe the interaction of a run with multiple SCCs and between two successive SCCs.
\begin{definition}[SCC-path]
  An \emph{SCC-path} $\pi$ of $\Aa$ is a non-empty sequence of portals linked by single-letter transitions
  \(\SCCpath = \portal{s_0}{x_0}{t_0}{y_0} \xrightarrow{a_1} \portal{s_1}{x_1}{t_1}{y_1} \cdots  \xrightarrow{a_{k}} \portal{s_k}{x_k}{t_k}{y_k},\)
  such that for all $i \in \set{1, \ldots, k}$,  $x_i = y_{i-1} +1 \pmod{p}$, $t_{i-1}\xrightarrow{a_i} s_{i}$, and $t_{i-1}  <_{\Aa} s_i$.
\end{definition}
Intuitively, an SCC-path is a description of the states and positions at which a path through the automaton enters and leaves each SCC.
If $k=0$, this means the path stays in the first SCC.

\begin{definition}\label{def:lang-of-sccpath}
  The language $\lang{\SCCpath}$ of an SCC-path $\SCCpath = P_0 \xrightarrow{a_1} P_1 \xrightarrow{a_2}\cdots  P_k$ is the set
  \[\lang{\SCCpath} = \lang{P_0} a_1 \lang{P_1} a_2 \cdots \lang{P_k}.\]
\end{definition}

We say that $\SCCpath$ is \emph{accepting} if $P_0 = \portal{s_0}{x_0}{t_0}{y_0}$, $P_k = \portal{s_k}{x_k}{t_k}{y_k}$ with $x_0 = 0$, $s_0 = q_{0}$, $t_k \in F$ and $\lang{\SCCpath}$ is non-empty.
In that case, all the words of $\lang{\SCCpath}$ have a length equal to $y_k$ modulo $p$, and we call $y_k$ the \emph{residue} of $\SCCpath$.

\begin{lemma}\label{lemma:lang-path-union}
  We have \(\timedlang{\Aa} = \bigcup_{\SCCpath \text{ accepting}} \lang{\SCCpath}.\)
\end{lemma}
\begin{proof}
  We show that for any word $\mu$ in $\timedlang{\Aa}$, there is an accepting SCC-path $\SCCpath$ whose language contains $\mu$.
  Let $\mu = a_1 \cdots a_n$ be a word of length $n$ in $\timedlang{\Aa}$: there exists an accepting run $\rho = q_0 \xrightarrow{a_1} q_1 \cdots \xrightarrow{a_n} q_n$ in $\Aa$, with $q_n \in F$.

  We define the sequence of indices $i_0 < i_1 < \ldots < i_k <i_{k+1}$ as follows:
  \begin{itemize}
    \item $i_0 = 0, i_{k+1} = n+1$,
    \item for every $j=1,\ldots, k$, $i_j$ is the smallest index greater than $i_{j-1}$ such that $q_{i_j-1} <_\Aa q_{i_j}$, i.e.~$q_{i_j-1}$ and~$q_{i_j}$ belong to distinct SCCs.
  \end{itemize}
  In other words, those are the indices at which $\rho$ enters a new SCC.
  We then define the SCC-path $\SCCpath(\rho)$ as follows:
  \[\SCCpath(\rho) =
  \portal{q_0}{0}{q_{i_1-1}}{y_0} \xrightarrow{a_{i_1}} \portal{q_{i_1}}{x_1}{q_{i_2-1}}{y_1} \cdots  \xrightarrow{a_{i_k}} \portal{q_{i_k}}{x_k}{q_n}{y_k} \]
  where $x_j = i_j \pmod{p}$ for all $j = 0, \ldots, k+1$, and $y_j = x_{j+1}-1 \pmod{p}$ for all $j = 0, \ldots, k$.

  By construction, $\mu \in \lang{\SCCpath(\rho)}$ and $\SCCpath(\rho)$ is an accepting SCC-path.

  The converse inclusion follows by definition of (accepting) SCC-paths.
\end{proof}

As a consequence the distance between a word $\mu$ and the (positional) language of $\Aa$ is equal to the minimum of the distances between $\mu$ and the languages of the SCC-paths of $\Aa$.
\begin{corollary}\label{coro:dist-to-l-path}
  For any positional word $\mu$, we have \[d(\mu, \timedlang{\Aa}) = \min_{\SCCpath \text{ accepting}} d(\mu,\lang{\SCCpath}).\]
\end{corollary}

\subparagraph*{Residues.}
A word of length $n$ can only belong to the language of an accepting SCC-path of residue $n \bmod p$, and the SCC-paths whose language is finite only contain words of bounded length.
Let $n_0$ be larger than the length of all the words in the finite languages of accepting SCC-paths.
For $r \in \ZZ/p\ZZ$, we write $\Pi_r$ for the set of accepting SCC-paths of residue $r$ whose language is infinite.
Then, for every positional word $\mu$ of length $n \geq n_0$ with $n = r \pmod{p}$, we have $\mu \in \timedlang{\Aa}$ if and only if $\mu \in \lang{\SCCpath}$ for some $\SCCpath \in \Pi_r$, and $d(\mu, \timedlang{\Aa}) = \min_{\SCCpath \in \Pi_r} d(\mu,\lang{\SCCpath})$.
Moreover, $\Pi_r$ is non-empty if and only if $\lang{\Aa}$ contains words of infinitely many lengths equal to $r$ modulo $p$, that is, $r \in R(\lang{\Aa})$.
For this reason, we define blocking sequences for $\Aa$ relative to a residue $r$.

The decomposition of $\Aa$ into SCC-paths is an intermediate step in the definition of blocking sequences: we first generalize blocking factors for portals to blocking sequences for SCC-paths, and then define blocking sequences for automata.

\begin{definition}[(Strongly) Blocking Sequences for SCC-paths]
  We say that a sequence $\sigma = (\mu_1, \ldots, \mu_\ell)$ of positional words is \emph{blocking} for an SCC-path
  $\SCCpath = P_0 \xrightarrow{a_1} \cdots  P_k$ if there is a sequence of indices $i_0 \leq i_1 \leq \cdots \leq i_k$
  such that for every $j, \mu_{i_j}$ is blocking for~$P_j$.

  Furthermore, if there is a sequence of indices $i_0 < i_1 < \cdots < i_k$ with the same property, then~$\sigma$ is said to be \emph{strongly blocking} for~$\SCCpath$.
\end{definition}

Note that in the definition of blocking sequences, consecutive indices $i_j$ and $i_{j+1}$ can be equal, i.e. a single factor of the sequence may be blocking for multiple consecutive SCCs in the SCC-path. This choice is motivated by \cref{example:hard-to-easy}, where the language is easy because consecutive SCCs share blocking factors.

The following observation relates blocking factors and blocking sequences.
Just like a blocking factor is a witness that a word cannot be read inside an SCC, a blocking sequence can be seen as a witness that a word cannot be read along an SCC-path.

\begin{observation}\label{obs:sbs-from-disj-factors}
  Let $\SCCpath  = P_0 \xrightarrow{a_1} \cdots  P_k$ be an SCC-path.
  Let $\nu_0, \ldots, \nu_k$ be positional words such that $\nu_i$ is blocking for $P_i$.
  Then, $\sigma = (\nu_0, \ldots, \nu_k)$ is a strongly blocking sequence of $\SCCpath$.
\end{observation}

We say that two occurrences of blocking sequences in a word $\mu$ are \emph{disjoint} if the occurrences of their factors appear on disjoint sets of positions in $\mu$.

In the strongly connected case, we had the property that if $\mu$ contains an occurrence of a factor blocking for $\Aa$, then $\mu$ is not in the language of $\Aa$.
The following lemma gives an extension of this result to \emph{strongly} blocking sequences and the language of an SCC-path.
We say that a word $w$ \emph{contains} a sequence of words $(v_1, \dots, v_n)$ if $w = u_0 v_1 \dots u_{n-1} v_n u_n$ for some words $u_0, \dots, u_n$.
\begin{lemma}\label{lem:strongly-blocking-implies-not-in-path}
  Let $\SCCpath$ be an SCC-path.
  If $\mu$ contains a strongly blocking sequence for~$\SCCpath$,
  then $\mu \notin \lang{\SCCpath}$.
\end{lemma}
\begin{proof}
  We proceed by induction on the length $k$ of the SCC-path $\SCCpath =  P_0 \xrightarrow{a_1} P_1 \cdots \xrightarrow{a_k} P_k$.
  Let $\sigma$ be a strongly blocking sequence for $\SCCpath$ that occurs in $\mu$; by extracting a subsequence, we can assume that $\sigma = (\nu_0,\ldots, \nu_k)$ where $\nu_j$ is blocking for $P_j$ for all $j$.
  If $k = 0$, then~$\sigma$ consists only of a blocking factor for $P_0$, hence $\mu$ is not in $\lang{P_0}$, which is equal to $\lang{\SCCpath}$.

  For $k > 0$, assume for the sake of contradiction that $\mu\in\lang{\SCCpath}$. By definition of $\lang{\SCCpath}$,
  $\mu$ can then be written as $\mu_0a_1\mu'$, with $\mu_0\in \lang{P_0}$ and $\mu' \in\lang{\SCCpath'}$, where $\SCCpath' = P_1 \xrightarrow{a_2}\cdots \xrightarrow{a_k} P_k$.
  As $\nu_0$ is blocking for $P_0$, the prefix $\mu_0$ of $\mu$ must end before the occurrence of~$\nu_0$ in~$\mu$, and the sequence $\sigma' = (\nu_1,\ldots, \nu_k)$ occurs in $\mu'$.
  Furthermore, because $\sigma$ is \emph{strongly} blocking for~$\SCCpath$, we have that $\sigma'$ is strongly blocking for~$\SCCpath'$.
  Using the induction hypothesis on~$\mu'$ and the path~$\SCCpath'$ of length $k-1$, this implies that $\mu' \notin\lang{\SCCpath'}$, a contradiction.
\end{proof}

We can now define sequences that are blocking for an automaton: they are sequences that are blocking for \emph{every} accepting SCC-path of the automaton.
\begin{definition}[{Blocking sequence for $\Aa$}]
  Let $\sigma = (\mu_1, \ldots, \mu_\ell)$ be a sequence of positional words, and let $r \in \ZZ/p\ZZ$.
  We say that $\sigma$ is \emph{blocking} for $\Aa$ at residue $r$ if it is blocking for all the SCC-paths of~$\Pi_r$.
\end{definition}

\emph{From now on, and until the end of \cref{sec:trivial-easy}, we fix a residue $r \in R(\lang{\Aa})$.
We simply say that a sequence is blocking for $\Aa$ when it is blocking for $\Aa$ at residue $r$, and accepting SCC-paths are implicitly taken in $\Pi_r$.
Accordingly, $\MBS(\Aa)$ stands for the set $\MBS_r(\Aa)$ of \cref{thm:general}, and all the statements about words of length $n$ assume that $n \geq n_0$ and $n = r \pmod{p}$.}
For lengths $n$ that are not equal to an element of $R(\lang{\Aa})$ modulo $p$, the language contains no word of length $n \ge n_0$, and testing is trivial.
As an example, observe that the sequences $(\timedword{0}{ab}, \timedword{1}{ab})$ and $(\timedword{0}{aa}, \timedword{0}{b})$ are both blocking for the automaton displayed in \cref{fig:SCCpaths} (see \cref{ex:SCCpath}).

Let us come back to the examples of \cref{sec:blocking-seq}, ignoring positions for simplicity.
In the automaton $\Aa_1$ of \cref{fig:aut-ab-bc}, after removing the sink state $q_2$, the accepting SCC-paths go either through $\set{q_0}$ only, or through $\set{q_0}$ and then $\set{q_1}$.
The blocking factors of $\set{q_0}$ are the words containing a $b$, and those of $\set{q_1}$ are the words containing a $c$.
Hence a sequence is blocking for $\Aa_1$ if and only if it contains a term with a $b$, and a term with a $c$ that is the same term or a later one.
For instance, $(b,c)$ is blocking, and so is the sequence $(ba^kc)$ for every $k$, since its single term contains both a $b$ and a $c$ and thus blocks both SCCs.
In the automaton $\Aa_2$ of \cref{fig:series}, a blocking sequence must contain a term with an $a$, to block the second SCC, and the sequence $(a)$ is blocking, since the factor $a$ blocks both SCCs.
We will see that this implies that $(a)$ is the only minimal blocking sequence of $\Aa_2$, although the first SCC has infinitely many minimal blocking factors.

\begin{example}
  \label{ex:SCCpath}
  Consider the automaton displayed in \cref{fig:SCCpaths}.
  The $\lcm$ of the lengths of its simple cycles is $p=2$.
  This automaton has six accepting SCC-paths, including
  \begin{align*}
    \pi_1 &= \portal{q_0}{0}{q_0}{0} \xrightarrow{a} \portal{q_1}{1}{q_1}{1} \xrightarrow{a} \portal{q_3}{0}{q_3}{0} \xrightarrow{b} \portal{q_4}{1}{q_4}{1}\\
    \pi_2 &= \portal{q_0}{0}{q_0}{0} \xrightarrow{a} \portal{q_2}{1}{q_1}{0} \xrightarrow{a} \portal{q_3}{1}{q_3}{0} \xrightarrow{b} \portal{q_4}{1}{q_4}{1}
  \end{align*}

  The language of the SCC-path $\pi_1$ is $a(ba)^*a(a^2)^*b$.
  A blocking sequence for this SCC-path is $(\timedword{0}{aa}, \timedword{0}{b})$, which is in fact blocking for all of the SCC-paths.
  It is, however, not strongly blocking for $\pi_1$ for instance, since $\pi_1$ goes through 4 SCCs and this sequence only has 2 elements.

  \begin{figure}[htbp]
    \begin{center}
      \begin{tikzpicture}[shorten >=1pt, node distance=3cm, on grid, auto,
        every initial by arrow/.style={thick}, every accepting by arrow/.style={thick}, >={stealth'}]
        \tikzset{
          state/.style={
            circle,
            thick,
            draw,
            minimum size=8mm,
            inner sep=2pt
          },
          accepting/.style={
            state,
            accepting by arrow
          }
        }

        \node[state , initial] (q0) {$q_0$};
        \node[state ,right=of q0] (q1) {$q_1$};
        \node[state] (q2) [above=of q1] {$q_2$};
        \node[state] (q3) [right=of q1] {$q_3$};
        \node[state, accepting, accepting right] (q4) [above=of q3] {$q_4$};

        \path[->, thick, >=stealth']
        (q0) edge[above] node {\textcolor{red}{$a$}} (q1)
        (q0) edge[above left] node {\textcolor{red}{$a$}} (q2)
        (q1) edge[above] node {\textcolor{red}{$a$}} (q3)
        (q2) edge[above] node {\textcolor{red}{$b$}} (q4)
        (q3) edge[right] node {\textcolor{red}{$b$}} (q4)
        (q3) edge[loop right] node {$a$} ();

        \path[->, thick, >=stealth',bend right]
        (q1) edge[right] node {{$b$}} (q2)
        (q2) edge[left] node {{$a$}} (q1);
      \end{tikzpicture}
      \caption{Automaton used for \cref{ex:SCCpath}.}
      \label{fig:SCCpaths}
    \end{center}
  \end{figure}

  On the other hand, $(\timedword{0}{ab})$ is not blocking for $\pi_1$, as $\timedword{0}{ab}$ is not a blocking factor for the portal $\portal{q_1}{1}{q_1}{1}$. It is, however, a blocking sequence for $\pi_2$.
  This is because if we enter the SCC $\set{q_1,q_2}$ through $q_1$, a factor $ab$ can only appear after an even number of steps, while if we enter through $q_2$, it can only appear after an odd number of steps.
\end{example}

\subsection{An efficient property tester}\label{sec:generic-ub}

In this section, we show that for any regular language $L$ and any small enough $\eps > 0$, there is an $\eps$-property tester for $L$ that uses $\cO(\epslogeps)$ queries.

\begin{theorem}\label{thm:gen-ub}
  For any NFA $\Aa$ and any small enough $\eps  >0$, there exists an $\eps$-property tester for $\lang{\Aa}$ that uses $\cO(\epslogeps)$ queries.
\end{theorem}

As mentioned in the overview, this result supersedes the one  given by Bathie and Starikovskaya~\cite{bathie2021property}:
while both testers use the same number of queries, the tester in~\cite{bathie2021property} works under the edit distance, while that of \cref{thm:gen-ub} is designed for the Hamming distance. As the edit distance never exceeds the Hamming distance, the set of words that are $\eps$-far with respect to the former is contained in the set of words $\eps$-far for the latter. Therefore, an $\eps$-tester for the Hamming distance is also an $\eps$-tester for the edit distance, and this result is stronger.

The property tester behind \cref{thm:gen-ub} uses the property tester for strongly connected NFAs as a subroutine, and its correctness is based on an extension of \cref{lemma:many-short-blocking} to blocking sequences.
We show that we can reduce property testing of $\lang{\Aa}$ to a search for blocking sequences in the word, in the following sense:
\begin{itemize}
  \item If $\mu$ contains a strongly blocking sequence for each of the SCC-paths of $\Aa$, then it is not in the language and we can answer no (\cref{cor:many-strongly-bs-implies-not-A}).
  \item If $\mu$ is $\eps$-far from the language, then for each accepting SCC-path $\SCCpath$ of $\Aa$, $\mu$ is far from the language of $\SCCpath$ and contains many disjoint strongly blocking sequences for $\SCCpath$ (\cref{lem:far-from-path-implies-many-blocking}), hence random sampling is likely to find at least one of them, and we reject $\mu$ with constant probability.
\end{itemize}

\begin{corollary}
  \label{cor:many-strongly-bs-implies-not-A}
  If $\mu$ contains a strongly blocking sequence for each SCC-path of~$\Aa$, then $\mu \notin \timedlang{\Aa}$.
\end{corollary}
\begin{proof}
  This follows from \cref{lemma:lang-path-union} and \cref{lem:strongly-blocking-implies-not-in-path}.
\end{proof}

The next lemma expresses a partial converse to \cref{cor:many-strongly-bs-implies-not-A} and generalizes \cref{lemma:many-blocking} from the strongly connected case: if a word is far from the language, then for every accepting SCC-path, it contains many blocking factors for each of the SCCs along that SCC-path, in the same order.
We prove it for all the SCC-paths of $\Pi_r$ simultaneously, as this will also give us disjoint blocking sequences for $\Aa$.
Recall that an SCC-path in $\Pi_r$ has an infinite language, hence it contains at least one recurrent portal, and that the language of a transient portal of an accepting SCC-path only contains the empty word.

We say that two factor occurrences $(v,[i,j])$ and $(v',[i',j'])$ are \emph{separated by $R$ positions} if $j + R < i'$ or $j' + R < i$.
An \emph{occurrence} of a sequence $\sigma = (\nu_1, \ldots, \nu_k)$ in a positional word $\mu$ is a sequence of factor occurrences $(\nu_1, I_1), \ldots, (\nu_k, I_k)$ of $\mu$ such that the intervals $I_1, \ldots, I_k$ are pairwise disjoint and appear in this order; then $\mu$ contains $\sigma$.
Two occurrences of sequences are \emph{disjoint} if all their intervals are pairwise disjoint.

\begin{lemma}
  \label{lem:far-from-path-implies-many-blocking}
  There are constants $C$ and $R \geq K$, where $K$ is the maximal number of portals of an SCC-path, such that the following holds for every $0 < \eps \le 1$ and every positional word $\mu$ of length $n \geq C/\eps$ (with $n \geq n_0$ and $n = r \pmod{p}$) that is $\eps$-far from $\timedlang{\Aa}$.
  Let $N = \lceil \eps n / C \rceil$.
  There is a set $\Tt$ of factor occurrences of $\mu$, pairwise separated by $R$ positions and contained in $[R, n-R-1]$, and a map $\theta$ that associates with each element of $\Tt$ a pair $(\SCCpath, h)$, where $\SCCpath \in \Pi_r$ and $P_h$ is a recurrent portal of $\SCCpath$, such that:
  \begin{enumerate}
    \item every $T \in \Tt$ with $\theta(T) = (\SCCpath, h)$ is blocking for the portal $P_h$ of $\SCCpath$;
    \item for every $\SCCpath \in \Pi_r$ and every recurrent portal $P_h$ of $\SCCpath$, exactly $N$ elements of $\Tt$ are mapped to $(\SCCpath,h)$;
    \item for every $\SCCpath = P_0 \xrightarrow{a_1} \cdots P_k \in \Pi_r$ and all recurrent portals $P_h, P_{h'}$ of $\SCCpath$ with $h < h'$, all the elements of $\Tt$ mapped to $(\SCCpath, h)$ appear before all the elements mapped to $(\SCCpath, h')$.
  \end{enumerate}
\end{lemma}
\begin{proof}
  We first note two facts on walks in $\Aa$.
  Every recurrent state $q$ has a closed walk of length $p$, since a closed walk on $q$ of minimal positive length is a simple cycle, whose length divides $p$ by definition of $p$.
  Consequently, there is a constant $\rho$ such that, for all states $q, q'$ and every walk from $q$ to $q'$ of length $\ell$ that visits a recurrent state, there is a walk from $q$ to $q'$ of length $L$ for every $L \geq \rho$ with $L = \ell \pmod p$.
  Indeed, for all states $q, q'$ and every residue modulo $p$, fix a shortest walk from $q$ to $q'$ that visits a recurrent state and has a length with this residue (if there is one), and let $\rho$ be the maximal length of these walks; a walk of length $L$ is obtained by inserting closed walks of length $p$ on the recurrent state.
  If $\Pi_r$ is empty, the lemma holds with $\Tt = \emptyset$, so we assume that $\Pi_r$ is non-empty, and we set $R = \max(\rho, K)$, $M = |\Pi_r| K \geq 1$ and $C = 8M(R+1)$.

  \emph{The construction.}
  We maintain, for every $\SCCpath \in \Pi_r$, a current recurrent portal of $\SCCpath$ and a counter; $\SCCpath$ is \emph{completed} when all its recurrent portals have been treated.
  Initially, the current portal of each $\SCCpath$ is its first recurrent portal, all counters are $0$, and the current position is $c = R$.
  We repeat the following step until all SCC-paths are completed.
  Let $j$ be the smallest integer such that $c \leq j \leq n-R-1$ and $\mu[c\dd j]$ is blocking for the current portal of some SCC-path $\SCCpath$ that is not completed; if there is no such $j$, we stop.
  Otherwise, we add $T = \mu[c\dd j]$ to $\Tt$ and set $\theta(T) = (\SCCpath, h)$, where $\SCCpath$ is the first such SCC-path in some fixed order and $P_h$ is its current portal.
  We then increment the counter of $\SCCpath$; if it reaches $N$, we reset it to $0$ and the current portal of $\SCCpath$ becomes its next recurrent portal (if any; otherwise $\SCCpath$ is completed), and we set $c = j+R+1$.

  By construction, the elements of $\Tt$ are pairwise separated by $R$ positions, are contained in $[R, n-R-1]$ and satisfy the first and the third item, and at most $N$ of them are mapped to each pair $(\SCCpath, h)$.
  As each element of $\Tt$ is mapped to one pair, we have $|\Tt| \le MN$.
  It remains to show that the construction stops because all SCC-paths are completed, which gives the second item.

  \emph{Repairing the word.}
  Assume for contradiction that the construction stops while some $\SCCpath = P_0 \xrightarrow{a_1} \cdots P_k \in \Pi_r$ is not completed.
  We build a positional word $\mu'$ of length $n$ accepted by $\Aa$ that differs from $\mu$ in fewer than $\eps n$ positions, which contradicts the fact that $\mu$ is $\eps$-far from $\timedlang{\Aa}$.
  We write $P_h = \portal{s_h}{x_h}{t_h}{y_h}$ for all $h$.
  Consider the factors of $\mu$ of the form $\mu[c\dd j-1]$, where $\mu[c\dd j]$ is an element of $\Tt$, together with the factor $\mu[c\dd n-R-1]$, where $c$ is the current position when the construction stops (if $c \le n-R-1$).
  At the time it was considered, each of these factors was not blocking for the current portal $P_h$ of $\SCCpath$, by minimality of $j$ (for the last one, because the construction stopped).
  Hence, if such a factor $\mu[c\dd j-1]$ is non-empty, it is read by a run of $\Aa$ from a state $q$ to a state $q'$ of the SCC of $P_h$, such that there are paths from $s_h$ to $q$ of length $c - x_h$ modulo $p$ and from $q'$ to $t_h$ of length $y_h - j$ modulo $p$.
  We keep the non-empty factors unchanged, and we substitute all the other positions of $\mu$, namely the first $R$ positions, the last position of each element of $\Tt$ together with the $R$ positions that follow it, and the positions after $n-R-1$ (or after the last element of $\Tt$ and its $R$ following positions, if they go beyond).
  The substituted positions form blocks before the first kept factor, between two consecutive kept factors and after the last one, each of length at least $R \geq \rho$, and we fill each block with a walk of $\Aa$ of exactly its length.
  \begin{itemize}
    \item Consider a block of length $G$ between a kept factor read up to a state $q'$ of the SCC of a portal $P_h$ and the next kept factor, read from a state $q$ of the SCC of a portal $P_{h'}$, with $h \le h'$.
    Following $\SCCpath$ from $q'$ to $t_h$, then through the portals $P_{h+1}, \ldots, P_{h'-1}$ and the letters between them to $s_{h'}$, and then to $q$, gives a walk from $q'$ to $q$ whose length is equal to $G$ modulo $p$, as the positions of the kept factors are part of the positional word $\mu$.
    This walk visits the recurrent state $q'$, so there is such a walk of length exactly $G$.
    \item The first block is filled similarly with a walk from $q_0$, following $\SCCpath$ up to the first kept factor, and the last block with a walk following $\SCCpath$ up to a final state; its length has the right residue since $n = r \pmod p$ is the residue of $\SCCpath$. If there is no kept factor, we use a single walk of length $n$ from $q_0$ to a final state that follows $\SCCpath$ and visits a state of one of its recurrent portals.
  \end{itemize}
  In the second item, the walks visit a recurrent state as well, since the kept factors are read in the SCCs of recurrent portals and $\SCCpath$ has a recurrent portal.
  The resulting word $\mu'$ has length $n$ and is read by an accepting run of $\Aa$, so $\mu' \in \timedlang{\Aa}$.
  The number of substituted positions is at most
  \[ R + |\Tt|(R+1) + 2R + 1 \le MN(R+1) + 3R + 1 \le \frac{\eps n}{8} + M(R+1) + 3R + 1 < \eps n,\]
  using $N \le \eps n/C + 1$ and $C = 8M(R+1)$, as soon as $\eps n \geq 2(M(R+1)+3R+1)$, which holds for $n \ge C/\eps$ since $M \geq 1$ (the filling of the blocks also requires $n \geq C$, which follows from $n \geq C/\eps$ and $\eps \le 1$).
  This is the desired contradiction.
\end{proof}

We obtain two consequences of \cref{lem:far-from-path-implies-many-blocking}, which give many short blocking factors for each portal of each SCC-path and many disjoint blocking sequences for $\Aa$.

\begin{corollary}\label{cor:many-short-blocking-path}
  Let $C, R, N, \Tt$ and $\theta$ be as in \cref{lem:far-from-path-implies-many-blocking}, and let $B = \lceil 2n/N \rceil \le 2C/\eps + 1$.
  For every $\SCCpath \in \Pi_r$ and every recurrent portal $P_h$ of $\SCCpath$, at least $N/2$ elements of $\Tt$ mapped to $(\SCCpath, h)$ have length at most $B$.
\end{corollary}
\begin{proof}
  The elements of $\Tt$ are disjoint, so the sum of their lengths is at most $n$, and fewer than $n/B \le N/2$ of them have length more than $B$.
\end{proof}

\begin{corollary}\label{cor:far-L-many-bs}
  Let $C, N, \Tt$ and $\theta$ be as in \cref{lem:far-from-path-implies-many-blocking}.
  For $\ell = 1, \ldots, N$, let $\sigma_\ell$ be the sequence of the elements $T$ of $\Tt$, in the order in which they appear in $\mu$, such that $T$ is the $\ell$-th element of $\Tt$ mapped to $\theta(T)$.
  Then $\sigma_1, \ldots, \sigma_N$ are $N \geq \eps n/C$ pairwise disjoint occurrences of blocking sequences for $\Aa$, and all their factor occurrences are pairwise disjoint.
\end{corollary}
\begin{proof}
  Every element of $\Tt$ belongs to exactly one $\sigma_\ell$.
  Let $\SCCpath = P_0 \xrightarrow{a_1} \cdots P_k \in \Pi_r$.
  For every recurrent portal $P_h$ of $\SCCpath$, $\sigma_\ell$ contains the $\ell$-th element of $\Tt$ mapped to $(\SCCpath, h)$, which is blocking for $P_h$, and these elements appear in the order of the portals.
  The transient portals of $\SCCpath$ are blocked by any non-empty positional word, so we can assign each of them to the term of $\sigma_\ell$ used for the previous recurrent portal (or the next one, for the transient portals before the first recurrent portal).
  This gives non-decreasing indices, so $\sigma_\ell$ is blocking for $\SCCpath$.
\end{proof}

We are now ready to prove \cref{thm:gen-ub}.
\begin{proof}[Proof of \cref{thm:gen-ub}]
  Let $C$ and $R$ be the constants of \cref{lem:far-from-path-implies-many-blocking}, and let $u$ be the input word, of length $n$, and $\mu = \twu$.
  We assume $\eps \le 1$. If $n < \max(n_0, 2C/\eps)$, the tester queries all of $u$ and accepts if and only if $u \in \lang{\Aa}$.
  Otherwise, let $r = n \bmod p$; if $r \notin R(\lang{\Aa})$, then $\lang{\Aa}$ contains no word of length $n$, and the tester rejects.
  Otherwise, let $N = \lceil \eps n/C \rceil \geq 2$ and $B = \lceil 2n/N\rceil \le n$, so that $\lfloor n/B \rfloor \geq 1$, and let $t = \lceil \log_3(3M) \rceil$, where $M$ is the number of pairs $(\SCCpath, h)$ with $\SCCpath \in \Pi_r$ and $P_h$ a recurrent portal of $\SCCpath$.
  The tester calls $t$ times $\Call{Sampler}{\mu, \lfloor n/B \rfloor, B}$ (\cref{alg:generic-sampling}), and for each factor occurrence $[i,j]$ returned by these calls, it queries all the positions in $[i-K, j+K]$, where $K$ is the maximal number of portals of an SCC-path.
  It rejects if and only if, for every $\SCCpath \in \Pi_r$, the queried positions contain an occurrence of a strongly blocking sequence for $\SCCpath$.
  This tester is non-adaptive, as the positions queried by \textsc{Sampler} do not depend on the answers to previous queries.

  If $\mu \in \timedlang{\Aa}$, then $\mu \in \lang{\SCCpath}$ for some $\SCCpath \in \Pi_r$, hence by \cref{lem:strongly-blocking-implies-not-in-path}, $\mu$ does not contain a strongly blocking sequence for $\SCCpath$, and the tester accepts, so it has perfect completeness.

  Now assume that $\mu$ is $\eps$-far from $\timedlang{\Aa}$, and let $\Tt$ and $\theta$ be given by \cref{lem:far-from-path-implies-many-blocking}.
  For every pair $(\SCCpath, h)$ as above, let $\Ss_{\SCCpath, h}$ be the set of elements of $\Tt$ mapped to $(\SCCpath, h)$ of length at most $B$: it contains at least $N/2 \geq n/B \ge \lfloor n/B \rfloor$ disjoint factor occurrences by \cref{cor:many-short-blocking-path}.
  By \cref{lemma:generic-sampling-alg}, each call to \textsc{Sampler} returns a factor occurrence covering an element of $\Ss_{\SCCpath, h}$ with probability at least $2/3$; note that the algorithm does not depend on $\Ss_{\SCCpath, h}$, so the same calls can be used for all pairs.
  Hence the probability that no call covers an element of $\Ss_{\SCCpath, h}$ is at most $3^{-t} \le 1/(3M)$, and by the union bound, with probability at least $2/3$, for every pair $(\SCCpath, h)$, some element $T_{\SCCpath,h}$ of $\Ss_{\SCCpath,h}$ is covered by a queried factor occurrence.

  In that case, let $\SCCpath = P_0 \xrightarrow{a_1} \cdots P_k \in \Pi_r$.
  By the third item of \cref{lem:far-from-path-implies-many-blocking}, the occurrences $T_{\SCCpath, h}$ for the recurrent portals $P_h$ of $\SCCpath$ are disjoint and appear in the order of the portals.
  For the transient portals, we use single letters: the $K$ positions that follow (resp.\ precede) an element of $\Tt$ exist, as elements of $\Tt$ are contained in $[R, n-R-1]$, they do not belong to any element of $\Tt$, as elements of $\Tt$ are separated by $R \ge K$ positions, and they have been queried.
  We assign to the transient portals that follow a recurrent portal $P_h$ the positions that follow $T_{\SCCpath,h}$, and to the transient portals before the first recurrent portal the positions that precede the occurrence used for that portal.
  As every non-empty positional word is blocking for a transient portal, we obtain an occurrence of a strongly blocking sequence for $\SCCpath$ in the queried positions.
  Hence the tester rejects with probability at least $2/3$.

  Finally, each call to \textsc{Sampler} makes $\cO(n(1+\log B)/\lfloor n/B\rfloor) = \cO(B(1+\log B)) = \cO(\epslogeps)$ queries by \cref{claim:generic-sampling-alg-complexity}, and returns $\cO(B) = \cO(1/\eps)$ factor occurrences, so extending them costs $\cO(K/\eps)$ additional queries.
  As $t$, $C$, $K$ and $M$ only depend on $\Aa$, the tester makes $\cO(\epslogeps)$ queries in total.
\end{proof}

\subsection{Lower bound}\label{sec:general-lb}

In the previous section we used strongly blocking sequences, since a word containing one of them is directly certified not to be in the language.
For the lower bound, we use blocking sequences instead, in which a single factor may block several consecutive portals, and which are closely related to strongly blocking sequences.
For a sequence $\sigma$ and an integer $r \geq 1$, let $\sigma^{(r)}$ denote the sequence obtained from $\sigma$ by repeating each term $r$ times.
Then $\sigma$ is blocking for an SCC-path with $k+1$ portals if and only if $\sigma^{(k+1)}$ is strongly blocking for it: if $\sigma$ is blocking with indices $i_0 \le \cdots \le i_k$, we can use a different copy of $\mu_{i_j}$ for each portal blocked by $\mu_{i_j}$.
In particular, if $\sigma$ is blocking for $\Aa$ and $K$ is the maximal number of portals of an accepting SCC-path, then by \cref{cor:many-strongly-bs-implies-not-A}, a positional word that contains $\sigma^{(K)}$ is not in $\timedlang{\Aa}$.

We do not use strongly blocking sequences in this section because they behave badly with respect to minimality.
With the order $\pobs$ defined below, the set of strongly blocking sequences is not upward closed: in \cref{example:hard-to-easy}, $(a,a)$ is strongly blocking for $\Aa_2$ and $(a,a) \pobs (a)$, but $(a)$ is not strongly blocking.
If we instead compared sequences by embedding each term into a different term, strongly blocking sequences would be upward closed, but $\Aa_2$ would have infinitely many minimal strongly blocking sequences, namely $(be^kc, a)$ for every $k \in \NN$, although $L_2$ is easy.

In order to characterize hard languages for all automata, we define a preorder $\pobs$ on sequences of positional words, which extends the factor partial order on blocking factors and lets us define \emph{minimal blocking sequences}; we use these to characterize the complexity of testing a language.

\begin{definition}[Minimal blocking sequence]\label{def:MBS}
  Let $\sigma = (\mu_1, \mu_2, \ldots, \mu_k)$ and $\sigma' =  (\mu_1',\ldots, \mu_{t}')$ be sequences of positional words.
  We have $\sigma \pobs \sigma'$ if there exists a sequence of indices $i_1 \leq i_2 \leq \ldots \leq i_k$ such that $\mu_{j}$ is a factor of $\mu_{i_j}'$ for all $j = 1,\ldots, k$.
  We write $\sigma \equiv \sigma'$ if $\sigma \pobs \sigma'$ and $\sigma' \pobs \sigma$, and $\sigma \triangleleft \sigma'$ if $\sigma \pobs \sigma'$ and $\sigma' \not\pobs \sigma$.

  A blocking sequence $\sigma$ of $\Aa$ (resp.\ $\SCCpath$) is \emph{minimal} if there is no blocking sequence $\sigma'$ of $\Aa$ (resp.\ $\SCCpath$) such that $\sigma' \triangleleft \sigma$.
  We write $\MBS(\Aa)$ (resp.\ $\MBS(\SCCpath)$) for the set of $\equiv$-classes of minimal blocking sequences of $\Aa$ (resp.\ $\SCCpath$).
\end{definition}

The relation $\pobs$ is a preorder but not an order (for instance $(a) \pobs (a,a) \pobs (a)$), which is why we consider minimal blocking sequences up to $\equiv$.
Note that a sequence equivalent to a minimal blocking sequence is also a minimal blocking sequence.
We say that a sequence is \emph{reduced} if none of its terms is a factor of an adjacent term.
If a term of $\sigma$ is a factor of an adjacent term, then deleting it gives a sequence equivalent to $\sigma$; hence every sequence is equivalent to a reduced sequence, obtained by repeating this operation.

\begin{remark}
  If $\sigma \pobs \sigma'$ and $\sigma$ is a blocking sequence for an SCC-path $\pi$ then $\sigma'$ is also a blocking sequence for $\pi$.
\end{remark}

For instance, in \cref{example:hard-to-easy}, every blocking sequence $\sigma$ of $\Aa_2$ contains a term with an $a$, hence $(a) \pobs \sigma$; as $(a)$ is blocking, every minimal blocking sequence is equivalent to $(a)$, and $\MBS(\Aa_2)$ is a singleton.
Similarly, every blocking sequence $\sigma$ of $\Aa_1$ satisfies $(b,c) \pobs \sigma$, and $\MBS(\Aa_1)$ is the class of $(b,c)$.

We now show that minimal blocking sequences have a bounded number of terms, up to equivalence.
The idea is that if we build a minimal blocking sequence from left to right by adding terms one by one, then each term should block some portal that was not blocked before.
To formalize this, we use the \emph{left effect} of a sequence $\sigma$ on an SCC-path $\SCCpath = P_0 \xrightarrow{a_1} \cdots P_k$, which is the largest index $i$ such that the sequence is blocking for $P_0 \xrightarrow{a_1} \cdots P_i$ (or $-1$ if there is no such~$i$). We denote it by $\lefteffect{\sigma}{\SCCpath}$.
As the indices in the definition of blocking sequences are non-decreasing, the left effect can be computed greedily:
for every sequence $\sigma$ and positional word $\mu$, if $e = \lefteffect{\sigma}{\SCCpath}$, then $\lefteffect{\sigma (\mu)}{\SCCpath}$ is the largest $j \geq e$ such that $\mu$ is blocking for $P_{e+1}, \ldots, P_j$.
In particular, the left effect of $\sigma(\mu)\sigma'$ only depends on the left effect of $\sigma(\mu)$ and on $\sigma'$.

\begin{lemma}
  \label{lem:bound-length-min-blocking}
  Let $M$ be the sum, over all accepting SCC-paths $\SCCpath$ of $\Aa$, of the number of portals of $\SCCpath$.
  \begin{enumerate}
    \item For every blocking sequence $\sigma$ of $\Aa$, there is a minimal blocking sequence $\sigma^*$ of $\Aa$ such that $\sigma^* \pobs \sigma$.
    \item Every minimal blocking sequence of $\Aa$ is equivalent to a reduced minimal blocking sequence with at most $M$ terms.
  \end{enumerate}
\end{lemma}
\begin{proof}
  Let $\sigma = (\mu_1, \ldots, \mu_\ell)$ be a blocking sequence and, for $i \le \ell$, let $\sigma_i = (\mu_1, \ldots, \mu_i)$.
  Say that the term $\mu_{i+1}$ is \emph{useless} if $\lefteffect{\sigma_{i+1}}{\SCCpath} = \lefteffect{\sigma_{i}}{\SCCpath}$ for every accepting SCC-path $\SCCpath$.
  By the greedy computation of left effects, deleting all useless terms from $\sigma$ does not change the left effects of the remaining prefixes, hence it gives a blocking sequence $\sigma'$ with $\sigma' \pobs \sigma$.
  In $\sigma'$, each term increases the left effect on at least one accepting SCC-path.
  As the left effect on an SCC-path with $k+1$ portals ranges between $-1$ and $k$, the sequence $\sigma'$ has at most $M$ terms.

  For the first item, let $R(\sigma)$ be the set of blocking sequences $\tau$ with at most $M$ terms such that $\tau \pobs \sigma$.
  All the terms of such a sequence are factors of terms of $\sigma$, hence $R(\sigma)$ is finite. It is non-empty, as it contains the sequence obtained from $\sigma$ by deleting its useless terms.
  Hence it has an element $\sigma^*$ such that no element $\tau$ of $R(\sigma)$ satisfies $\tau \triangleleft \sigma^*$.
  We show that $\sigma^*$ is minimal.
  Let $\tau$ be a blocking sequence such that $\tau \triangleleft \sigma^*$.
  Deleting the useless terms of $\tau$ gives a blocking sequence $\tau' \pobs \tau$ with at most $M$ terms, and $\tau' \pobs \tau \pobs \sigma^* \pobs \sigma$, hence $\tau' \in R(\sigma)$.
  Moreover, $\sigma^* \not\pobs \tau'$, as otherwise $\sigma^* \pobs \tau' \pobs \tau$, contradicting $\tau \triangleleft \sigma^*$.
  Hence $\tau' \triangleleft \sigma^*$, a contradiction.

  For the second item, let $\sigma$ be a minimal blocking sequence. Deleting its useless terms gives a blocking sequence $\sigma' \pobs \sigma$, so $\sigma' \equiv \sigma$ by minimality, and $\sigma'$ has at most $M$ terms.
  We then reduce $\sigma'$, which does not increase its number of terms and gives an equivalent sequence, which is minimal as it is equivalent to $\sigma$.
\end{proof}

\subsubsection{Reducing to the strongly connected case}

To prove a lower bound on the number of queries necessary to test a language when $\MBS(\Aa)$ is infinite, we present a reduction to the strongly connected case.
Under the assumption that $\Aa$ has infinitely many minimal blocking sequences, we exhibit a portal $P$ of $\Aa$ with infinitely many minimal blocking factors and ``isolate it'' by constructing two sequences of positional words $\sigma_l$ and $\sigma_r$ such that for all $\mu$, the sequence $\sigma_l, (\mu), \sigma_r$ is blocking for $\Aa$ if and only if $\mu$ is a blocking factor of $P$.
Then we reduce the problem of testing the language of $P$ to the problem of testing the language of $\Aa$.

To define ``isolating $P$'' formally, we use the left effect of a sequence on an SCC-path, defined above, and its symmetric notion.
Informally, the left effect of a sequence $\sigma$ on an SCC-path $\SCCpath$ is the index of the last portal of $\SCCpath$ that is blocked by $\sigma$ together with all the portals before it, so that a run that has read $\sigma$ can only be in the following portals.
The right effect represents the same in reverse, starting from the end of the run.

Formally, the \emph{right effect} of a sequence on $\pi$ is the smallest index~$i$ such that the sequence is blocking for $P_i \xrightarrow{a_{i+1}} \cdots P_k$  (or $k+1$ if there is no such~$i$); we denote it by $\righteffect{\sigma}{\pi}$.

\begin{remark}\label{rmk:effects}
  A sequence $\sigma$ is blocking for an SCC-path $\SCCpath = P_0 \xrightarrow{a_1} \cdots  P_k$ if and only if $\lefteffect{\sigma}{\SCCpath} = k$, if and only if $\righteffect{\sigma}{\SCCpath} = 0$.

  Also, given two sequences $\sigma_l, \sigma_r$, the sequence $\sigma_l \sigma_r$ is blocking for $\SCCpath$ if and only if $\lefteffect{\sigma_l}{\SCCpath} \ge \righteffect{\sigma_r}{\SCCpath} -1$.
\end{remark}

For the next lemma we define a preorder on portals\footnote{This preorder concerns the set of blocking factors for portals and is entirely unrelated to the partial order given by the reachability relation between SCCs.}:
$P \leqportals P'$ if the set of blocking factors of $P$ is included in the one of $P'$.
We write $\geqportals$ for the reverse relation, $\equivportals$ for the equivalence relation $\leqportals \cap \geqportals$ and $\nequivportals$ for the complement relation of $\equivportals$.

Given an SCC-path~$\SCCpath = P_0\xrightarrow{a_1} \ldots P_k$ and two sequences of positional words $\sigma_l, \sigma_r$, we say that the portal $P_i$ \emph{survives $(\sigma_l, \sigma_r)$ in $\SCCpath$} if $\lefteffect{\sigma_l}{\SCCpath} < i < \righteffect{\sigma_r}{\SCCpath}$.

\begin{definition}
  Let $P_*$ be a portal and  $\sigma_l$ and $\sigma_r$ sequences of positional words.

  We define three properties that those objects may have:
  \begin{description}%
    \item[(P1) (Non-blocking context)] $\sigma_l \sigma_r$ is not blocking for $\Aa$
    \item[(P2) (Hard portal)] $P_*$ has infinitely many minimal blocking factors
    \item[(P3) (Context isolates the hard portal)] for every accepting SCC-path $\SCCpath$ of $\Aa$,
      every portal of $\SCCpath$ that survives $(\sigma_l, \sigma_r)$ in $\SCCpath$ is $\geqportals P_*$; moreover, there is an accepting SCC-path $\SCCpath'$ of $\Aa$ and a portal of $\SCCpath'$ that survives $(\sigma_l, \sigma_r)$ in $\SCCpath'$ and is $\equivportals P_*$.
  \end{description}
\end{definition}

Property P2 says that $P_*$ is a hard portal, and properties P1 and P3 imply that for all $\nu'$, $\sigma_l \nu' \sigma_r$ is blocking for $\Aa$ if and only if $\nu'$ is blocking for $P_*$, which lets us reduce testing the language of $P_*$ to testing the one of $\Aa$.

\subparagraph*{Running example: isolating the parity component.}
We extend the repeated-parity example of \cref{sec:scc} by considering the language
\[
  L'=e^*\,\#\,L_0\,\#\,f^*,
\]
where $e$, $f$ and $\#$ are new letters.
An automaton for $L'$ has three strongly connected components: one reading $e^*$, one for $L_0$, and one reading $f^*$.
Ignoring positions modulo the period for simplicity, we write $P_L$, $P_*$ and $P_R$ for the corresponding portals, and consider the SCC-path
\[
  \pi=P_L\xrightarrow{\#}P_*\xrightarrow{\#}P_R.
\]
The factor $a$ is blocking for $P_L$ and $P_R$ but not for $P_*$, and we take $\sigma_l = (a)$ and $\sigma_r = (a)$.
Indexing the portals of $\pi$ by $0$, $1$ and $2$, we have $(\sigma_l\gg\pi)=0$ and $(\pi\ll\sigma_r)=2$, so $P_*$ is the only portal of $\pi$ that survives $(\sigma_l, \sigma_r)$.

Property~P1 holds since $\sigma_l$ and $\sigma_r$ do not block $P_*$, P2 holds since $P_*$ has infinitely many minimal blocking factors, namely the words $ca^rbd$, and P3 holds since $P_*$ is the only surviving portal.
As a consequence, for every positional word $\nu$, the sequence $(a,\nu,a)$ is blocking for $L'$ if and only if $\nu$ is blocking for $P_*$.

In general, several portals may survive $(\sigma_l, \sigma_r)$, and property~P3 then says that every blocking factor of $P_*$ is blocking for all surviving portals, which gives one direction of the above equivalence, and that one of the surviving portals has the same blocking factors as $P_*$, which gives the other direction.

\begin{restatable}{lemma}{HardAutHardPortal}
  \label{lem:hard-aut-to-hard-portal}
  If $\Aa$ has infinitely many minimal blocking sequences, then there exist a portal $P_*$ and sequences $\sigma_l$ and $\sigma_r$ satisfying properties P1, P2 and P3.
\end{restatable}

\begin{proof}
  We start by building $\sigma_l$ and $\sigma_r$.
  By \cref{lem:bound-length-min-blocking}, every class of $\MBS(\Aa)$ contains a reduced minimal blocking sequence with at most $M$ terms.
  Since $\MBS(\Aa)$ is infinite and there are finitely many sequences with at most $M$ terms of bounded length, the lengths of the terms of these representatives are unbounded.
  Hence there exist an index $i$ and an infinite family of reduced minimal blocking sequences $\sigma_j = \alpha_j\,(\nu_j)\,\beta_j$, for $j \in \NN$, where each $\alpha_j$ has $i-1$ terms and $(|\nu_j|)_{j \in \NN}$ is strictly increasing.

  Since $\Aa$ has finitely many accepting SCC-paths, we may extract a subsequence and assume that all $\alpha_j$ have the same left effect on every accepting SCC-path, and that all $\beta_j$ have the same right effect on every accepting SCC-path.

  Let $\sigma_l = \alpha_0$ and $\sigma_r = \beta_0$.
  By Remark~\ref{rmk:effects} and the greedy computation of left effects, whether a sequence of the form $\alpha_j\gamma\beta_j$ is blocking for an accepting SCC-path does not change if we replace $\alpha_j$ and $\beta_j$ by $\sigma_l$ and $\sigma_r$.
  We use two consequences of the minimality of the $\sigma_j$.
  First, for every $j$ and every sequence $\gamma$ whose terms are all proper factors of $\nu_j$, the sequence $\alpha_j\gamma\beta_j$ is not blocking for $\Aa$.
  Indeed, we have $\alpha_j\gamma\beta_j \pobs \sigma_j$, by mapping all the terms of $\gamma$ to $\nu_j$.
  If $\alpha_j\gamma\beta_j$ were blocking, then by minimality we would have $\sigma_j \pobs \alpha_j\gamma\beta_j$.
  Let $f$ be such an embedding, write $\sigma_j = (\mu_1, \ldots, \mu_L)$, so that $\nu_j = \mu_i$, and let $g(t) = f(t) - t$ for $t \le i$.
  The first $i-1$ terms of $\alpha_j\gamma\beta_j$ are $\mu_1, \ldots, \mu_{i-1}$.
  Suppose that $\mu_i$ is mapped to one of these terms, so that $g(i) \le -1$.
  As $g(1) \geq 0$ and $g(t+1) \geq g(t) - 1$ since $f$ is non-decreasing, there is $t < i$ such that $f(t) = f(t+1) = t$, so $\mu_{t+1}$ is a factor of the adjacent term $\mu_t$, which contradicts the fact that $\sigma_j$ is reduced.
  Symmetrically, $\mu_i$ is not mapped to a term of $\beta_j$.
  Hence $\nu_j$ is mapped to a term of $\gamma$, which is impossible as the terms of $\gamma$ are shorter than $\nu_j$.
  Applying this to the empty sequence $\gamma$, we obtain that $\sigma_l\sigma_r$ is not blocking for $\Aa$, while $\sigma_l(\nu_j)\sigma_r$ is blocking for $\Aa$ for every $j$.

  We now construct $P_*$.
  Let $\Pp$ denote the set of portals that survive
  $(\sigma_l,\sigma_r)$ in at least one accepting SCC-path.
  Since $\sigma_l\sigma_r$ is not blocking for $\Aa$,
  the set $\Pp$ is nonempty by Remark~\ref{rmk:effects}.
  For every accepting SCC-path $\SCCpath$, by the greedy computation of left effects, $\sigma_l(\nu_j)\sigma_r$ is blocking for $\SCCpath$ if and only if $\nu_j$ is blocking for all the portals that survive $(\sigma_l, \sigma_r)$ in $\SCCpath$.
  Hence, for every $j$, the word $\nu_j$ is blocking for every portal in $\Pp$.

  We claim that some $\preceq$-minimal portal in $\Pp$ has
  infinitely many minimal blocking factors.
  Suppose otherwise. For each $\preceq$-minimal $P\in\Pp$,
  let $B_P$ be its finite set of minimal blocking
  factors, and set
  \[
    M' = \max\bigl\{
      |w| : P\in\Pp \text{ is $\preceq$-minimal},\ w\in B_P
    \bigr\}.
  \]
  Select $j$ so that $|\nu_j|> M'$.
  For each $\preceq$-minimal $P \in \Pp$, the word $\nu_j$ is blocking for $P$, hence it contains a minimal blocking factor of $P$, which belongs to $B_P$ and is therefore a proper factor of $\nu_j$.
  Let $\gamma$ be a sequence made of these factors, one for each $\preceq$-minimal portal of $\Pp$, in any order.
  Every portal $P \in \Pp$ is $\succeq$ some $\preceq$-minimal portal $P' \in \Pp$, so every blocking factor of $P'$ is blocking for $P$: hence $\gamma$ contains a blocking factor for every portal in $\Pp$.
  Let $K$ be the maximal number of portals of an SCC-path.
  The sequence $\gamma^{K}$ obtained by concatenating $K$ copies of $\gamma$ is blocking for every sequence of at most $K$ portals of $\Pp$, and in particular for the portals that survive $(\sigma_l, \sigma_r)$ in each accepting SCC-path.
  By the greedy computation of left effects, $\sigma_l\gamma^{K}\sigma_r$ is thus blocking for $\Aa$, and so is $\alpha_j\gamma^{K}\beta_j$.
  As all the terms of $\gamma^{K}$ are proper factors of $\nu_j$, this contradicts the first consequence of minimality above, which proves the claim.

  Fix such a portal $P_*$.
  For every portal $P \in \Pp$ such that $P \not\simeq  P_*$,
  the $\preceq$-minimality of $P_*$ implies that $P \not\preceq P_*$, so there exists a positional word $\tau_P$ that is
  blocking for $P$ but not for $P_*$.

  We now extend $\sigma_l$ and $\sigma_r$ until every
  surviving portal is $\succeq P_*$.
  Note that extending $\sigma_l$ and $\sigma_r$ can only increase left effects and decrease right effects, so the set of portals that survive in an accepting SCC-path can only shrink, and stays included in $\Pp$.
  As long as some accepting SCC-path $\pi=P_0\rightarrow\cdots\rightarrow P_k$ contains a portal $P_j \nsucceq P_*$ that survives $(\sigma_l,\sigma_r)$ in $\pi$, we do the following.
  Write
  $a=(\sigma_l\gg\pi)$ and
  $b=(\pi\ll\sigma_r)$, so that the portals that survive in $\pi$ are
  $P_{a+1},\ldots,P_{b-1}$, with $a< j <b$.
  We distinguish three cases.
  \begin{itemize}
    \item If $P_{a+1}$ is not equivalent to  $P_*$, we append $\tau_{P_{a+1}}$ to $\sigma_l$.

    \item Otherwise, if $P_{b-1}$ is not equivalent to  $P_*$, we prepend $\tau_{P_{b-1}}$ to $\sigma_r$.

    \item Otherwise, both $P_{a+1}$ and $P_{b-1}$ are equivalent to $P_*$.
      Since $P_j \nsucceq P_*$, there is a word $\tau'$ that is blocking for $P_*$ but not for $P_j$.
      We append $\tau'$ to $\sigma_l$.
      As $\tau'$ is blocking for $P_{a+1}$ but not for $P_j$, the new left effect of $\sigma_l$ on $\pi$ is between $a+1$ and $j-1$; in particular, $P_{b-1}$ still survives in $\pi$.
  \end{itemize}

  To prove termination, consider the potential
  \[
    \Phi(\sigma_l,\sigma_r)
    =
    \sum_{\pi \text{ accepting}} \max((\pi\ll\sigma_r) -(\sigma_l\gg\pi) -1, 0),
  \]
  which counts the pairs made of an accepting SCC-path and a portal that survives in it.
  In all three cases, a surviving portal of $\pi$ stops surviving, and no new portal survives, so $\Phi$ decreases. Since $\Phi$ is a nonnegative integer, the
  procedure terminates.

  We now show that during the procedure, some accepting SCC-path always contains a surviving portal equivalent to $P_*$.
  This holds at the start, as $P_* \in \Pp$.
  In the first two cases, the appended or prepended word $\tau_P$ is not blocking for $P_*$, hence for no portal equivalent to $P_*$; by the greedy computation of effects, the effects of the new sequences cannot go past a portal equivalent to $P_*$, so every surviving portal equivalent to $P_*$ still survives.
  In the third case, $P_{b-1}$ is equivalent to $P_*$ and still survives in $\pi$.

  When the procedure stops, every portal that survives in an accepting SCC-path is $\succeq P_*$, and some accepting SCC-path contains a surviving portal equivalent to $P_*$.
  In particular, this SCC-path is not blocked by $\sigma_l\sigma_r$, so $\sigma_l\sigma_r$ is not blocking for $\Aa$.
  Hence the portal $P_*$ and the sequences $\sigma_l, \sigma_r$ satisfy
  P1, P2 and P3, completing the proof.
\end{proof}

\paragraph*{From isolation to a distance-preserving reduction.}
Properties P1, P2 and P3 state that the sequences $\sigma_l$ and $\sigma_r$ isolate a hard portal $P_*$, which allows us to transform a word $v$ into a word $w$ such that $w$ is in the language of $\Aa$ if and only if $v$ is in the language of $P_*$.
For property testing, this transformation must also preserve the distance to the language, up to a constant factor.
Given a word $v$, we use $\sigma_l$ and $\sigma_r$ to construct words $w_L$ and $w_R$ such that $w = w_Lvw_R$ satisfies the two following properties, where $c > 0$ is a constant that only depends on the automaton:
\[
  v\in \lang{P_*}\ \Longrightarrow\ w\in \lang{\Aa}
  \qquad\text{and}\qquad
  d\bigl(v,\lang{P_*}\bigr)\ge\eps|v|
  \ \Longrightarrow\
  d\bigl(w,\lang{\Aa}\bigr)\ge c\eps|w|.
\]
For the second property, it is not enough for $w_L$ and $w_R$ to contain a single occurrence of $\sigma_l$ and $\sigma_r$, as a few substitutions in $w_L$ or $w_R$ would suffice to remove them.
We therefore let $w_L$ and $w_R$ contain $\Theta(\eps|v|)$ disjoint occurrences of the factors of $\sigma_l$ and $\sigma_r$, which we combine with the blocking factors of $v$ to obtain $\Omega(\eps|v|)$ disjoint blocking sequences in $w$.

We then turn a tester for $\lang{\Aa}$ into a tester for $\lang{P_*}$, which on input $v$ runs the tester for $\lang{\Aa}$ with precision $c\eps$ on $w$.
A query to a position of $v$ in $w$ is answered by querying $v$, and a query to a position of $w_L$ or $w_R$ requires no query.
Hence the lower bound for $P_*$ transfers to $\lang{\Aa}$.

\begin{lemma}
  \label{lem:seq-left}
  Let $\SCCpath = P_0 \xrightarrow{a_1} \cdots  P_\ell$ be an SCC-path with a non-empty language,
  denote $P_j = \portal{s_j}{x_j}{t_j}{y_j}$ for each $j = 0,\ldots,\ell$,
  let $i \in \set{0, \ldots, \ell}$,
  and let $\sigma_l = (\nu_{1,l}, \ldots, \nu_{k,l})$ be a sequence such that $\lefteffect{\sigma_l}{\SCCpath} < i$.

  Then, we can compute $E_l, F_l \in\NN$ such that for any integer $N \in \NN$, there is a word $w_L$ of length at most $E_l  + N F_l$ such that $|w_L| = x_i-x_0 \pmod{p}$, there is a run reading $w_L$ from $s_0$ to $s_i$ in $\Aa$, and $\timedword{x_0}{w_L}$ contains $N$ occurrences of $\nu_{1,l}$, followed by $N$ occurrences of $\nu_{2,l}$, and so on up to $\nu_{k,l}$, all disjoint.
\end{lemma}
\begin{proof}
  Let $\rho$ be the maximum of the reachability constants of the recurrent SCCs of $\Aa$; by \cref{fact:periodicity}, $\rho \le 3|\Aa|^2$.
  We first note that if there is a path of length $r$ between two states $q, q'$ of the same SCC, then there is a path between $q$ and $q'$ of length $r' < \rho + p$ with $r' = r \pmod{p}$.
  This is clear if the SCC is transient (then $r = 0$), and otherwise it follows from \cref{fact:periodicity}, as $p$ is a multiple of the period of the SCC.

  We define $w_L$ by induction on $k$, the length of $\sigma_l$.
  As $\lang{\SCCpath}$ is nonempty, for all $j \in \set{0, \ldots, \ell}$, there exists a word $u_j$ of length $y_j-x_j \pmod{p}$ that labels a path from $s_j$ to $t_j$, and we can choose it of length less than $\rho+p$.
  As a result, for all $z \in \set{0, \ldots, \ell}$, the word $w_z = u_0 a_1 u_1 \cdots u_{z-1} a_{z}$ labels a path of length $x_z-x_0 \pmod{p}$ from $s_0$ to $s_z$ in $\Aa$, and has length at most $C_0 = |\Aa|(\rho+p+1)$, as $\ell+1$ is at most the number of SCCs of $\Aa$.
  If $k=0$, we can simply set $w_L = w_i$.

  Let $k> 0$, and assume that the lemma holds for $k-1$.
  Let $j = \lefteffect{(\nu_{1,l})}{\SCCpath}$; we have $j \le \lefteffect{\sigma_l}{\SCCpath} < i$.
  As $\nu_{1,l}$ is not blocking for $P_{j+1}$, there is a word $v$ that labels a path from $s_{j+1}$ to $t_{j+1}$ such that $\timedword{x_{j+1}}{v}$ belongs to $\lang{P_{j+1}}$ and has $\nu_{1,l}$ as a factor.
  By the remark above, we can replace the part of this path before (resp.\ after) the occurrence of $\nu_{1,l}$ by a path between the same states, of the same length modulo $p$ and of length less than $\rho+p$.
  We obtain a word $v$ with the same properties and with $|v| \leq |\nu_{1,l}| + 2(\rho+p)$.
  As $s_{j+1}$ and $t_{j+1}$ are in the same SCC, we can extend $v$ into a word $v'$ of length at most $|\nu_{1,l}| + C_1$, with $C_1 = 2(\rho+p) + |\Aa|$, that labels a cycle from $s_{j+1}$ to itself.

  Let $\sigma' = (\nu_{2,l}, \ldots, \nu_{k,l})$ and $\SCCpath' = P_{j+1} \xrightarrow{a_{j+2}} \cdots  P_\ell$.
  By the greedy computation of left effects, we have $\lefteffect{\sigma_l}{\SCCpath} = j+1+\lefteffect{\sigma'}{\SCCpath'}$, hence $\lefteffect{\sigma'}{\SCCpath'} < i-j-1$, where $i-j-1$ is the index of $P_i$ in $\SCCpath'$.
  By induction hypothesis, there is a word $w'$ of length at most $kC_0  + pN\sum_{t=2}^{k}(|\nu_{t,l}|+C_1)$ such that $|w'| = x_i-x_{j+1} \pmod{p}$, there is a run reading $w'$ from $s_{j+1}$ to $s_i$ in $\Aa$, and $\timedword{x_{j+1}}{w'}$ contains $N$ occurrences of $\nu_{2,l}$, followed by $N$ occurrences of $\nu_{3,l}$, and so on up to $\nu_{k,l}$, all disjoint.

  We set $w_L = w_{j+1} (v')^{pN} w'$.
  It labels a path from $s_0$ to $s_i$, its length is $x_i-x_0 \pmod{p}$, and it satisfies
  \[
    |w_L|
    \le|w_{j+1}|+pN|v'|+ |w'|
    \le (k+1)C_0 + pN\sum_{t=1}^k(|\nu_{t,l}|+C_1).
  \]
  We set $E_l =  (k+1)C_0$  and $F_l = p\sum_{t=1}^k(|\nu_{t,l}|+C_1)$.
  In $\timedword{x_0}{w_L}$, the $c$-th copy of $v'$ starts at index $x_{j+1} + c|v'|$ modulo $p$.
  Hence for every $c$ that is a multiple of $p$, this copy is the positional word $\timedword{x_{j+1}}{v'}$, which contains an occurrence of $\nu_{1,l}$.
  This gives $N$ disjoint occurrences of $\nu_{1,l}$, followed by the occurrences contained in $w'$, which concludes the proof.
\end{proof}

\begin{lemma}
  \label{lem:seq-right}
  Let $\SCCpath = P_0 \xrightarrow{a_1} \cdots  P_\ell$ be an SCC-path with a non-empty language,
  denote $P_j = \portal{s_j}{x_j}{t_j}{y_j}$ for each $j = 0,\ldots,\ell$,
  let $i \in \set{0, \ldots, \ell}$,
  and let $\sigma_r = (\nu_{1,r}, \ldots, \nu_{k,r})$ be a sequence such that $\righteffect{\sigma_r}{\SCCpath} > i$.

  Then we can compute $E_r, F_r \in\NN$ such that for all $N \in \NN$, there is a word $w_R$ of length at most $E_r  + NF_r$ such that $|w_R| = y_\ell-y_i \pmod{p}$, there is a run reading $w_R$ from $t_i$ to $t_\ell$ in $\Aa$, and $\timedword{y_i}{w_R}$ contains $N$ occurrences of $\nu_{1,r}$, followed by $N$ occurrences of $\nu_{2,r}$, and so on up to $\nu_{k,r}$, all disjoint.
\end{lemma}

\begin{proof}
  The proof is symmetric to the one of the previous lemma, processing $\sigma_r$ from its last term, with one difference in the construction of the cycle.
  We use the same constant $\rho$ and the same remark on lengths of paths, and define $w_R$ by induction on $k$.
  For all $z \in \set{0, \ldots, \ell}$, the word $w^R_z = a_{z+1} u_{z+1} \cdots a_\ell u_\ell$, with the words $u_j$ of the previous proof, labels a path from $t_z$ to $t_\ell$ of length $y_\ell - y_z \pmod p$ and length at most $C_0$.
  If $k = 0$, we set $w_R = w^R_i$.

  Let $k > 0$, and let $j = \righteffect{(\nu_{k,r})}{\SCCpath}$; we have $j \geq \righteffect{\sigma_r}{\SCCpath} > i$.
  As $\nu_{k,r}$ is not blocking for $P_{j-1}$, as before there is a word $v$ of length at most $|\nu_{k,r}| + 2(\rho+p)$ such that $\timedword{x_{j-1}}{v} \in \lang{P_{j-1}}$ contains $\nu_{k,r}$.
  Unlike in the previous proof, we close the cycle at $t_{j-1}$, with a run from $t_{j-1}$ to $s_{j-1}$ whose length is equal to $x_{j-1} - y_{j-1}$ modulo $p$.
  Such a run exists since, as $\lang{P_{j-1}}$ is non-empty, the residue $x_{j-1} - y_{j-1}$ is equal, modulo the period of the SCC, to the length of every path from $t_{j-1}$ to $s_{j-1}$, and $p$ is a multiple of this period, so \cref{fact:periodicity} gives such a run, which we choose of length less than $\rho + p$ by the remark.
  Let $v''$ be its label and $v' = v''v$: it labels a cycle on $t_{j-1}$ of length $0$ modulo $p$, of length at most $|\nu_{k,r}| + C_1'$ with $C_1' = 3(\rho+p)$, and every copy of $v$ in $\timedword{y_{j-1}}{(v')^N}$ is the positional word $\timedword{x_{j-1}}{v}$, which contains an occurrence of $\nu_{k,r}$.

  Let $\sigma' = (\nu_{1,r}, \ldots, \nu_{k-1,r})$ and $\SCCpath' = P_0 \xrightarrow{a_1} \cdots P_{j-1}$.
  As in the previous proof, $\righteffect{\sigma'}{\SCCpath'} > i$, and the induction hypothesis gives a word $w'$ of length at most $kC_0 + N\sum_{t=1}^{k-1}(|\nu_{t,r}|+C_1')$ that labels a path from $t_i$ to $t_{j-1}$, with $|w'| = y_{j-1} - y_i \pmod p$, such that $\timedword{y_i}{w'}$ contains $N$ occurrences of $\nu_{1,r}$, followed by $N$ occurrences of $\nu_{2,r}$, and so on up to $\nu_{k-1,r}$, all disjoint.
  We set $w_R = w'(v')^N w^R_{j-1}$, which labels a path from $t_i$ to $t_\ell$ of length $y_\ell - y_i$ modulo $p$, contains the required occurrences, and has length at most $E_r + NF_r$ with $E_r = (k+1)C_0$ and $F_r = \sum_{t=1}^k(|\nu_{t,r}|+C_1')$.
\end{proof}

\begin{restatable}{lemma}{PropertiesImpliesHard}
  \label{lem:P1-P3-then-hard}
  If there exist a portal $P_*$ and sequences $\sigma_l$, $\sigma_r$ satisfying properties P1, P2 and P3 then $\lang{\Aa}$  is hard.
\end{restatable}
\begin{proof}[Proof of \cref{lem:P1-P3-then-hard}]
  By P3, there is an accepting SCC-path $\SCCpath = P_0 \xrightarrow{a_1} \cdots P_\ell$ and an index $i$ such that $P_i$ survives $(\sigma_l, \sigma_r)$ in $\SCCpath$ and $P_i \equivportals P_*$.
  We write $P_j = \portal{s_j}{x_j}{t_j}{y_j}$ for all $j$, so that $s_0 = q_0$, $x_0 = 0$ and $t_\ell \in F$.
  As $P_i$ and $P_*$ have the same blocking factors, they have the same minimal blocking factors, so $P_i$ has infinitely many minimal blocking factors by P2.
  Let $S$ be the SCC of $s_i$ and $t_i$; it is recurrent, since for a transient portal all minimal blocking factors have length $1$.
  We define an NFA $\Bb$ over $\Sigma$ that recognises $L_i = \set{v \mid \timedword{x_i}{v} \in \lang{P_i}}$.
  Its states are the pairs $(q, h) \in S \times \ZZ/p\ZZ$ that are reachable from the initial state $(s_i, x_i)$, its transitions are $(q, h) \xrightarrow{a} (q', h+1)$ for $q \xrightarrow{a} q'$ in $S$, and its only final state is $(t_i, y_i)$.
  It has at most $m = p|\Aa|$ states, and it recognises $L_i$ by construction.
  It is strongly connected, since if $(q, h)$ is reached from $(s_i, x_i)$ by a word $u$ and $u'$ labels a path from $q$ to $s_i$ in $S$, then $u'(uu')^{p-1}$ labels a path from $(q, h)$ back to $(s_i, x_i)$, as $p$ divides its length plus $|u|$.
  Its period is $p$, since every transition increases the second component by $1$ modulo $p$, so that $p$ divides the length of every cycle, and the closed walk of length $p$ on the recurrent state $s_i$ (see the proof of \cref{lem:far-from-path-implies-many-blocking}) gives a cycle of length $p$ on $(s_i, x_i)$.
  Finally, a positional word $\timedword{z}{u}$ is blocking for $\Bb$ if and only if $\timedword{x_i + z}{u}$ is blocking for $P_i$, so $\MBF(\Bb)$ is infinite.
  We show that a non-adaptive tester for $\lang{\Aa}$ yields a non-adaptive tester for $L_i$ with a similar number of queries, and conclude with \cref{thm:scc-lb}.

  Let $K$ be the maximal number of portals of an accepting SCC-path of $\Aa$.
  We first show that for every sequence $\gamma = (\gamma_1, \ldots, \gamma_K)$ of positional words that are blocking for $P_i$, the sequence $\sigma_l^{(K)}\gamma\sigma_r^{(K)}$ is strongly blocking for every accepting SCC-path $\SCCpath'$ of $\Aa$ (recall that $\sigma^{(K)}$ is obtained by repeating each term of $\sigma$ $K$ times).
  By P1 and P3, every portal that survives $(\sigma_l, \sigma_r)$ in $\SCCpath'$ is $\geqportals P_* \equivportals P_i$, so each $\gamma_c$ is blocking for all these portals.
  As $\sigma_l$ is blocking for the portals of $\SCCpath'$ up to $\lefteffect{\sigma_l}{\SCCpath'}$, the sequence $\sigma_l^{(K)}$ is strongly blocking for them, and similarly for $\sigma_r^{(K)}$ and the portals from $\righteffect{\sigma_r}{\SCCpath'}$ on.
  There are at most $K$ surviving portals in $\SCCpath'$, and we can block each of them with a different $\gamma_c$, in order.
  Hence $\sigma_l^{(K)}\gamma\sigma_r^{(K)}$ is strongly blocking for $\SCCpath'$.

  Let $T$ be a non-adaptive tester for $\lang{\Aa}$ with perfect completeness, which makes $f(\eps, n)$ queries on inputs of length $n$ with precision $\eps$.
  We describe a non-adaptive tester $T'$ for $L_i$.
  Its input is a word $v$ of length $n$, and a precision $\eps$; we write $\mu = \timedword{x_i}{v}$.
  Let $N = \lceil \eps n\rceil$.
  As $\lefteffect{\sigma_l}{\SCCpath} < i < \righteffect{\sigma_r}{\SCCpath}$, \cref{lem:seq-left,lem:seq-right} give words $w_L$ and $w_R$ of length at most $E + NF$, for some constants $E, F$, such that $w_L$ labels a path from $q_0$ to $s_i$ and has length $x_i \pmod p$, $w_R$ labels a path from $t_i$ to $t_\ell$, $\timedword{0}{w_L}$ contains $N$ occurrences of each term of $\sigma_l$, all disjoint and in the order of $\sigma_l$, and similarly for $\timedword{y_i}{w_R}$ and $\sigma_r$.
  Let $w = w_L v w_R$.
  The tester $T'$ works as follows.
  \begin{itemize}
    \item If $L_i$ contains no word of length $n$, then $T'$ rejects.
    \item If $n < Y/\eps$, where $Y = \max(12Km^2, 2E+2F+n_0)$, then $T'$ queries all of $v$ and accepts if and only if $v \in L_i$.
    \item Otherwise, $T'$ simulates $T$ on the word $w$ with precision $\eps/X$, for a constant $X$ defined below: a query of $T$ to a position of $w_L$ or $w_R$ is answered directly, and a query to a position of $v$ is answered by querying $\mu$ at the corresponding position. $T'$ gives the same answer as $T$.
  \end{itemize}
  The positions queried by $T'$ only depend on $n$, $\eps$ and the random choices of $T$, so $T'$ is non-adaptive, and it makes at most $\max(Y/\eps, f(\eps/X, |w|))$ queries.
  In the third case, as $L_i$ contains a word of length $n$, we have $x_i + n = y_i \pmod p$, so $w_R$ is placed at a position equal to $y_i$ modulo $p$, $|w|$ is equal to the residue $r$ of $\SCCpath$ modulo $p$, and $|w| \ge n \ge n_0$.

  We show that $T'$ is a tester for $L_i$; it suffices to consider the third case.
  If $v \in L_i$, then $v$ labels a path from $s_i$ to $t_i$ and $w$ is in $\lang{\Aa}$, so $T$ and $T'$ accept.
  Now assume that $v$ is $\eps$-far from $L_i$.
  By \cref{lemma:many-blocking} applied to $\Bb$, $\mu$ contains at least $D = \lceil\eps n/(6m^2)\rceil$ disjoint occurrences of blocking factors for $P_i$; note that $D \le N$, and $D > K$ for $n$ large enough.
  Let $w'$ be a word obtained from $w$ by substituting fewer than $D - K$ letters.
  As $N \geq D$, for each term of $\sigma_l$, at least $K$ of its $N$ occurrences in $w_L$ are not modified, and similarly for $\sigma_r$ in $w_R$; moreover, at least $K$ of the $D$ occurrences of blocking factors for $P_i$ in $v$ are not modified.
  Hence $\timedword{0}{w'}$ contains a sequence of the form $\sigma_l^{(K)}\gamma\sigma_r^{(K)}$ as above, which is strongly blocking for every accepting SCC-path, so $w' \notin \lang{\Aa}$ by \cref{cor:many-strongly-bs-implies-not-A}.
  As a consequence, $d(w, \lang{\Aa}) \geq D - K \geq \eps n/(12m^2)$ for $n$ large enough.
  As $|w| \le n + 2E + 2F(\eps n + 1) \le (2F+3)n$ for $n$ large enough, we obtain $d(w, \lang{\Aa}) \geq \frac{\eps}{X}|w|$ with $X = 12m^2(2F+3)$, so $T$ rejects $w$ with probability at least $2/3$, and so does $T'$.

  By \cref{thm:scc-lb} applied to $\Bb$, for all small enough $\eps$ and all large enough $n$ such that $L_i$ contains a word of length $n$, the tester $T'$ makes $\Omega(\epslogeps)$ queries.
  Hence $f(\eps/X, |w|) = \Omega(\epslogeps)$ for infinitely many lengths $|w|$, which shows that every non-adaptive tester for $\lang{\Aa}$ makes $\Omega(\epslogeps)$ queries.
  Together with \cref{thm:gen-ub}, this shows that $\lang{\Aa}$ is hard.
\end{proof}

\begin{proposition}
  If $\Aa$ has infinitely many minimal blocking sequences,
  then $\lang{\Aa}$ is hard.
\end{proposition}
\begin{proof}
  We combine Lemmas~\ref{lem:hard-aut-to-hard-portal} and~\ref{lem:P1-P3-then-hard}.
\end{proof}

\section{Trivial and easy languages}
\label{sec:trivial-easy}
\subsection{Upper bound for easy languages}

We first establish that an automaton with finitely many minimal blocking sequences (at residue $r$) is easy (or trivial) to test.
As in \cref{sec:general}, we fix a residue $r \in R(\lang{\Aa})$, and we consider positional words $\mu$ of length $n \geq n_0$ with $n = r \pmod{p}$.

\begin{lemma}\label{lem:far-bounded-mbs}
  Let $\Aa$ be an NFA with a finite number of minimal blocking sequences.
  There are constants $B, D$ such that the following holds for every $0 < \eps \le 1$ and every positional word $\mu$ of length $n \geq D/\eps$ that is $\eps$-far from $\timedlang{\Aa}$.
  For every $\SCCpath = P_0 \xrightarrow{a_1} \cdots P_k \in \Pi_r$, $\mu$ can be partitioned into $\mu = \tau_0\tau_1\cdots\tau_k$ such that for every $h = 0,\ldots, k$, $\tau_h$ contains at least $\frac{\eps n}{D}$ disjoint occurrences of blocking factors for~$P_h$, each of length at most~$B$.
\end{lemma}
Note that the assumption is that $\mu$ is far from $\timedlang{\Aa}$, and not only from $\lang{\SCCpath}$: in \cref{example:hard-to-easy}, the word $(be^kc)^m$ is far from the language of the SCC-path made of the first SCC only, but it does not contain short blocking factors for this SCC.
\begin{proof}
  As $\MBS(\Aa)$ is finite, there is a bound $B$ on the length of the terms of minimal blocking sequences: indeed, if $\sigma \equiv \sigma'$, then every term of $\sigma$ is a factor of a term of $\sigma'$ and conversely, so equivalent sequences have the same maximal term length.

  Let $C$ be the constant of \cref{lem:far-from-path-implies-many-blocking}, and assume $n \geq C/\eps$.
  By \cref{cor:far-L-many-bs}, $\mu$ contains pairwise disjoint occurrences of $N \ge \eps n/C$ blocking sequences $\sigma_1, \ldots, \sigma_N$ for $\Aa$, and all their factor occurrences are pairwise disjoint.
  For every $j$, by \cref{lem:bound-length-min-blocking}, there is a minimal blocking sequence $\sigma'_j \pobs \sigma_j$, whose terms have length at most $B$.

  Fix $\SCCpath = P_0 \xrightarrow{a_1} \cdots P_k \in \Pi_r$.
  For every $j$, as $\sigma'_j$ is blocking for $\SCCpath$, there are non-decreasing indices such that the $h$-th of them points to a term $t_{j,h}$ of $\sigma'_j$ that is blocking for $P_h$.
  As $\sigma'_j \pobs \sigma_j$, each $t_{j,h}$ is a factor of a term of $\sigma_j$, and we can choose these terms of $\sigma_j$ with non-decreasing indices.
  We write $G_{j,h}$ for the corresponding factor occurrence of $\mu$ (an element of the occurrence of $\sigma_j$), and $b_{j,h}$ for an occurrence of $t_{j,h}$ inside $G_{j,h}$.
  Hence $b_{j,h}$ is an occurrence of a blocking factor for $P_h$ of length at most $B$, and for every $j$, the occurrences $G_{j,0}, G_{j,1}, \ldots, G_{j,k}$ appear in this order in $\mu$ (some of them may be equal).
  Moreover, the occurrences $G_{j,h}$ and $G_{j',h'}$ are disjoint if $j \neq j'$.

  We now cut $\mu$. Let $N' = \lfloor N/(k+1) \rfloor$.
  We define positions $c_{-1} = -1 < c_0 \le \cdots \le c_k = n-1$ and pairwise disjoint sets $J_0, \ldots, J_k \subseteq \set{1,\ldots,N}$ of size $N'$ as follows.
  Let $J = \set{1, \ldots, N}$.
  For $h = 0, \ldots, k$, let $J_h$ be the set of the $N'$ indices $j \in J$ for which the right end of $G_{j,h}$ is the smallest; let $c_h$ be the largest right end of the $G_{j,h}$ for $j \in J_h$ (or $n-1$ if $h = k$), and remove $J_h$ from $J$.
  We show by induction on $h$ that for every $j$ that remains in $J$ before step $h$, the occurrence $G_{j,h}$ starts after $c_{h-1}$.
  This is clear for $h = 0$.
  Let $j$ be an index that remains in $J$ after step $h$.
  The right end of $G_{j,h}$ is larger than or equal to $c_h$, and $G_{j,h}$ is disjoint from the occurrence $G_{j^*,h}$ whose right end is $c_h$, as $j \neq j^*$; hence $G_{j,h}$ starts after $c_h$.
  As $G_{j,h+1}$ is equal to $G_{j,h}$ or appears after it, $G_{j,h+1}$ also starts after $c_h$, which completes the induction.

  Let $\tau_h = \mu[c_{h-1}+1\dd c_h]$ for $h = 0, \ldots, k$.
  For every $j \in J_h$, the occurrence $G_{j,h}$, and thus $b_{j,h}$, is included in $\tau_h$.
  As the $G_{j,h}$ for $j \in J_h$ are pairwise disjoint, $\tau_h$ contains $N' \ge \eps n/D$ disjoint occurrences of blocking factors for $P_h$ of length at most $B$, for $D = 2C(k+1)$ and $n$ large enough.
  As the number of portals of an SCC-path is bounded, we can choose $D$ independently of $\SCCpath$.
\end{proof}

\begin{corollary}
  \label{coro:fin-bs-then-easy}
  If $\Aa$ has finitely many minimal blocking sequences at every residue $r \in R(\lang{\Aa})$, then there is a non-adaptive tester for $\lang{\Aa}$ that uses $\cO(1/\eps)$ queries.
\end{corollary}
\begin{proof}
  We use the same algorithm as in the proof of \cref{thm:gen-ub}, with the following changes.
  For each residue $r$, let $B_r, D_r$ be the constants given by \cref{lem:far-bounded-mbs}.
  We assume $\eps \le 1$. For an input of length $n \geq n_0$ whose residue modulo $p$ is not in $R(\lang{\Aa})$, $\lang{\Aa}$ has no word of length $n$, and the tester rejects without querying.
  For an input of length $n = r \pmod p$ with $r \in R(\lang{\Aa})$ and $n \geq \max(n_0, B_r, D_r/\eps)$ (otherwise, the tester queries all of $u$), we call \textsc{Sampler} $t$ times with parameters $N' = \min(\lceil \eps n/D_r\rceil, \lfloor n/B_r\rfloor)$ and $B_r$, where $t = \lceil \log_3(3M)\rceil$ and $M$ is the number of pairs made of an SCC-path of $\Pi_r$ and one of its portals, and we do not extend the factor occurrences that it returns.
  Perfect completeness holds for the same reason as before.
  If $\mu$ is $\eps$-far from $\timedlang{\Aa}$, then for every $\SCCpath \in \Pi_r$, let $\mu = \tau_0 \cdots \tau_k$ be given by \cref{lem:far-bounded-mbs}.
  For each $h$, $\tau_h$ contains at least $N'$ disjoint occurrences of blocking factors for $P_h$ of length at most $B_r$, so with probability at least $1 - 1/(3M)$, the factor occurrences returned by the calls to \textsc{Sampler} cover one of them.
  By the union bound over all pairs, with probability at least $2/3$ this holds for all SCC-paths of $\Pi_r$ and all their portals, and as the $\tau_h$ are consecutive, the covered occurrences form, for each $\SCCpath \in \Pi_r$, an occurrence of a strongly blocking sequence for $\SCCpath$; hence the tester rejects.
  As $N' \geq 1$, the query complexity becomes $\cO(n(1+\log B_r)/N') = \cO(\max(D_r/\eps, B_r)\log B_r) = \cO(1/\eps)$.
\end{proof}

This already gives us a dichotomy: all languages either require $\Theta(\epslogeps)$ queries to be tested, or can be tested with $\cO(1/\eps)$ queries.

\subsection{Separation between trivial and easy languages}

It remains to show that languages that can be tested with $\cO(1/\eps)$ queries have query complexity either $\Theta(1/\eps)$, or $0$ for large enough $n$.
Our proof uses the class of \emph{trivial} regular languages identified by Alon et al.~\cite{alon2001regular}, which we revisit next.

An example of a trivial language is the language of words containing at least one $a$ over the alphabet $\set{a,b}$.
For any word $u$, replacing any letter by $a$ yields a word in this language, hence the distance from $u$ to the language is at most $1$.
Therefore, for $n > 1/\eps$, no word of length $n$ is $\eps$-far from it, and the trivial property tester that answers ``yes'' without sampling any letter is correct.

Alon et al.~\cite{alon2001regular} define non-trivial languages as follows.
\begin{restatable}[{\cite[Definition 3.1]{alon2001regular}}]{definition}{trivialalondef}\label{def:trivial-alon}
	A language~$L$ is non-trivial if there exists a constant $\eps_0 > 0$ such that for infinitely
	many values of $n$ the following are true:
	\begin{itemize}
		\item  the set $L\cap\Sigma^n$ is non-empty, and
		\item there exists a word $w \in \Sigma^n$ such that $d(w, L) \ge \eps_0 n$.
	\end{itemize}
  We say that $L$ is non-trivial \emph{at residue} $r$ if this holds for infinitely many values of $n$ equal to $r$ modulo $p$.
\end{restatable}
Clearly, $L$ is non-trivial if and only if it is non-trivial at some residue $r$, and in that case $r \in R(L)$.
It is easy to see that if a language is trivial in the above sense (i.e. not non-trivial), then for large enough input length $n$, the answer to testing membership in~$L$ only depends on $n$, and the algorithm does not need to query the input.
Alon et al.~\cite[Property 2]{alon2001regular} show that if a language is non-trivial, then testing it requires $\Omega(1/\eps)$ queries for small enough $\eps > 0$.

To obtain our characterization of \emph{trivial} languages, we show that, for $r \in R(\lang{\Aa})$, $\MBS_r(\Aa)$ is non-empty if and only if $\lang{\Aa}$ is non-trivial at residue $r$.
It follows that if $\MBS_r(\Aa)$ is empty for all $r \in R(\lang{\Aa})$, then testing $\lang{\Aa}$ requires $0$ queries for large enough $n$.
Furthermore, by the result of Alon et al.~\cite{alon2001regular}, if $\MBS_r(\Aa)$ is non-empty for some $r \in R(\lang{\Aa})$, then testing $\lang{\Aa}$ requires $\Omega(1/\eps)$ queries.

Recall that we focus on infinite languages, since we know that all finite ones are trivial (Remark~\ref{rmk:finite}).
\begin{lemma}\label{lem:trivialiffMBSempty}
	For every $r \in R(\lang{\Aa})$, $\MBS_r(\Aa)$ is empty if and only if $L = \lang{\Aa}$ is trivial at residue $r$.
\end{lemma}

We prove the two directions separately.
By \cref{lem:bound-length-min-blocking}, $\MBS_r(\Aa)$ is empty if and only if $\Aa$ has no blocking sequence at residue $r$.
\begin{lemma}
	If $\MBS_r(\Aa)$ is empty, then $L = \lang{\Aa}$ is trivial at residue $r$.
\end{lemma}
\begin{proof}
	By \cref{cor:far-L-many-bs}, for every $\eps > 0$, every positional word of length $n \geq \max(n_0, C/\eps)$ with $n = r \pmod p$ that is $\eps$-far from $\timedlang{\Aa}$ contains a blocking sequence for $\Aa$ at residue $r$.
	As $\Aa$ has no blocking sequence at residue $r$, no such word is $\eps$-far from $L$, hence $L$ is trivial at residue $r$.
\end{proof}

To prove the converse property, we need the following extension of Kleene's Lemma for languages of SCC-paths: for large enough~$\ell$, whether $\lang{\SCCpath}$ contains a word of length~$\ell$ only depends on the value of~$\ell$ modulo~$p$ ($p$ is the $\lcm$ of all the lengths of the simple cycles in $\Aa$).
\begin{lemma}
    Let $\SCCpath = P_0 \xrightarrow{a_1} \cdots P_k$ be an SCC-path.
    There exists a constant $B $ such that, for all $\ell \ge B$, if there is a word $\mu$ of length~$\ell$ in $\lang{\SCCpath}$,
    then there exists a word $\mu'$ of length  $\ell - p$  and a word $\mu''$ of length  $\ell + p$ in $\lang{\SCCpath}$.
\end{lemma}
\begin{proof}
    Recall the definition of $\lang{\SCCpath}$ (\cref{def:lang-of-sccpath}):
    \[\lang{\SCCpath} = L_0 a_1 L_1 a_2 \cdots L_k, \text{ where } L_i = \lang{P_i} \text { for } i =0,\ldots,k.\]
    It follows that a word $\mu\in \lang{\SCCpath}$ can be written as $\mu= \mu_1a_1\mu_2\ldots \mu_k$ with $\mu_i\in L_i$.
    Each $L_i$ is recognized by a strongly connected automaton $\Aa_i$ with at most $p|\Aa|$ states.
    Let $B = (3(p|\Aa|)^2 + p) |\Aa|$.
    If the length $\ell$ of $\mu$ exceeds $B$, then the run of $\mu$ remains in one of the $\Aa_i$'s for $\ell_0 \ge 3(p|\Aa|)^2 + p$ steps.
    Thus there is a run of length $\ell_0$ through $\Aa_i$, and by \cref{fact:periodicity}, there is also a run of length $\ell_0-p \geq 3(p|\Aa|)^2$ through $\Aa_i$.
    We replace the former run through $\Aa_i$ by this one to obtain $\mu'$.

    To obtain $\mu''$, consider any simple cycle in the run of $\mu$ in $\Aa$, and let $m$ denote the length of this cycle.
    By definition of $p$, $m$ divides $p$. Iterating this cycle $p/m$ times yields a word $\mu''$ of length $\ell+p$ that is in $\lang{\SCCpath}$.
\end{proof}
\begin{corollary}\label{lem:fin-dist-iff-modulo}
    Let $\SCCpath$ be an SCC path.
    For large enough $\ell$, whether there is a word of length $\ell$ in $\lang{\SCCpath}$ only depends on the value of $\ell \pmod{p}$.
\end{corollary}

To finish our characterization of trivial languages, we show that if $\MBS_r(\Aa)$ is not empty, then $L = \lang{\Aa}$ is non-trivial at residue $r$.
\begin{lemma}
	\label{lem:charac-trivial}
	Let $\Aa$ be a trim NFA and let $r \in R(\lang{\Aa})$.
	If $\Aa$ admits a blocking sequence at residue $r$,
	then the language of $\Aa$ is non-trivial at residue $r$.
\end{lemma}
\begin{proof}
	Let $\sigma$ be a blocking sequence for $\Aa$ at residue $r$, and let $K$ be the maximal number of portals of an SCC-path.
	The sequence $\sigma^{(K)}$, obtained by repeating each term of $\sigma$ $K$ times, is strongly blocking for every SCC-path of $\Pi_r$ (see the beginning of \cref{sec:general-lb}); we write $\sigma^{(K)} = (\mu_1, \ldots, \mu_k)$, and let $C$ be the maximum length of a $\mu_i$.
	As $r \in R(\lang{\Aa})$, there is an SCC-path $\SCCpath \in \Pi_r$, whose language is infinite.
	By \cref{lem:fin-dist-iff-modulo}, there is a constant $B'$ such that for every $\ell \geq B'$ with $\ell = r \pmod{p}$, $\lang{\SCCpath}$ contains a word of length $\ell$.

	For all $i = 1,\ldots, k$, let $\nu_i$ be a shortest positional word of the form $\timedword{0}{v_i}$ such that $|v_i| = 0 \pmod{p}$ and $\mu_i$ is a factor of $\nu_i$. By minimality, $|\nu_i|$ is at most $C+2p$.
	Let $w_0$ be a word of length $\ell_0 \geq B'$ with $\ell_0 = r \pmod{p}$.
	For any integer $N \in \NN$, let $w_N = \nu_1^N \cdots \nu_k^N \timedword{0}{w_0}$.

	As $w_N$ has length $n_N = \ell_0 \pmod{p}$ and $n_N \geq B'$, the language $\lang{\Aa}$ contains a word of length $n_N$.
	On the other hand, $w_N$ contains $N$ disjoint occurrences of $\sigma^{(K)}$: for every $c \le N$, the $c$-th copies of $\nu_1, \ldots, \nu_k$ contain an occurrence of $\sigma^{(K)}$.
	A positional word obtained from $w_N$ by fewer than $N$ substitutions still contains one of these occurrences, so it is not in $\lang{\SCCpath'}$ for any $\SCCpath' \in \Pi_r$ by \cref{lem:strongly-blocking-implies-not-in-path}, and it is not in $\timedlang{\Aa}$ if $n_N \geq n_0$.
	Hence $d(w_N, \lang{\Aa}) \geq N$ for $N$ large enough.
	Furthermore, the length of $w_N$ is at most $\ell_0 + N k (C+2p)$.
	Therefore, if we let $\eps_0 = \frac{1}{k(C+2p)+\ell_0}$, then we have $\eps_0 |w_N| \leq N \leq d(w_N, \lang{\Aa})$, i.e. $w_N$ is $\eps_0$-far from~$L$ for all large enough $N$.
\end{proof}

\section{Hardness of classifying}
\label{sec:complexity}
In the previous sections, we have shown that testing some regular languages (\emph{easy} ones) requires fewer queries than testing others (\emph{hard} ones).
In this section, we investigate the computational complexity of checking which class of the trichotomy the language of a given automaton belongs to.
Note that the automaton is fixed when testing, so the class of its language can be computed once and for all before testing any word.
We formalize this question as the following decision problems:

\begin{problem}[Triviality problem]
  Given a finite automaton $\Aa$, is $\lang{\Aa}$ trivial?
\end{problem}
\begin{problem}[Easiness problem]
  Given a finite automaton $\Aa$, is $\lang{\Aa}$ easy?
\end{problem}
\begin{problem}[Hardness problem]
  Given a finite automaton $\Aa$, is $\lang{\Aa}$ hard?
\end{problem}

In these problems, the automaton $\Aa$ is the input and is no longer fixed.
We show that our combinatorial characterization based on minimal blocking sequences is effective, in the sense that all three problems are decidable. However, it does not lead to efficient algorithms, as all three problems are \PSPACE-complete.

\begin{theorem}\label{thm:PSPACE-main}
  The triviality, easiness and hardness problems are all \PSPACE-complete, even for strongly connected NFAs.
\end{theorem}

In \cref{sec:ppace-ub} we show the \PSPACE upper bounds on the hardness and triviality problems (Propositions~\ref{prop:hard-to-test-PSPACE} and~\ref{prop:trivial-PSPACE}). The upper bound on the easiness problem follows immediately, as the three properties form a trichotomy.

In \cref{sec:pspace-hard}, we show that all three problems are \PSPACE-hard (\cref{lemma:trivial-complexity} and \cref{cor:hardness-easy}).

\subsection{A \texorpdfstring{$\PSPACE$}{PSPACE} upper-bound}\label{sec:ppace-ub}

\subsubsection{Testing hardness}

A naive algorithm to check hardness of a language $\lang{\Aa}$ would be to construct an automaton recognising blocking sequences of $\lang{\Aa}$ (exponential in $\Aa$), and use it to get an automaton recognising the minimal ones (which requires complementation and could yield another exponential blow-up). This would a priori not give a \PSPACE algorithm, since we obtain a doubly-exponential state space.
We avoid this by giving a slightly different characterisation of automata with hard languages, which yields a recursive \PSPACE algorithm.
We start with a technical lemma that lets us merge several sequences that preserve a specific portal into a single one that still preserves it.

\begin{lemma}
  \label{lem:uniformisation-blocking-seq}
  Let $\SCCpath = P_0 \xrightarrow{a_1} \cdots  P_\ell$ be an SCC-path, $i$ an index, $\Pi$ a finite family of SCC-paths (possibly with repetitions) and $(\sigma_{\SCCpath'})_{\SCCpath' \in \Pi}$ a family of sequences of positional words such that $\lefteffect{\sigma_{\SCCpath'}}{\SCCpath} < i$ for all $\SCCpath'$.

  There exists a sequence of positional words $\sigma$ such that:
  \begin{itemize}
    \item $\lefteffect{\sigma}{\pi} < i$
    \item $\lefteffect{\sigma_{\SCCpath'}}{\SCCpath'} \leq \lefteffect{\sigma}{\SCCpath'}$ for all $\SCCpath' \in \Pi$.
  \end{itemize}
\end{lemma}

\begin{proof}
  We prove this by induction on the sum of the lengths of the elements of~$\Pi$.
  If~$\Pi$ is empty or contains only empty sequences, then we can set~$\sigma$ as the empty sequence.

  If not, let $\SCCpath^*$ be, among the elements of $\Pi$ whose sequence is non-empty, one such that the first term $\nu_1$ of $\sigma_{\SCCpath^*} = (\nu_1, \ldots, \nu_k)$ has the least left effect on $\SCCpath$ among all SCC-paths in $\Pi$; write $\SCCpath^* = P_0' \xrightarrow{a_1'} \cdots  P_{\ell'}'$. We consider the effect of $\nu_1$ (as a single-element sequence) on $\SCCpath^*$ and $\SCCpath$: let~$j = \lefteffect{\nu_1}{\SCCpath^*}$ and~$r = \lefteffect{\nu_1}{\SCCpath}$.

  Next, we build a set $\Pi'$ of SCC-paths as follows.
  Let $\overline{\SCCpath}$ denote the part of $\SCCpath^*$ that survives $\nu_1$, if any, i.e. $\overline{\SCCpath} = P_{j+1}' \xrightarrow{a_{j+1}'} \cdots P_{\ell'}'$.
  We define $\Pi' = \Pi \setminus \set{\SCCpath^*} \cup \set{\overline{\SCCpath}}$ if $j<\ell'$ and $\Pi' = \Pi \setminus \set{\SCCpath^*}$ otherwise.
  In the first case the sequence associated with $\overline{\SCCpath}$ is $\sigma_{\overline{\SCCpath}} = (\nu_2, \ldots, \nu_k)$.

  We now wish to apply the induction hypothesis to the set $\Pi'$ and the part of $\SCCpath$ that survives $\nu_1$, i.e. on $\tilde{\SCCpath} = P_{r+1}\xrightarrow{a_{r+1}} \ldots \rightarrow P_\ell$, with a left effect of $i-r-1$.
  By construction, the sum of the lengths of the elements in~$\Pi'$ is smaller than that of~$\Pi$.
  The following claim shows that, for any~$\SCCpath'$ in~$\Pi'$, the left effect of~$\sigma_{\SCCpath'}$ on~$\tilde{\SCCpath}$ is less than~$i-r-1$.

  \begin{claim}
    For all $\SCCpath' \in \Pi'$, we have $\lefteffect{\sigma_{\SCCpath'}}{\tilde{\SCCpath}} < i-r-1$.
  \end{claim}

  \begin{claimproof}
    If $\SCCpath' = \overline{\SCCpath}$, then $\sigma_{\SCCpath'} = (\nu_2, \ldots, \nu_k)$; as $\nu_1$ has left effect $r$ on $\SCCpath$, by the greedy computation of left effects, $\lefteffect{\sigma_{\SCCpath^*}}{\SCCpath} = r+1+\lefteffect{(\nu_2, \ldots, \nu_k)}{\tilde{\SCCpath}}$, and this is less than $i$ by assumption.
    Now let $\pi' \in \Pi\setminus \set{\SCCpath^*}$, and let $\sigma_{\SCCpath'} = (\nu_1', \ldots, \nu_m')$; if this sequence is empty, its left effect on $\tilde{\SCCpath}$ is $-1 < i-r-1$, as $r < i$.
    Let $z = \lefteffect{\nu_1'}{\SCCpath}$.
    Since the first term of $\sigma_{\SCCpath^*}$ was the one with the least left effect on $\SCCpath$, we have $z \ge r$.
    By definition of the left effect, $\nu_1'$ is blocking for $P_0, \ldots, P_z$, and not for $P_{z+1}$ if $z < \ell$.
    As $z \geq r$ and the portals of $\tilde{\SCCpath}$ are $P_{r+1}, \ldots, P_\ell$, this gives $\lefteffect{\nu_1'}{\tilde{\SCCpath}} = z-r-1$, which is the index of $P_z$ in $\tilde{\SCCpath}$.
    By the greedy computation of left effects, the sequence $(\nu_2', \ldots, \nu_m')$ is then applied from the portal $P_{z+1}$ both in~$\SCCpath$ and in~$\tilde{\SCCpath}$, so $\lefteffect{\sigma_{\SCCpath'}}{\tilde{\SCCpath}} = \lefteffect{\sigma_{\SCCpath'}}{\SCCpath} - r - 1$, which is less than $i-r-1$ by assumption.
  \end{claimproof}

  By induction hypothesis, we obtain a sequence $\tilde{\sigma}$ such that
  \begin{itemize}
    \item $\lefteffect{\tilde{\sigma}}{\tilde{\SCCpath}} < i-r-1$
    \item $\lefteffect{\sigma_{\SCCpath'}}{\SCCpath'} \leq \lefteffect{\tilde{\sigma}}{\SCCpath'}$ for all $\SCCpath' \in \Pi'$.
  \end{itemize}
  Let $\sigma$ be the sequence obtained by prepending $\nu_1$ to $\tilde{\sigma}$.
  As $\tilde{\SCCpath}$ is the part of $\SCCpath$ that survives $\nu_1$, we have $\lefteffect{\sigma}{\SCCpath} = r+1+\lefteffect{\tilde\sigma}{\tilde\SCCpath} < i$.
  For $\SCCpath' \in \Pi \setminus \set{\SCCpath^*}$, prepending $\nu_1$ cannot decrease the left effect of a sequence, hence $\lefteffect{\sigma_{\SCCpath'}}{\SCCpath'} \le \lefteffect{\tilde\sigma}{\SCCpath'} \le \lefteffect{\sigma}{\SCCpath'}$.
  Finally, if $j = \ell'$, then $\lefteffect{\sigma}{\SCCpath^*} \geq j = \ell'$ is maximal; otherwise, by the greedy computation of left effects,
  \[\lefteffect{\sigma}{\SCCpath^*} = j+1+\lefteffect{\tilde\sigma}{\overline\SCCpath} \geq j+1+\lefteffect{(\nu_2, \ldots, \nu_k)}{\overline\SCCpath} = \lefteffect{\sigma_{\SCCpath^*}}{\SCCpath^*}.\]
  Hence $\sigma$ satisfies both conditions of the lemma.
\end{proof}

The following lemma is a variant of the characterisation of hard languages given by \cref{lem:hard-aut-to-hard-portal,lem:P1-P3-then-hard}.
The difference is that the sequences $\sigma_l$ and $\sigma_r$ of properties P1, P2 and P3 are the same for all accepting SCC-paths, while in the following lemma they may depend on the SCC-path, and we show that such sequences can always be merged into sequences that do not depend on it.
This variant is easier to check in polynomial space, as the sequences can be guessed independently for each SCC-path.

\begin{lemma}\label{lem:charac-hard-PSPACE}
  The language of an automaton $\Aa$ is hard if and only if there exist a residue $r \in R(\lang{\Aa})$ and an SCC-path $\SCCpath \in \Pi_r$ containing a portal $P_*$ such that:
  \begin{itemize}
    \item $P_*$ has infinitely many minimal blocking factors.
    \item For any SCC-path $\SCCpath' \in \Pi_r$ there exist sequences $\sigma_{l, \SCCpath'}, \sigma_{r, \SCCpath'}$ such that:
      \begin{itemize}
        \item $P_*$ survives $(\sigma_{l, \SCCpath'}, \sigma_{r, \SCCpath'})$ in $\SCCpath$
        \item All portals surviving $(\sigma_{l, \SCCpath'}, \sigma_{r, \SCCpath'})$ in $\SCCpath'$ are  $\succeq P_*$
      \end{itemize}
  \end{itemize}
\end{lemma}

\begin{proof}
  For the left-to-right direction, by \cref{thm:general} and \cref{coro:fin-bs-then-easy}, there is a residue $r \in R(\lang{\Aa})$ such that $\MBS_r(\Aa)$ is infinite.
  Let $P$, $\sigma_l$ and $\sigma_r$ be given by \cref{lem:hard-aut-to-hard-portal} at residue $r$.
  By P3, there is an SCC-path $\SCCpath \in \Pi_r$ and a portal $P_*$ of $\SCCpath$ that survives $(\sigma_l, \sigma_r)$ in $\SCCpath$ and is equivalent to $P$.
  As $P_*$ and $P$ have the same blocking factors, $P_*$ has infinitely many minimal blocking factors by P2, and all the surviving portals of all SCC-paths of $\Pi_r$ are $\succeq P_*$ by P3.
  Hence we can take $\sigma_{l, \SCCpath'} = \sigma_l$ and $\sigma_{r, \SCCpath'} = \sigma_r$ for every $\SCCpath'$.

  Conversely, suppose that $\SCCpath$ and $P_*$ satisfy the conditions of the lemma.
  By \cref{lem:P1-P3-then-hard}, it suffices to construct two sequences $\sigma_l, \sigma_r$ such that properties P1 and P3 are satisfied.

  Let $\Pi = \Pi_r$.
  Consider families of sequences $(\sigma_{l, \SCCpath'})_{\SCCpath' \in \Pi}$ and $(\sigma_{r, \SCCpath'})_{\SCCpath' \in \Pi}$ as given by the hypothesis of the lemma, i.e., such that for all $\SCCpath' \in \Pi$:
  \begin{itemize}
    \item $P_*$ survives $(\sigma_{l, \SCCpath'}, \sigma_{r, \SCCpath'})$ in $\SCCpath$
    \item All portals surviving $(\sigma_{l, \SCCpath'}, \sigma_{r, \SCCpath'})$ in $\SCCpath'$ are $\succeq P_*$
  \end{itemize}

  Let $i$ be the index of $P_*$ in $\SCCpath$.
  By \cref{lem:uniformisation-blocking-seq} we can build a sequence $\sigma_l$ such that
  \begin{itemize}
    \item $\lefteffect{\sigma_l}{\SCCpath} < i$, and
    \item $\lefteffect{\sigma_{l, \SCCpath'}}{\SCCpath'} \leq \lefteffect{\sigma_l}{\SCCpath'}$ for all $\SCCpath' \in \Pi$.
  \end{itemize}

  Using a symmetric argument, we build a sequence $\sigma_r$ such that
  \begin{itemize}
    \item $i< \righteffect{\sigma_r}{\SCCpath}$, and
    \item $\righteffect{\sigma_{r, \SCCpath'}}{\SCCpath'} \geq \righteffect{\sigma_r}{\SCCpath'}$ for all $\SCCpath' \in \Pi$.
  \end{itemize}

  As a consequence, for all SCC-paths $\SCCpath' \in \Pi$, the portals that survive $(\sigma_l, \sigma_r)$ in $\SCCpath'$ also survive $(\sigma_{l, \SCCpath'}, \sigma_{r,\SCCpath'})$, hence they are $\succeq P_*$.
  Moreover, $P_*$ survives $(\sigma_l, \sigma_r)$ in $\SCCpath$; in particular $\sigma_l\sigma_r$ is not blocking for $\SCCpath$.

  Therefore $P_*$ and $(\sigma_l, \sigma_r)$ satisfy properties P1 and P3 at residue $r$, and P2 holds by assumption, so the result follows from \cref{lem:P1-P3-then-hard}.
\end{proof}

Next, we establish that the items listed in the previous lemma can all be checked in polynomial space in $|\Aa|$.

\begin{lemma}
  \label{lem:aut-block-fact}
  Given a portal $P =\portal{s}{x}{t}{y}$, we can check whether it has infinitely many minimal blocking factors in space polynomial in $|\Aa|$.
\end{lemma}

\begin{proof}
  We start by building an NFA recognising the set of minimal blocking factors of $P$, similarly to Lemma~\ref{lemma:blocking-regular}.
  The only difference is that we need to keep track of the position modulo $p$ (the lcm of the simple cycle sizes, which may be exponential in the number of states), and not modulo the periodicity constant of the component (which is bounded by the size of the cycle and thus number of states).

  Given $i \in [0,p-1]$ and a state $q$ we can check whether there is a run from $s$ to $q$ of length $(i-x) (\mod p)$ in PSPACE: it suffices to guess a path from $s$ on the fly while keeping track of the residue modulo $p$. This requires only polynomial memory since we encode at most exponential numbers in binary. Savitch's theorem~\cite[Theorem 1]{SAVITCH1970177} lets us remove the non-determinism while remaining in PSPACE.
  Let $S_i$ be the set of states for which such a run exists.

  Similarly given $j \in [0,p-1]$ and a state $q$ we can check whether there is a run from $q$ to $t$ of length $(y-j) (\mod p)$ in PSPACE.
  Let $T_j$ be the set of states for which such a run exists.

  We build an NFA that records four things while reading a word $\timedword{i}{u}$:
  \begin{itemize}
    \item The number $|u|+i$ modulo $p$
    \item The set of states $q$ such that there is a run from $S_i$ to $q$ reading $u$
    \item The set of states $q$ such that there is a run from $S_{i+1}$ to $q$ reading $u$ minus its first letter
    \item The set of states $q$ such that there is a run from $S_i$ to $q$ reading $u$ minus its last letter
  \end{itemize}
  The NFA has an additional initial state $q_{init}$, since the two last parts are undefined for the empty word.
  These sets are straightforward to initialise and update as we read the positional word.
  A word is a minimal blocking factor if and only if after reading it, the current residue modulo $p$ is some $j$ such that the first set does not intersect $T_j$, the second one intersects $T_j$, and the third one intersects $T_{j-1}$.
  We define our set of final states accordingly.
  This automaton has exponential size, and it remains to check whether its language is infinite, which is the case if and only if it has a cycle reachable from the initial state and from which a final state is reachable.
  This can be checked by exploring the state space of the automaton in non-deterministic polynomial space (in $|\Aa|$) and applying Savitch's theorem.
\end{proof}

\begin{lemma}
  \label{lem:PSPACE-blocking-for-one}
  Given two SCC-paths $\SCCpath$ and $\SCCpath'$, one can check in \PSPACE whether there is a sequence $\sigma$ that is blocking for $\SCCpath$ and not for $\SCCpath'$.
\end{lemma}
\begin{proof}
  In this proof, we also consider the empty sequence of portals, written $\varepsilon$, for which every sequence is blocking.
  For a positional word $\mu$ and an SCC-path $\SCCpath = P_0 \xrightarrow{a_1} \cdots P_k$, let $\SCCpath_{> \mu}$ denote the suffix $P_{e+1} \xrightarrow{a_{e+2}} \cdots P_k$ of $\SCCpath$, where $e = \lefteffect{(\mu)}{\SCCpath}$ (it is $\varepsilon$ if $e = k$).
  By the greedy computation of left effects, for every sequence $\sigma'$, the sequence $(\mu)\sigma'$ is blocking for $\SCCpath$ if and only if $\sigma'$ is blocking for $\SCCpath_{>\mu}$.
  The algorithm relies on the following property.
  \begin{claim}
    There is a sequence $\sigma$ that is blocking for $\SCCpath$ and not for $\SCCpath'$ if and only if either:
    \begin{itemize}
      \item $\SCCpath = \varepsilon$ and $\SCCpath' \neq \varepsilon$, or
      \item there is a positional word $\mu$ that is blocking for the first portal of $\SCCpath$, and a sequence $\sigma'$ that is blocking for $\SCCpath_{>\mu}$ and not for $\SCCpath'_{>\mu}$.
    \end{itemize}
  \end{claim}
  \begin{claimproof}
    The right-to-left direction is clear: take $\sigma$ the empty sequence in the first case, and $\sigma = (\mu) \sigma'$ in the second case.
    For the left-to-right direction, consider a sequence $\sigma$ that is blocking for $\SCCpath$ and not for $\SCCpath'$, of minimal length.
    If it is empty, then $\SCCpath = \varepsilon$ and $\SCCpath' \neq \varepsilon$.
    Otherwise, let $\sigma = (\mu) \sigma'$.
    If $\mu$ is not blocking for the first portal of $\SCCpath$, then $\SCCpath_{>\mu} = \SCCpath$, so $\sigma'$ is blocking for $\SCCpath$, and it is not blocking for $\SCCpath'$ as $\SCCpath'_{>\mu}$ is a suffix of $\SCCpath'$; this contradicts the minimality of $\sigma$.
    Hence $\mu$ is blocking for the first portal of $\SCCpath$, and $\sigma'$ is blocking for $\SCCpath_{>\mu}$ and not for $\SCCpath'_{>\mu}$.
  \end{claimproof}

  The claim gives a recursive algorithm.
  If $\SCCpath = \varepsilon$, we answer yes if $\SCCpath' \neq \varepsilon$ and no otherwise.
  If $\SCCpath' = \varepsilon$, we answer no.
  Otherwise, we guess the suffixes $\SCCpath_1$ of $\SCCpath$ and $\SCCpath'_1$ of $\SCCpath'$, with $\SCCpath_1$ shorter than $\SCCpath$, and we check that there is a positional word $\mu$ such that $\SCCpath_{>\mu} = \SCCpath_1$ and $\SCCpath'_{>\mu} = \SCCpath'_1$; then we make a recursive call on $\SCCpath_1$ and $\SCCpath'_1$.
  The existence of such a $\mu$ means that $\mu$ is blocking for a given set of portals and not blocking for at most two given portals, which can be checked in polynomial space by exploring on the fly the product of the automata built in the proof of \cref{lem:aut-block-fact} for these portals.
  As the depth of the recursive calls is at most the length of $\SCCpath$, which is bounded by $|\Aa|$, this algorithm runs in polynomial space.
\end{proof}

\begin{proposition}
  \label{prop:hard-to-test-PSPACE}
  The hardness problem is in \PSPACE.
\end{proposition}
\begin{proof}
  We use the following algorithm to check whether the characterization of \cref{lem:charac-hard-PSPACE} holds.
  \begin{enumerate}
    \item First, we nondeterministically guess a residue $r$, an SCC-path $\SCCpath = P_0 \xrightarrow{a_1} \cdots  P_k \in \Pi_r$ and an index $i$. Checking that $\SCCpath$ is accepting, has residue $r$ and has an infinite language (that is, one of its portals is recurrent) can be done in polynomial space.
    \item Using \cref{lem:aut-block-fact}, we check that $P_i$ has infinitely many minimal blocking factors.

    \item For each SCC-path $\SCCpath' = P_0' \xrightarrow{a'_1} \cdots  P_\ell' \in \Pi_r$, we guess indices $-1 \le j_l < j_r \le \ell+1$, and check that every portal $P_j'$ with $j_l < j < j_r$ satisfies $P'_j \succeq P_i$, that is, that there is no positional word that is blocking for $P_i$ and not for $P'_j$; this can be checked in polynomial space as in \cref{lem:PSPACE-blocking-for-one}.

    \item Then, we use \cref{lem:PSPACE-blocking-for-one} to check that there is a sequence $\sigma_l$ that is blocking for $P_0' \xrightarrow{a'_1} \cdots  P_{j_l}'$ and not $P_0 \xrightarrow{a_1} \cdots P_i$.
      Symmetrically, we check that there is a sequence $\sigma_r$ that is blocking for $P_{j_r}' \xrightarrow{a'_{j_r+1}} \cdots  P_\ell'$ and not $P_i \xrightarrow{a_{i+1}} \cdots  P_k$.
  \end{enumerate}

  If all those tests succeed, we answer ``yes'', otherwise we answer ``no''.
  This algorithm is correct and complete by \cref{lem:charac-hard-PSPACE}.
\end{proof}

\subsubsection{Testing triviality}

We show the \PSPACE upper bound on the complexity of checking if a language is trivial.
It is based on the characterisation of trivial languages given by \cref{lem:charac-trivial}, and uses the following result.

\begin{lemma}
  \label{lem:exists-blocking-fact}
  Given a portal $P$, we can check whether it has a blocking factor in space polynomial in $|\Aa|$.
\end{lemma}
\begin{proof}
  We proceed as in the proof of \cref{lem:aut-block-fact}, except that we only need to check whether some final state is reachable from the initial state.
\end{proof}

\begin{proposition}\label{prop:trivial-PSPACE}
  The triviality problem is in \PSPACE.
\end{proposition}
\begin{proof}
  Recall that $\lang{\Aa}$ is trivial if and only if, for every $r \in R(\lang{\Aa})$, $\Aa$ has no blocking sequence at residue $r$.
  \begin{claim}
    Let $r \in R(\lang{\Aa})$. There is an SCC-path $\SCCpath \in \Pi_r$ that contains a portal $P$ with no blocking factors if and only if $\Aa$ has no blocking sequence at residue $r$.
  \end{claim}
  \begin{claimproof}
    A blocking sequence of $\Aa$ must be blocking for each SCC-path $\SCCpath$, and therefore contain a blocking factor for each of its portals.
    Conversely, if every portal in every accepting SCC-path has a blocking factor, then we can build a blocking sequence for each SCC-path $\pi \in \Pi_r$, and concatenating these sequences in any order gives a blocking sequence for $\Aa$ at residue $r$.
  \end{claimproof}
  Therefore, it suffices to enumerate all residues $r$ and all SCC-paths $\SCCpath \in \Pi_r$, and check, using \cref{lem:exists-blocking-fact}, that for each $r$ such that $\Pi_r$ is non-empty, some SCC-path of $\Pi_r$ contains a portal without blocking factors.
\end{proof}

\subsection{Hardness of classifying automata}\label{sec:pspace-hard}

We prove hardness of the triviality, easiness and hardness problems, completing the proof of their \PSPACE-completeness.
We reduce from the universality problem for NFAs, which is well-known to be \PSPACE-complete (see e.g.~\cite[Theorem 10.14]{aho1974design}).

\begin{lemma}\label{lemma:trivial-complexity}
  The triviality and hardness problems are \PSPACE-hard.
\end{lemma}
\begin{proof}
  Consider an instance of the universality problem, that is, an NFA $\Aa = (Q, \Sigma, \delta, q_0, F)$ on an alphabet $\Sigma$.
  We first check, in polynomial time, that $\Aa$ accepts all words of length less than 2; if it does not, then $\Aa$ is not universal, and we simply output an arbitrary hard automaton.
  From now on, $\Aa$ accepts all words of length less than 2, in particular the empty word, so $q_0$ is final. Hence removing the unreachable and the non-co-reachable states keeps $q_0$ and does not change the language, and we may assume that $\Aa$ is trim.
  Let $\#$ and $!$ be two letters that are not in $\Sigma$.
  We apply the following transformations to $\Aa$:
  \begin{itemize}
    \item add a transition labeled by $!$ from every final state to the initial state~$q_0$
    \item add a self-loop labeled by $\#$ to each state.
  \end{itemize}

  We call the resulting automaton $\Bb = (Q, \Sigma\cup\set{!,\#}, \delta', q_0, F)$.
  Note that $\Bb$ is strongly connected.
  Indeed, let $q, q' \in Q$. As $\Aa$ is trim, there exists $q_f \in F$ that is reachable from~$q$, and~$q'$ is reachable from the initial state~$q_0$; since there is a $!$ transition from~$q_f$ to~$q_0$, the state~$q'$ is reachable from~$q$.

  To complete the proof, we show that $\lang{\Bb}$ is universal (hence trivial) when $\Aa$ is universal, and that $\MBF(\Bb)$ is infinite (hence $\lang{\Bb}$ is hard by \cref{thm:scc}) otherwise.

  The language recognized by $\Bb$ is
  \[\lang{\Bb} = \set{u_1!u_2!\cdots !u_n \mid \forall i, u_i \in (\Sigma\cup\set{\#})^* ~\land~ \pi_{\Sigma}(u_i) \in \lang{\Aa}},\]
  where $\pi_{\Sigma}(u)$ is the word in $\Sigma^*$ obtained by removing all letters not in $\Sigma$ from $u$.

  \begin{claim}
    If $\Aa$ is universal, then $\Bb$ is also universal.
  \end{claim}
  \begin{claimproof}
    Indeed, any word $u \in (\Sigma \cup \set{\#,!})^*$ can be uniquely decomposed into $u = u_1!u_2!\cdots !u_n$ where each $u_i$ does not contain the letter ``$!$''.
    As $\#$ is neutral on $\Bb$, $\delta'(q_0,u_i)$ is equal to $\delta(q_0, \pi_{\Sigma}(u_i))$ for every $i$.
    Since $\Aa$ is universal, each of the $\delta'(q_0, u_i)$ contains a final state, hence $\delta'(q_0, u_i!) = \set{q_0}$.
    Therefore, the set $\delta'(q_0, u)$ is equal to $\delta'(q_0, u_n)$, which contains a final state, and  $u$ is in $\lang{\Bb}$, which shows that $\Bb$ is universal.
  \end{claimproof}
  In particular, if $\Aa$ is universal, then $\lang{\Bb}$ is trivial.

  Now we show that a word $w \in \Sigma^*$ not in $\lang{\Aa}$ induces infinitely many minimal blocking factors for $\Bb$.
  Consider such a $w$ of minimal size; as we assumed that $\Aa$ accepts all words of size less than 2, we have $|w| \geq 2$, so we can write $w=uv$ where $u, v$ are words of length at least 1.
  Since $\#$ labels a loop on every state of $\Bb$ and $!$ always leads to $q_0$, a word $z \in (\Sigma\cup\set{\#})^*$ is a factor of a word of $\lang{\Bb}$ if and only if $\pi_\Sigma(z)$ is a factor of a word of $\lang{\Aa}$; the word $!z$ (resp.\ $z!$, resp.\ $!z!$) is a factor of a word of $\lang{\Bb}$ if and only if $\pi_\Sigma(z)$ is a prefix (resp.\ a suffix, resp.\ an element) of $\lang{\Aa}$. For the converse implications, a word of $\lang{\Aa}$ is needed before (resp.\ after) the letter~$!$, and we can use the empty word, which belongs to $\lang{\Aa}$.
  Moreover, all the words shorter than $w$ belong to $\lang{\Aa}$, by minimality of $w$.
  Let $n \in \NN$; we define a minimal blocking factor $z_n$ for $\Bb$ as follows.
  \begin{itemize}
    \item If $w$ is not a factor of a word of $\lang{\Aa}$, let $z_n = u\#^nv$, which is blocking. Deleting its first or its last letter gives a word whose projection on $\Sigma$ is shorter than $w$, hence in $\lang{\Aa}$, so these words are not blocking.
    \item Otherwise, if $w$ is not a prefix of a word of $\lang{\Aa}$, let $z_n = !u\#^nv$, which is blocking. Deleting its first letter gives $u\#^nv$, which is not blocking as $w$ is a factor of a word of $\lang{\Aa}$; deleting its last letter gives a word $!z'$ where $\pi_\Sigma(z')$ is shorter than $w$, hence in $\lang{\Aa}$, so it is not blocking.
    \item Otherwise, if $w$ is not a suffix of a word of $\lang{\Aa}$, let $z_n = u\#^nv!$; this case is symmetric to the previous one.
    \item Otherwise, let $z_n = !u\#^nv!$, which is blocking as $w \notin \lang{\Aa}$. Deleting its first (resp.\ last) letter gives a word that is not blocking, as $w$ is a suffix (resp.\ a prefix) of a word of $\lang{\Aa}$.
  \end{itemize}
  In all cases, as every factor of a non-blocking word is non-blocking, $z_n$ is a minimal blocking factor.
  As the words $z_n$ are pairwise distinct, $\Bb$ has infinitely many minimal blocking factors, and is thus hard to test by \cref{thm:scc}.

  Hence $\Aa$ is universal if and only if $\Bb$ is trivial to test, and  $\Aa$ is \emph{not} universal if and only if $\Bb$ is hard to test.
  This shows the \PSPACE-hardness of the triviality problem, and, since \PSPACE is closed under complement, of the hardness problem.
\end{proof}

The above proof can be extended to show the \PSPACE-hardness of the easiness problem.
\begin{corollary}\label{cor:hardness-easy}
  The easiness problem is \PSPACE-hard.
\end{corollary}
\begin{proof}
  We proceed as in the proof of \cref{lemma:trivial-complexity}: given an automaton $\Aa$ over an alphabet $\Sigma$,
  we build the automaton $\Bb$ over the alphabet $\Sigma\cup\set{!,\#}$.
  Let $\flat$ denote a new letter not in $\Sigma \cup \set{\#, !}$ and consider the automaton $\Bb'$ equal to $\Bb$ but taken over the alphabet $\Sigma \cup \set{\#, !, \flat}$; it is strongly connected.
  If $\Aa$ is universal, then $\lang{\Bb'}$ is the set of words without $\flat$.
  This language is easy. Indeed, a word is $\eps$-far from it if and only if it contains at least $\eps n$ letters $\flat$, so sampling $\cO(1/\eps)$ random positions gives a tester; moreover, the language is non-trivial, so it requires $\Omega(1/\eps)$ queries by the result of Alon et al.~\cite{alon2001regular}.
  If $\Aa$ is not universal, then $\MBF(\Bb') \supseteq \MBF(\Bb)$ is infinite, so $\lang{\Bb'}$ is hard by \cref{thm:scc}.
  Hence $\Aa$ is universal if and only if $\lang{\Bb'}$ is easy, and the easiness problem is also \PSPACE-hard.
\end{proof}

This concludes the proof of \cref{thm:PSPACE-main}.

\section{Conclusion}
We presented an effective classification of regular languages in three classes, each associated with an optimal query complexity for property testing.
We thus close a line of research aiming to determine the optimal complexity of regular languages.
All our results concern the Hamming distance and non-adaptive testers with perfect completeness.
We conjecture that the same three classes, with the same characterisation, describe the query complexity of testing under the edit distance (see \cref{sec:intro-overview}).
For adaptive testers, the classification is different by \cref{prop:adaptive}, and characterising their query complexity remains open.
Finally, since we represent regular languages by non-deterministic automata and our \PSPACE-hardness results rely on non-determinism, the complexity of classifying languages represented by \emph{deterministic} automata is also open.

\bibliography{biblio}

\end{document}